\documentclass[oneside,english,reqno]{amsart}
\usepackage[T1]{fontenc}
\usepackage[latin9]{inputenc}
\usepackage{textcomp}
\usepackage{amsthm}
\usepackage{bbm}
\usepackage{amssymb}
\usepackage{graphicx}

\makeatletter
\numberwithin{equation}{section}
  \theoremstyle{plain}
  \newtheorem*{thm*}{\protect\theoremname}
\theoremstyle{plain}
\newtheorem{thm}{\protect\theoremname}[section]
  \theoremstyle{plain}
  \newtheorem{cor}[thm]{\protect\corollaryname}
  \theoremstyle{remark}
  \newtheorem{rem}[thm]{\protect\remarkname}
  \theoremstyle{definition}
  \newtheorem{defn}[thm]{\protect\definitionname}
  \theoremstyle{plain}
  \newtheorem{prop}[thm]{\protect\propositionname}
  \theoremstyle{plain}
  \newtheorem{lem}[thm]{\protect\lemmaname}
  \theoremstyle{plain}
  
  \theoremstyle{plain}
  \newtheorem*{prop*}{\protect\propositionname}
  \theoremstyle{remark}
  \newtheorem*{rem*}{\protect\remarkname}

\makeatother

\usepackage{babel}
  \providecommand{\corollaryname}{Corollary}
  \providecommand{\definitionname}{Definition}
  \providecommand{\lemmaname}{Lemma}
  \providecommand{\propositionname}{Proposition}
  \providecommand{\remarkname}{Remark}
  \providecommand{\theoremname}{Theorem}
\providecommand{\theoremname}{Theorem}

\newcommand{\C}{\mathbb{C}}
\newcommand{\E}{\mathbb{E}}
\newcommand{\R}{\mathbb{R}}
\newcommand{\bH}{\mathbb{bH}}

\newcommand{\wind}{\mathrm{wind}}
\newcommand{\Ak}{a_1,\dots,a_n}
\begin{document}

\title[Conformal Invariance of Spin Correlations in the Ising Model]{Conformal Invariance of Spin Correlations in~the~Planar Ising Model}

\author[Dmitry Chelkak]{Dmitry Chelkak$^\mathrm{a,c}$}
\author[Cl\'ement Hongler]{Cl\'ement Hongler$^\mathrm{b}$}
\author[Konstantin Izyurov]{Konstantin Izyurov$^\mathrm{c}$}

\date{\today}

\thanks{\textsc{${}^\mathrm{A}$ St.Petersburg Department of Steklov Mathematical Institute.
Fontanka~27, 191023 St.Petersburg, Russia,} {\it E-mail address:} \texttt{dchelkak@pdmi.ras.ru}}

\thanks{\textsc{${}^\mathrm{B}$ Department of Mathematics, Columbia University. 2990 Broadway, New York, NY 10027, USA,} {\it E-mail address:} \texttt{hongler@math.columbia.edu}}

\thanks{\textsc{${}^\mathrm{C}$ Chebyshev Laboratory, Department of Mathematics and
Mechanics, Saint-Petersburg State University, 14th Line, 29b, 199178 Saint-Petersburg,
Russia. {\it \mbox{Current} address:} Department of Mathematics and Statistics, University of Helsinki, P.O. box 68, 00014 Helsinki, Finland,} {\it E-mail address:} \texttt{konstantin.izyurov@helsinki.fi}}

\begin{abstract}
We rigorously prove the existence and the conformal invariance of scaling limits of the magnetization and multi-point spin correlations in the critical Ising model on arbitrary simply connected planar domains. This solves a number of conjectures coming from the physical and the mathematical literature. The proof relies on convergence results for discrete holomorphic spinor observables and probabilistic techniques.
\end{abstract}
\maketitle

\newpage
\tableofcontents

\section{Introduction}

The Ising model plays a central role in equilibrium statistical mechanics, being a standard example of an order-disorder phase transition in dimensions
two and above. Besides pure mathematical interest, it has found successful applications to several fields in theoretical physics and computer sciences.

The phase transition in the Ising model in two dimensions has been a subject of extensive study, both in the mathematics and in the physics literature. The value of the critical temperature on the square lattice was determined by
Kramers and Wannier \cite{kramers-wannier}. Onsager \cite{onsager} computed the free energy and the critical exponents of the model. Later on, many exact computations were carried out, in particular by McCoy and Wu \cite{mccoy-wu-i}.

Further, a gradual understanding of the Ising model at criticality led to the conjecture by Belavin, Polyakov and Zamolodchikov that its scaling limit (as well as scaling limits of other critical models) should be conformally invariant and described by Conformal Field Theory \cite{belavin-polyakov-zamolodchikov-i,belavin-polyakov-zamolodchikov-ii}.
Loosely speaking, this conjecture can be formulated as follows: for any conformal map $\varphi:\Omega\to\Omega'$,
\begin{center}
\noindent (scaling limit of the model on $\Omega'$) = $\varphi$(scaling limit of the model on $\Omega$).
\end{center}
In particular, if $\Omega_\delta$ are discrete approximations to a continuous planar domain $\Omega$, then various quantities (expectations, probabilities etc.) in the model, under a proper normalization, have conformally invariant or covariant limits as the mesh size $\delta$ tends to zero. Moreover, Conformal Field Theory predicts exact formulae for these limits.

In the \emph{full-plane} case, many results were obtained following the seminal works of Onsager and Kaufman in late 1940's. The Onsager's formula for the spontaneous magnetization was proven by Yang \cite{yang}. The diagonal and horizontal spin-spin correlations were explicitly computed by Wu \cite{mccoy-wu-i}. A number of remarkable results were obtained for the massive limits, see \cite{wu-mccoy-tracy-barouch,sato-miwa-jimbo-iv,palmer-tracy} and references therein. Palmer \cite{palmer} justified the CFT predictions at criticality by taking the zero-mass limit. At criticality, the full-plane {energy} correlation functions (that is, the correlations of $n$ pairs of neighboring spins) were computed on periodic isoradial graphs by Boutillier and de Tili\`ere \cite{boutillier-de-tiliere-ii,boutillier-de-tiliere-i}, and the 2n-point full-plane {spin} correlation functions were treated by Dub\'edat, combining exact bosonization techniques \cite{dubedat-ii} and results on monomer correlations in the dimer model \cite{dubedat-i}.

However, the group of conformal self-maps of the full plane is only finite dimensional. Hence, in order to reveal the full strength of the conformal invariance property, it is important to consider \emph{general planar domains with boundary}. In this setting, mathematical proofs of {conformal invariance} and CFT predictions at criticality have remained out of reach until recently. Smirnov \cite{smirnov-i} has rigorously established conformal covariance of the \emph{fermionic observables} in the Ising model on the square grid. Later, this result has been proven to be universal in the family of isoradial graphs \cite{chelkak-smirnov-ii}, and led to the proof of convergence of the interfaces in the Ising model to Schramm's SLE${}_3$ curves \cite{CDHKS}. At the same time, the scaling limit of \emph{energy correlations} for bounded domains has been rigorously treated in \cite{hongler-smirnov-ii,hongler-i}, confirming the CFT predictions for the energy field. Nevertheless, the corresponding question about the \emph{spin correlations} remained open.

In this paper, we rigorously prove the existence and the conformal covariance of the scaling limits of all the multi-point spin correlation functions in any simply connected planar domain with $+$ boundary conditions at criticality.

Our main result (see Theorem~\ref{thm:1-3pts}) reads as follows. Let $\E_{\Omega_\delta}^{+}[\sigma_{a_1}\dots\sigma_{a_n}]$ denote the correlation of spins at the sites $a_1,\dots,a_n$ with $+$ boundary conditions in a discrete domain $\Omega_\delta$. Then
\[
\delta^{-\frac{n}{8}}\E_{\Omega_\delta}^{+}[\sigma_{a_1}\dots\sigma_{a_n}] \to \mathcal{C}^n\cdot\left\langle\sigma_{a_1}\dots\sigma_{a_n}\right\rangle^+_{\Omega}
\]
as $\Omega_\delta$ approximates $\Omega$ and the mesh size $\delta$ tends to zero. Here $\mathcal{C}$ is an explicit lattice-dependent constant and  $\left\langle\sigma_{a_1}\dots\sigma_{a_n}\right\rangle^+_{\Omega}$ is an explicit conformally covariant tensor of degree $\frac{1}{8}$ with respect to each of the variables, see (\ref{eq: zeta_constant}), (\ref{12ptFcts}) and (\ref{conf-cov}).

For example, in the case of the magnetization (the expectation of a single spin), one gets
\[
\delta^{-\frac{1}{8}}\E_{\Omega_\delta}^{+}[\sigma_a] \to \mathcal{C}\cdot2^{\frac14}\mathrm{rad}^{-\frac{1}{8}}(a,\Omega),
\]
where $\text{rad}(a,\Omega)$ denotes the conformal radius of $\Omega$ as seen from $a$, in other words, $\text{rad}(a,\Omega)=|\varphi'(0)|$, where $\varphi$ is a conformal map from the disc $\{z\in\C:|z|<1\}$ to $\Omega$ mapping the origin to $a$. A sketch of the proof of this result can also be found in the ICMP2012 proceedings \cite{hongler-icmp-2012}.

In the case of free boundary conditions, we establish a similar convergence result for the two-point function (see Theorem~ \ref{thm:2pts}). In that framework, the convergence of the corresponding $n$-point correlations (which by symmetry are only non-zero for even $n$) can be obtained by our methods as well, but for conciseness it is not included in the present paper.

Our previous results \cite{chelkak-izyurov} immediately allow one to treat alternating $+$/$-$ boundary conditions (see Corollary \ref{rem: alt_bc}). The technique we use also applies to \emph{mixed correlations}, involving spin, energy, disorder and boundary change operators, and extends to \emph{multiply connected domains}, which will be worked out in a subsequent paper.

The explicit formulae for what we prove to be the scaling limits of $\E_{\Omega_\delta}^{+}[\sigma_{a_1}\dots\sigma_{a_n}]$ were predicted by Conformal Field Theory methods in a number of papers originating in the seminal work \cite{belavin-polyakov-zamolodchikov-i}. In \cite{cardy-i}, it was explained how to handle the half-plane case by CFT means, in particular, the two-point correlations were treated. This result was later extended to $n=3$ \cite{burkhard-guim-i}, and extrapolated to larger $n$ \cite{burkhard-guim}.

Our method is based on the extraction of information from some discrete holomorphic observables in bounded domains by means of discrete complex analysis -- the approach that was firstly implemented in \cite{smirnov-i,smirnov-ii} for basic fermionic observables. More precisely, we use the spinor version of those which was introduced in \cite{chelkak-izyurov}. Fermionic observables per se essentially go back to the Kaufman-Onsager considerations and can be written as a product $\psi_z=\sigma_z\mu_z$ of spin and disorder operators in the notation of \cite{kadanoff-ceva}. The correlators $\langle\psi_z\sigma_{a_1}\mu_{a_2}\dots\mu_{a_n}\rangle$ can be found in the works of Kyoto's school \cite{sato-miwa-jimbo-i,sato-miwa-jimbo-ii,sato-miwa-jimbo-iii,sato-miwa-jimbo-iv}. However, rigorous proofs of convergence results require a delicate analysis of some Riemann-type boundary value problems for discrete holomorphic functions, the analysis thereof was initiated more recently \cite{smirnov-i,smirnov-ii,chelkak-smirnov-ii}.

Simultaneously and independently of our work, Dub\'edat announced analogous results for $2n$-point spin correlations in bounded domains $\Omega$ 
via the exact bosonization approach \cite{dubedat-ii}, and Camia, Garban and Newman obtained some results \cite{camia-garban-newman-i} about the properly renormalized spin field seen as a random generalized function (i.e. Schwartz distribution) on $\Omega$.

\subsection{Main results} The Ising model on a graph $\mathcal{G}$ is a random assignment of
$\pm1$ spins to the vertices of $\mathcal{G}$. In our paper we prefer a dual setup and consider the model on the \emph{faces} $\mathcal{F}=\mathcal{V}^\circ_{\Omega_\delta}$ of lattice approximations $\Omega_\delta$, $\delta\to 0$, to a bounded simply connected planar domain $\Omega\subset\C$. The probability of a spin configuration $\sigma\in\left\{ \pm1\right\} ^{\mathcal{F}}$ is proportional to $e^{-\beta\mathbf{H}\left(\sigma\right)}$, where $\beta>0$ is the inverse temperature and
\[
\textstyle \mathbf{H}\left(\sigma\right):=-\sum_{x\sim y}\sigma_{x}\sigma_{y}
\]
is the energy of the configuration. More precisely,
we will work with discrete planar domains $\Omega_\delta$ which are subsets of the \emph{square grid} rotated by 45\textdegree{} of diagonal mesh size $2\delta$ (the distance between adjacent spins is thus $\sqrt{2}\delta$, see Figure~\ref{Fig:discrete-domain}).

From now on we only consider the model at its \emph{critical point}, which for the square grid corresponds to the parameter value $\beta_\mathrm{c}=\frac{1}{2}\ln\left(\sqrt{2}+1\right)$. 
We also introduce a (lattice-dependent) constant that will appear in the statements of our theorems:
\begin{equation}
\label{eq: zeta_constant}
\mathcal{C}=2^{\frac{1}{6}}e^{-\frac32\zeta'\left(-1\right)},
\end{equation}
where $\zeta'$ denotes the derivative of Riemann's zeta function.

Below we follow the CFT notation $\langle\,\dots\rangle$ to represent the (predicted formulae for the scaling limit of) correlations on a simply connected domain $\Omega$: we define $\langle \sigma_{a_1}\dots\sigma_{a_n}\rangle_{\Omega}^\mathfrak{b}$ as some \emph{explicit functions} of the points $\Ak\in\Omega$ which may also depend on $\Omega$ and boundary conditions $\mathfrak{b}$ as parameters.
In particular, we define the one-point and the two-point functions $\left\langle \sigma_{a}\right\rangle _{\Omega}^{+}$ and $\left\langle \sigma_{a}\sigma_{b}\right\rangle _{\Omega}^{+}$, $\left\langle \sigma_{a}\sigma_{b}\right\rangle_{\Omega}^{\mathrm{free}}$ on the upper half-plane $\bH$ by the formulae
\begin{align}
\label{Expl_free}
\langle \sigma_{a}\rangle_{\bH}^{+}=
\frac{2^{\frac14}}{(2\Im\mathfrak{m}\,a)^{\frac18}}\,,\quad
\langle \sigma_{a}\sigma_{b}\rangle_{\bH}^{+} & = \langle \sigma_{a}\rangle_{\bH}^{+}\langle \sigma_{b}\rangle_{\bH}^{+}\cdot \biggl[\frac12\biggl(\left|\frac{b-\overline{a}}{b-a}\right|^{\frac12}\!+\left|\frac{b-a}{b-\overline{a}}\right|^\frac12\biggr)\biggr]^{\frac12},
\cr
 \langle \sigma_{a}\sigma_{b}\rangle_{\bH}^{\mathrm{free}} & =  \langle \sigma_{a}\rangle_{\bH}^{+}\langle \sigma_{b}\rangle_{\bH}^{+}\cdot {\biggl[\frac12\biggl(\left|\frac{b-\overline{a}}{b-a}\right|^{\frac12}\!-\left|\frac{b-a}{b-\overline{a}}\right|^\frac12\biggr)\biggr]^{\frac12}}
\end{align}
and on all other simply-connected domains $\Omega\ne\C$ by the condition that, for any conformal mapping $\varphi:\Omega\to\Omega'$, one has
\begin{align}
\left\langle \sigma_{a}\right\rangle_{\Omega}^{+} = \langle\sigma_{\varphi(a)}\rangle_{\Omega'}^+\cdot|\varphi'\left(a\right)|^{\frac{1}{8}},\quad
\left\langle \sigma_{a}\sigma_{b}\right\rangle_{\Omega}^{+} & = \langle\sigma_{\varphi(a)}\sigma_{\varphi(b)}\rangle_{\Omega'}\cdot|\varphi'\left(a\right)|^{\frac{1}{8}}|\varphi'\left(b\right)|^{\frac{1}{8}},
\cr
\left\langle \sigma_{a}\sigma_{b}\right\rangle_{\Omega}^{\mathrm{free}} & = \langle\sigma_{\varphi(a)}\sigma_{\varphi(b)}\rangle_{\Omega'}\cdot|\varphi'\left(a\right)|^{\frac{1}{8}}|\varphi'\left(b\right)|^{\frac{1}{8}}.
\label{eq: conf_conv_2}
\end{align}
Note in particular that this definition is consistent, as (\ref{eq: conf_conv_2}) is satisfied when $\Omega=\Omega'=\bH$ and $\varphi$ is a M\"obius map since the ratio ${|b-\overline{a}|}/{|b-a|}$ is M\"obius invariant. Therefore, due to the Riemann mapping theorem, the relations (\ref{Expl_free}), (\ref{eq: conf_conv_2}) uniquely define $\left\langle \sigma_{a}\sigma_{b}\right\rangle^{+}_{\Omega}$ and $\left\langle \sigma_{a}\sigma_{b}\right\rangle_{\Omega}^{\text{free}}$ for any simply-connected domain $\Omega\neq \C$.  Equivalent formulae for the one- and two-point functions can be given in terms of the hyperbolic distance $\mathrm{d}_{\Omega}(a,b)$ and the conformal radius $\text{rad}(a,\Omega)$:
\begin{align*}
\left\langle \sigma_{a}\right\rangle _{\Omega}^{+}  =  2^{\frac14}\text{rad}^{-\frac{1}{8}}(a,\Omega),\quad \left\langle \sigma_{a}\sigma_{b}\right\rangle _{\Omega}^{+} & = \left\langle \sigma_{a}\right\rangle _{\Omega}^{+}\left\langle \sigma_{b}\right\rangle _{\Omega}^{+}\cdot(1-\exp(-2\mathrm{d}_{\Omega}(a,b)))^{-\frac{1}{4}},\\ \left\langle \sigma_{a}\sigma_{b}\right\rangle _{\Omega}^{\mathrm{free}} & =
\left\langle \sigma_{a}\sigma_{b}\right\rangle _{\Omega}^{+} \cdot \exp\left(-\tfrac12\mathrm{d}_{\Omega}(a,b)\right).
\end{align*}

We have the following convergence theorem for the two-point functions, both with $+$ (the spins on the boundary of $\Omega_\delta$ are set to $+$) and free (no restrictions are set for boundary spins) boundary conditions:
\begin{thm}\label{thm:2pts}
Let $\Omega$ be a bounded simply connected domain and for $\delta>0$, let $\Omega_\delta$ be a discretization of $\Omega$ by the square grid of diagonal mesh size $2\delta$. Then, for any $\epsilon>0$, we have
\begin{eqnarray*}
\delta^{-\frac14}\mathbb{E}_{\Omega_{\delta}}^{+}\left[\sigma_{a}\sigma_{b}\right] \; \to \; \mathcal{C}^2\cdot\left\langle \sigma_{a}\sigma_{b}\right\rangle _{\Omega}^{+}\quad\text{and}\quad
\delta^{-\frac14}\mathbb{E}_{\Omega_{\delta}}^{\mathrm{free}}\left[\sigma_{a}\sigma_{b}\right] \; \to \; \mathcal{C}^2 \cdot \left\langle \sigma_{a}\sigma_{b}\right\rangle _{\Omega}^{\mathrm{free}},
\end{eqnarray*}
as $\delta\to 0$, uniformly over all $a,b\in\Omega$ at distance at least $\epsilon$ from $\partial\Omega$ and from each other, where the constant $\mathcal{C}$ is given by (\ref{eq: zeta_constant}).
\end{thm}

Our result generalizes to multipoint correlations as follows. As for the two-point functions, define the continuous correlation functions $\langle\sigma_{a_1}\dots\sigma_{a_n}\rangle_{\Omega}^{+}$ in the upper half-plane $\bH$ by the explicit formula
\def\MuPm{{\mu \in \{ \pm 1 \}^n}}
\begin{equation}
\label{12ptFcts}
 \langle \sigma_{a_1} \dots \sigma_{a_n}\rangle_{\bH}^{+}=\prod\limits_{k=1}^{n}\frac{1}{\left(2\Im\mathfrak{m}\, a_k\right)^{\frac18}}\,\cdot \left(2^{-\frac{n}{2}}\sum\limits_{\MuPm}
 \, \prod\limits_{1\le k<m\le n}\left|\frac{a_k-a_m}{a_k-\overline{a}_m}\right|^{\frac{\mu_k\mu_m}{2}}\right)^{\!\frac12}
\end{equation}
and on other domains by the condition that for any conformal mapping $\phi:\Omega\to\Omega'$, we have
\begin{equation}
\label{conf-cov}
\langle\sigma_{a_1}\dots\sigma_{a_n}\rangle_{\Omega}^{+}=\langle\sigma_{\varphi(a_1)}\dots\sigma_{\varphi(a_n)}\rangle _{\Omega'}^{+} \cdot {\textstyle\prod_{k=1}^n}|\varphi'\left(a_k\right)|^{\frac{1}{8}}.
\end{equation}
Note as before that the explicit functions (\ref{12ptFcts}) are covariant under M\"obius maps and thus the multipoint correlations are well defined.

\begin{thm}\label{thm:1-3pts}
Let $\Omega$ be a bounded simply connected domain and let $\Omega_\delta$ be a discretization of $\Omega$ by the square grid of diagonal mesh size $2\delta$. Then, for any $\epsilon>0$ and any $n=1,2,\dots$, we have
\begin{eqnarray*}
\delta^{-\frac{n}{8}}\cdot \mathbb{E}_{\Omega_{\delta}}^{+}\left[\sigma_{a_1}\dots\sigma_{a_n}\right] & \to & \mathcal{C}^n\cdot \left\langle \sigma_{a_1}\dots\sigma_{a_n}\right\rangle _{\Omega}^{+}
\end{eqnarray*}
as $\delta\to 0$, uniformly over all $a_1,\dots,a_n\in\Omega$ at distance at least $\epsilon$ from $\partial\Omega$ and from each other.
\end{thm}
Using the results of \cite{chelkak-izyurov}, one immediately arrives at the following generalization. Let $\mathfrak{b}_\delta=\left\{ b^\delta_{1},\dots,b^\delta_{2m}\right\}$ be a collection of points on $\partial\Omega_\delta$ appearing in counterclockwise
order and approximating a continuous collection $\mathfrak{b}=\{b_1,\dots,b_{2m}\}\subset\partial\Omega$. Denote by $\mathbb{E}_{\Omega_{\delta}}^{\mathfrak{b}_\delta}$ an expectation for the Ising model with $+$ boundary conditions on the counterclockwise arcs $\left[b^\delta_{2j-1},b^\delta_{2j}\right]$ and $-$ boundary condition on the complementary arcs $\left[b^\delta_{2j},b^\delta_{2j+1}\right]$, $j=1,\dots,m$, where we set $b^\delta_{2m+1}:=b^\delta_1$.

\begin{cor}
Suppose that the approximation $\Omega_\delta$ of $\Omega$ is regular near the points of $\mathfrak{b}$ in the sense of \cite[Definition 3.14]{chelkak-izyurov}. As $\delta\to 0$, one has
\[\delta^{-\frac{n}{8}}\mathbb{E}_{\Omega_{\delta}}^{\mathfrak{b}_\delta}[\sigma_{a_1}{\dots}\sigma_{a_n}]\to \mathcal{C}^n\cdot\langle \sigma_{a_1}{\dots}\sigma_{a_n}\rangle_\Omega^\mathfrak{b},
\]
where $\langle \sigma_{a_1}{\dots}\sigma_{a_n}\rangle_\Omega^\mathfrak{b}$ is explicit and satisfies the conformal covariance property (\ref{conf-cov}).
\label{rem: alt_bc}
\end{cor}
\begin{proof}
 Write $\mathbb{E}_{\Omega_{\delta}}^{\mathfrak{b}_\delta}[\sigma_{a_1}{\dots}\sigma_{a_n}]= ({\mathbb{E}_{\Omega_{\delta}}^{\mathfrak{b}_\delta}[\sigma_{a_1}{\dots}\sigma_{a_n}]}/{\mathbb{E}_{\Omega_{\delta}}^{+}[\sigma_{a_1}{\dots}\sigma_{a_n}]})\cdot \mathbb{E}_{\Omega_{\delta}}^{+}[\sigma_{a_1}{\dots}\sigma_{a_n}]$. By \cite[Corollary 5.10]{chelkak-izyurov}, the first term converges to an explicit conformally \emph{invariant} limit. Thus the result follows from Theorem~\ref{thm:1-3pts}.
\end{proof}
\begin{rem}
The condition of \emph{boundedness} of $\Omega$ in Theorems \ref{thm:2pts}, \ref{thm:1-3pts} above is imposed for the sake of simplicity and can be removed with no essential changes in the proofs. In particular, one could take the \emph{bulk limit} (i.e., $\Omega_\delta\to \C$ and $\delta\to 0$ while keeping the positions of $\Ak$ fixed), in which case \cite{difrancesco-saleur-zuber, dubedat-ii}
\[
\delta^{-\frac{n}{8}}\cdot \mathbb{E}_{\Omega_{\delta}}^{+}\left[\sigma_{a_1}\dots\sigma_{a_n}\right] \to \mathcal{C}^n\cdot
\left\langle \sigma_{a_1}\dots\sigma_{a_n}\right\rangle _{\C},
\]
where
\begin{equation}
\label{expl_bulk}
\left\langle \sigma_{a_1}\dots\sigma_{a_n}\right\rangle _{\C}=\biggl(2^{-\frac{n}{2}}\!\!\!\sum\limits_{\substack{\MuPm:\\\mu_1+\dots+\mu_n=0}}
 \, \prod\limits_{1\le k<m\le n}\left|a_k-a_m\right|^{\frac{\mu_k\mu_m}{2}}\biggr)^{\frac12}
\end{equation}
(this formula can be seen as a limit of (\ref{12ptFcts}) when all $a_k$ are replaced by $a_k+iy$ and $y\to+\infty$).
It is worth noting that, if $\Ak$ are kept at a \emph{finite number of lattice steps} from each other, then other scaling exponents appear. For $n=2$, the scaling limit of the {energy densities} $\lim_{\delta\to 0}\delta^{-1}(\mathbb{E}_{\Omega_{\delta}}\left[\sigma_{a}\sigma_{a'}\right]-\sqrt{2}/2)$, where two {neighboring} faces $a,a'$ approximate the same point $a\in\Omega$, has been treated in~\cite{hongler-smirnov-ii}. For $n>2$, two terms of asymptotics have been obtained at \cite{gheissari-hongler-park}. One can also wonder about the intermediate situation when $\Ak$ are at distances of order $\delta^{\beta}$ from each other, for some fixed $0<\beta<1$. Then, the leading term in the asymptotic expansion of the discrete correlation functions has the order $\delta^{\frac{1-\beta}{8}n}$ and is provided by the bulk limit (\ref{expl_bulk}). A more sophisticated analysis is required in order to find the second term of asymptotics in this case.
\end{rem}

\subsection{Key steps in the proof}\label{sub:key-theorems}
In this section we list the key results that allow us to prove Theorems~\ref{thm:2pts} and \ref{thm:1-3pts}. The first small step deals with the normalizing factors. It is a celebrated result of Wu \cite{mccoy-wu-i} that in the unique infinite-volume limit of the critical planar Ising model (i.e., in the case $\Omega=\C$), one has the following asymptotics:
\begin{equation}
\mathbb{E}_{\mathbb{C}_{\delta}}[\sigma_{0_\delta}\sigma_{1_\delta}]\sim \mathcal{C}^2\cdot\delta^{\frac{1}{4}},\quad \delta\to 0,
\label{eq: T_T_Wu}
\end{equation}
where $\mathbb{C}_{\delta}$ denotes the square grid $\delta (1\!+\!i)\mathbb{Z}^2$, while $0_\delta$ and $1_\delta$ stand for proper approximations of the points $0,1\in\mathbb{C}$ (keep in mind that our square lattice is rotated by 45\textdegree, so this is the \emph{diagonal} spin-spin correlation).
Instead of deriving the correct normalization of spin correlations in bounded domains directly, we relate it to the behavior of the normalizing factor
\[
\varrho\left(\delta\right):=\mathbb{E}_{\mathbb{C}_{\delta}}[\sigma_{0_\delta}\sigma_{1_\delta}].
\]
Namely, we prove that, as $\delta\to 0$,
\begin{eqnarray*}
(\varrho(\delta))^{-\frac{n}{2}}\cdot \mathbb{E}_{\Omega_{\delta}}^{+}\left[\sigma_{a_1}\dots\sigma_{a_n}\right] & \to &  \left\langle \sigma_{a_1}\dots\sigma_{a_n}\right\rangle _{\Omega}^{+}, \\
(\varrho(\delta))^{-1}\cdot\mathbb{E}_{\Omega_{\delta}}^{\mathrm{free}}\left[\sigma_{a}\sigma_{b}\right] & \to & \left\langle \sigma_{a}\sigma_{b}\right\rangle _{\Omega}^{\mathrm{free}},
\end{eqnarray*}
which, combined with (\ref{eq: T_T_Wu}), readily gives Theorems \ref{thm:2pts} and \ref{thm:1-3pts}. We point out that apart from this reduction, we never use (\ref{eq: T_T_Wu}) in this paper. On the other hand, our methods also allow one to give a new proof of the explicit formula for the diagonal spin-spin correlations in the full-plane case as well as to derive  explicit formulae for the magnetization in the half-plane, see the forthcoming work \cite{chelkak-hongler}. 

\smallskip

The following theorem, concerning \emph{discrete logarithmic derivatives} of the spin correlations with $+$ boundary conditions, is a cornerstone for the whole paper:
\begin{thm}
\label{thm:log-derivatives}
Let $\Omega$ be a bounded simply connected domain and $\Omega_\delta$ be discretizations of $\Omega$ by the square grids $\delta(1\!+\!i)\mathbb{Z}^2$. Then, for any $\epsilon>0$ and any $n=1,2\dots$, we have
\begin{eqnarray}
\frac{1}{2\delta}\left(\frac{\mathbb{E}_{\Omega_{\delta}}^{+}\left[\sigma_{a_1+2\delta}\sigma_{a_2}\dots\sigma_{a_n}\right]} {\mathbb{E}_{\Omega_{\delta}}^{+}\left[\sigma_{a_1}\dots\sigma_{a_n}\right]}-1\right) & \to & \Re\mathfrak{e}\mathcal{A}_\Omega(a_1,\dots,a_n),
\label{eq:log-der-x}\\
\frac{1}{2\delta}\left(\frac{\mathbb{E}_{\Omega_{\delta}}^{+}\left[\sigma_{a_1+2i\delta}\sigma_{a_2}\dots\sigma_{a_n}\right]} {\mathbb{E}_{\Omega_{\delta}}^{+}\left[\sigma_{a_1}\dots\sigma_{a_n}\right]}-1\right) & \to & -\Im\mathfrak{m}\mathcal{A}_\Omega(a_1,\dots,a_n)
\label{eq:log-der-y}
\end{eqnarray}
as $\delta\to 0$, uniformly over all faces $a_1,\dots a_n\in\Omega_{\delta}$ at distance at least $\epsilon$ from $\partial\Omega$ and from each other. The function $\mathcal{A}_\Omega(a_1,\dots,a_n)$ is defined explicitly via the solution to some special interpolation problem (see further details in Section~\ref{sub:cont-spinors}) and has the following covariance property under conformal mappings $\varphi:\Omega\to\Omega'$:
\begin{equation}
\label{A-covariance}
\mathcal{A}_\Omega(a_1,\dots,a_n)= \mathcal{A}_{\Omega'}(\varphi(a_1),\dots,\varphi(a_n))\cdot\varphi'(a_1)+
\frac{1}{8}\frac{\varphi''(a_1)}{\varphi'(a_1)}.
\end{equation}
\end{thm}
\begin{proof}
The proof is based on the convergence results for the discrete spinor observables. A rather delicate analysis is needed since we are interested in the values of observables near their singular points. See further details in Sections~\ref{sec:spinors-and-correlation}~and~\ref{sec:analysis-of-spinors}.
\end{proof}

Integrating the result of Theorem~\ref{thm:log-derivatives}, 
we get the following weaker form of the convergence result for the spin correlations. Note that the conformal covariance degree $\frac{1}{8}$ in (\ref{conf-cov}) is a direct consequence of the covariance rule (\ref{A-covariance}) for $\mathcal{A}_\Omega(a_1,\dots,a_n)$.

\begin{cor}
\label{cor:conv-up-to-normalization}
Under conditions of Theorem~\ref{thm:log-derivatives}, for any $n\ge 1$, there exist some normalizing factors $\varrho_{n}(\delta,\Omega_\delta)$ that might depend on $\Omega_\delta$ but not on the positions of the points $a_1,\dots,a_n$ such that
\[
\mathbb{E}_{\Omega_{\delta}}^{+}\left[\sigma_{a_1}\dots\sigma_{a_n}\right] ~\sim~ \varrho_{n}(\delta,\Omega_\delta)\cdot \left\langle \sigma_{a_1}\dots\sigma_{a_n}\right\rangle _{\Omega}^{+}
\]
as $\delta\to 0$, uniformly over all faces $a_1,\dots a_n\in\Omega_{\delta}$ at distance at least $\epsilon$ from $\partial\Omega$ and from each other.
\end{cor}
\begin{proof}
See Section~\ref{sub:expl_formulae}
\end{proof}

We now focus on the special case $n=2$. The next theorem is a crucial tool which allows us to compare the normalizing factors $\varrho_2(\delta,\Omega_\delta)$ with the full-plane case. We denote by $\mathbb{E}_{\Omega_{\delta}^\bullet}^{\mathrm{free}}$ the expectation for the critical Ising model defined on the \emph{vertices} of $\Omega_\delta$ (with free boundary conditions).
\begin{thm}
\label{thm:ratio-free-plus}
Let $\Omega$ be a bounded simply connected domain and $\Omega_\delta$ be a discretization of $\Omega$ by the square grid $\delta(1\!+\!i)\mathbb{Z}^2$.
Then, for any $\epsilon>0$, we have
\[
\frac{\mathbb{E}_{\Omega_{\delta}^\bullet}^{\mathrm{free}}\left[\sigma_{a+\delta}\sigma_{b+\delta}\right]}{\mathbb{E}_{\Omega_{\delta}}^{+}\left[\sigma_{a}\sigma_{b}\right]} ~\underset{\delta\to0}{\longrightarrow}~
\mathcal{B}_{\Omega}(a,b),
\]
uniformly over all faces $a,b$ at distance at least $\epsilon$ from $\partial\Omega$ and from each other, where $\mathcal{B}_{\Omega}(a,b)$ is a conformal invariant of $(\Omega,a,b)$ which can be explicitly written as $\mathcal{B}_{\Omega}(a,b)={\left\langle \sigma_{a}\sigma_{b}\right\rangle_{\Omega}^{\mathrm{free}}}/{\left\langle \sigma_{a}\sigma_{b}\right\rangle_{\Omega}^{+}}$.
\end{thm}
\begin{proof}
The proof is based on the convergence results for the discrete spinor observables and the Kramers-Wannier duality,
see further details in Sections~\ref{sub:discrete-spinors-to-ratios-corr},~\ref{sub:convergence-of-spinors}. 
\end{proof}

\begin{proof}[Sketch of the proof of Theorem~\ref{thm:2pts}]
Having the results of Corollary~\ref{cor:conv-up-to-normalization} and Theorem~\ref{thm:ratio-free-plus}, we only need to prove that $\varrho_2(\delta,\Omega_\delta)\sim \varrho(\delta)$ as $\delta\to 0$. The classical FKG inequality gives
\[
\mathbb{E}_{\Omega_{\delta}^\bullet}^{\mathrm{free}}\left[\sigma_{a+\delta}\sigma_{b+\delta}\right]\le \mathbb{E}_{\mathbb{C}_{\delta}}\left[\sigma_{a}\sigma_{b}\right]\le \mathbb{E}_{\Omega_{\delta}}^{+}\left[\sigma_{a}\sigma_{b}\right]
\] and it is easy to see that $\mathcal{B}_{\Omega}(a,b)\to 1$ as $b$ approaches $a$. Since the normalizing factors $\varrho_2(\delta,\Omega_\delta)$ do not depend on the positions of $a,b\in\Omega$, we conclude that $\mathbb{E}_{\Omega_{\delta}}^{+}\left[\sigma_{a}\sigma_{b}\right]\sim \mathbb{E}_{\mathbb{C}_{\delta}}\left[\sigma_{a}\sigma_{b}\right]$ in the double limit when we first let $\delta\to 0$ and then $b\to a$. This relates $\varrho_2(\delta,\Omega_\delta)$ to the full-plane normalization,  
see details in Section~\ref{sub:ratios-to-thm-2pts}.
\end{proof}

\begin{proof}[Sketch of the proof of Theorem~\ref{thm:1-3pts}]
Once the asymptotics $\varrho_2(\delta,\Omega_\delta)\sim\varrho(\delta)$ as {$\delta\to 0$} is established, we derive the asymptotics of all the other $\varrho_{n}(\delta,\Omega_\delta)$, $n\ne 2$, using the following observation: as the point $a_1$ approaches $\partial\Omega$, the continuous correlation functions behave in the following way:
\begin{equation}
\label{eq: dec-identity-intro}
\langle\sigma_{a_1}\dots\sigma_{a_n}\rangle_\Omega^+\sim \langle\sigma_{a_1}\rangle_\Omega^+ \langle\sigma_{a_2}\dots\sigma_{a_n}\rangle_\Omega^+
\end{equation}
and the same decorrelation result $\E_{\Omega_\delta}^{+}[\sigma_{a_1}\dots\sigma_{a_n}]\sim\E_{\Omega_\delta}^{+}[\sigma_{a_1}]\E_{\Omega_\delta}^{+}[\sigma_{a_2}\dots\sigma_{a_n}]$ holds true in the double limit $\delta\to 0$ and $a_1\to\partial\Omega$. This implies the recurrent formula $\varrho_{n+1}(\delta,\Omega_\delta)\sim \varrho_1(\delta,\Omega_\delta)\varrho_n(\delta,\Omega_\delta)$ for $n=1,2,\dots$ and, further, $\varrho_{n}(\delta,\Omega_\delta)\sim(\varrho(\delta))^{n/2}$ for all $n$, see further details in Section \ref{sub:proof-thm-k-pts}.
\end{proof}

\subsection{Organization of the paper}
Section~\ref{sec:spinors-and-correlation} 
contains all the main ideas. Details, especially those involving hard discrete complex analysis techniques, are mostly postponed to Section~\ref{sec:analysis-of-spinors}. The readers not interested in these details may restrict themselves to Section~\ref{sec:spinors-and-correlation} only.

We fix the notation in Section~\ref{sub:notation-and-definitions}. The main tool of this paper, the discrete holomorphic spinor observables, is introduced and discussed in Sections~\ref{sub:discrete-spinors} and~\ref{sub:s-holomorphicity}. In Section~\ref{sub:discrete-spinors-to-ratios-corr}, we prove that the ratios of spin correlations that appear in Theorems~\ref{thm:log-derivatives} and~\ref{thm:ratio-free-plus} can be expressed in terms of these observables. In Sections~\ref{sub:cont-spinors} and~\ref{sub:convergence-of-spinors}, we discuss the continuous counterparts of the discrete observables and state the convergence Theorems \ref{thm:obs-away},~\ref{thm:localization-near-a} and~\ref{thm:localization-near-b}, that easily imply Theorems \ref{thm:log-derivatives} and~\ref{thm:ratio-free-plus}. We do some explicit computaions and prove Corollary~\ref{cor:conv-up-to-normalization} in Section \ref{sub:expl_formulae}. We complete the proofs of Theorems~\ref{thm:2pts} and~\ref{thm:1-3pts} in Sections~\ref{sub:ratios-to-thm-2pts} and~\ref{sub:proof-thm-k-pts}, respectively.

Section~\ref{sec:analysis-of-spinors} is devoted to the proof of Theorems \ref{thm:obs-away},~\ref{thm:localization-near-a} and~\ref{thm:localization-near-b}. We discuss the discrete properties of our observables and their full-plane analogue in Sections~\ref{sub:s-holomorphicity-proof}--\ref{sub: riemann-problem}, and finish the proof of the main convergence theorems in Sections~\ref{sub:convergence-away-sing} and~\ref{sub:convergence-near-sing}. In the Appendix, we explicitly compute $\mathcal{A}_{\bH}(a_1,\dots,a_n)$ for arbitrary $n$, which completes the proof of the formula (\ref{12ptFcts}).

\medskip \noindent {\bf Acknowledgements.}
All three authors wish to thank Stanislav Smirnov for introducing them to the Ising model and discrete complex analysis, for sharing his ideas and for constant support. We also would like to thank St\'ephane Benoist, David Cimasoni, Julien Dub\'edat, Hugo Duminil-Copin, Pierluigi Falco, Christophe Garban, Nikolai Makarov, Barry McCoy, Charles Newman, John Palmer, Yuval Peres, Mikhail Sodin, Tom Spencer and Yvan Velenik for many helpful discussions and comments.

The work of Dmitry Chelkak and Konstantin Izyurov was supported by the Chebyshev Laboratory (Department of Mathematics and Mechanics, St.-Petersburg State University) under RF Government grant 11.G34.31.0026. Dmitry Chelkak was partly supported by P.~Deligne's 2004 Balzan prize in Mathematics.
Cl\'ement Hongler was supported by the National Science Foundation under grant DMS-1106588 and by the Minerva Foundation.
Konstantin Izyurov was partially supported by ERC AG CONFRA, the Swiss National Science Foundation, and the Academy of Finland.

\section{\label{sec:spinors-and-correlation}Holomorphic Spinors and Correlation Functions}

\subsection{\label{sub:notation-and-definitions}Notation} We start by fixing the notation which is used throughout the paper.

\begin{figure}

\centering{\includegraphics[height=0.33\textheight]{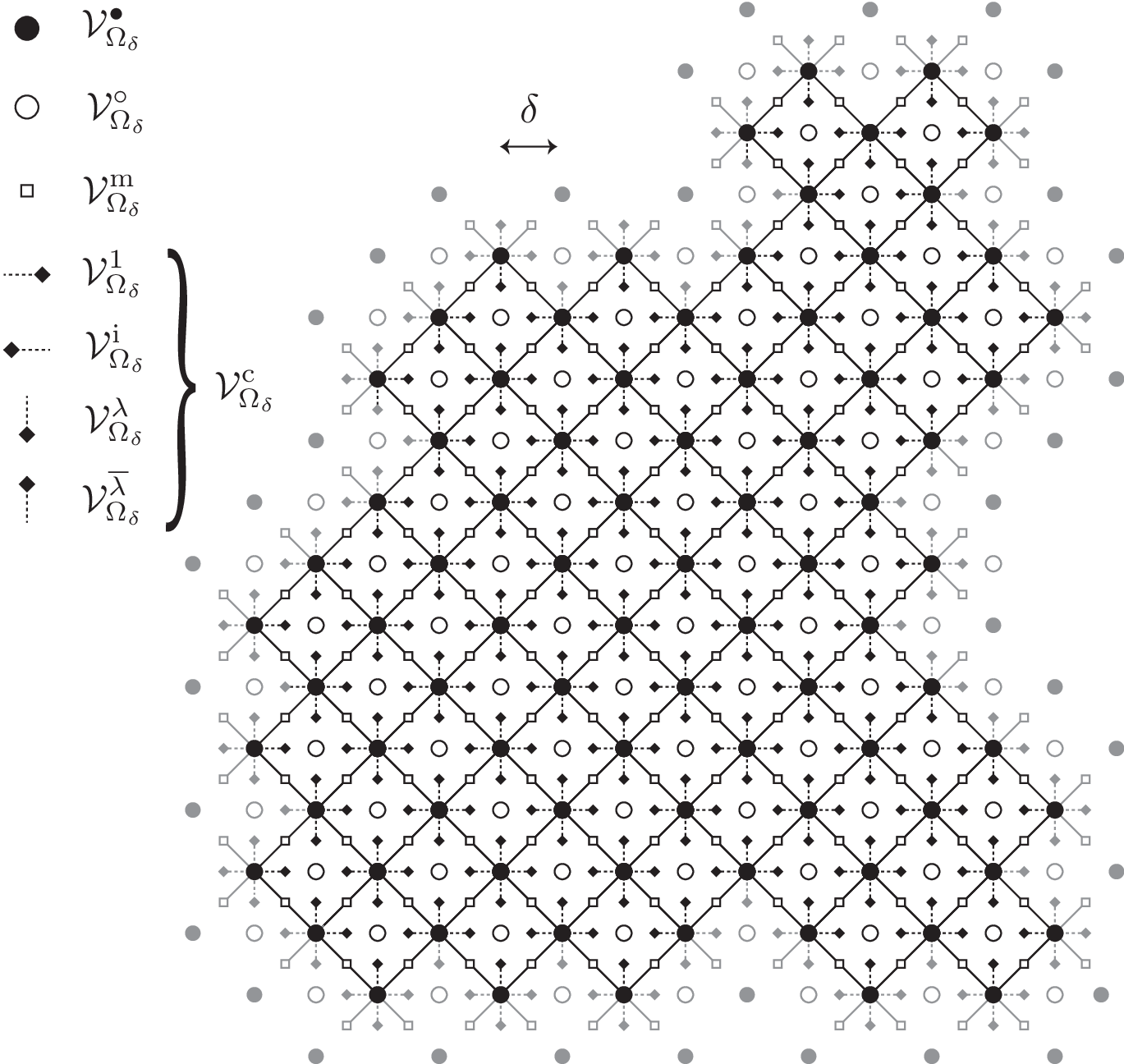}}

\caption{\label{Fig:discrete-domain} An example of a discrete domain $\Omega_\delta$ and notation for the sets of vertices ($\mathcal{V}_{\Omega_\delta}^\bullet$), faces ($\mathcal{V}_{\Omega_\delta}^\circ$), edge midpoints ($\mathcal{V}_{\Omega_\delta}^\mathrm{m}$) and (four types) of corners ($\mathcal{V}_{\Omega_\delta}^\mathrm{c}=\mathcal{V}_{\Omega_\delta}^\mathrm{1}\cup\mathcal{V}_{\Omega_\delta}^\mathrm{i}\cup \mathcal{V}_{\Omega_\delta}^\mathrm{\lambda}\cup\mathcal{V}_{\Omega_\delta}^\mathrm{\overline{\lambda}}$). The mesh size $\delta$ is a half-diagonal of
a square face, thus the distance between adjacent spins is $\sqrt{2}\delta$. The inner vertices, faces, edges and corners are colored black, while the boundary ones are colored gray.}
\end{figure}

\subsubsection{Graph notation}
Recall that we work on the square grid rotated by 45\textdegree{}
\[
\mathbb{C}_{\delta}:= \left\{\delta\left(1+i\right)\left(m+in\right):m,n\in\mathbb{Z}\right\}.
\]
The mesh size $\delta$ is hence the size of a half-diagonal of a square face. We often will identify the vertices of $\mathbb{C}_{\delta}$ with the corresponding complex numbers, the faces of $\mathbb{C}_{\delta}$ with their centers, the edges of $\mathbb{C}_{\delta}$ with their midpoints, etc.

We call \emph{discrete domain of mesh size $\delta$} a union of grid faces (see also Figure~\ref{Fig:discrete-domain} for notation given below). We say that $\Omega_{\delta}$ is \emph{simply connected} if the corresponding polygonal domain is simply connected.

\begin{itemize}
\item We denote by ${\mathrm{Int}\mathcal{V}_{\Omega_{\delta}}^{\circ}}$ the set of all \emph{faces} belonging to $\Omega_{\delta}$ (which are identified with their centers), by ${\mathrm{Int}\mathcal{V}_{\Omega_{\delta}}^{\bullet}}$ the set of all \emph{vertices} incident to these faces, and by ${\mathrm{Int}\mathcal{V}_{\Omega_{\delta}}^{\mathrm{m}}}$ the set of all \emph{edges} incident to $\mathcal{V}_{\Omega_{\delta}}^{\circ}$, which are identified with their midpoints (or the vertices of a medial lattice).
\end{itemize}

In order to simplify the presentation, we also assume that \emph{$\C_\delta\setminus\Omega_\delta$ has no fiords of a single face width}, i.e., all the edges joining vertices from $\mathrm{Int}\mathcal{V}_{\Omega_{\delta}}^{\bullet}$ belong to $\mathrm{Int}\mathcal{V}_{\Omega_{\delta}}^{\mathrm{m}}$. This technical assumption can be easily relaxed, if necessary.

\begin{itemize}
\item We denote by ${\partial\mathcal{V}_{\Omega_{\delta}}^{\circ}}$, ${\partial\mathcal{V}_{\Omega_{\delta}}^{\bullet}}$ and ${\partial\mathcal{V}_{\Omega_{\delta}}^{\mathrm{m}}}$ the sets of \emph{boundary faces, vertices and edges}, i.e. those faces, vertices and edges which are incident but do not belong to $\mathrm{Int}\mathcal{V}_{\Omega_{\delta}}^{\circ}$, $\mathrm{Int}\mathcal{V}_{\Omega_{\delta}}^{\bullet}$ and $\mathrm{Int}\mathcal{V}_{\Omega_{\delta}}^{\mathrm{m}}$, respectively (see Figure~\ref{Fig:discrete-domain}).
\item We set $\mathcal{V}_{\Omega_{\delta}}^{\circ}:=\mathrm{Int}\mathcal{V}_{\Omega_{\delta}}^{\circ}\cup\partial\mathcal{V}_{\Omega_{\delta}}^{\circ}$ etc.
\end{itemize}

Below we also need to work with four \emph{corners} of a given square face separately (in particular, see the Definition~\ref{def:discrete-spinor} below). For a given vertex $v\in\mathbb{C}_\delta$, we identify the nearby corners with the points $v\pm\frac{1}{2}\delta$ and $v\pm\frac{1}{2}\delta i$ on the complex plane.
\begin{itemize}
\item We denote by ${\mathcal{V}_{\Omega_{\delta}}^{\mathrm{c}}}$ the set of all corners incident to the vertices from $\mathrm{Int}\mathcal{V}_{\Omega_{\delta}}^{\bullet}$. We also set ${\mathcal{V}_{\Omega_{\delta}}^{\mathrm{cm}}}:=\mathcal{V}_{\Omega_{\delta}}^{\mathrm{c}}\cup\mathcal{V}_{\Omega_{\delta}}^{\mathrm{m}}$.
\item We use the notation $x\sim y$ if each of $x,y$ is either a vertex, a face, an edge or a corner, and they are adjacent or incident to each other.
\end{itemize}

\subsubsection{Double covers}\label{sub:double-covers}
In this paper we often deal with holomorphic functions (both discrete and continuous) which are defined on a double cover of a planar domain $\Omega$ and with opposite signs on opposite sheets (i.e., they have  a $-1$ multiplicative monodromy around branching points). We call such functions \emph{holomorphic spinors}. The following notation will be used below:
\begin{itemize}
\item For a planar domain $\Omega$ and $a\in\Omega$, we denote by ${\left[\Omega,a\right]}$
the double cover of $\Omega\setminus\left\{ a\right\}$  branching around $a$. All such double covers are naturally
viewed as subdomains of $\left[\mathbb{C},a\right]$. We often identify points on a double cover with their base points (so each base point is identified with the two points above it).
\item For several marked points $a_1,\dots,a_n\in\Omega$, we denote by $\left[\Omega,a_1,\dots,a_n\right]$
    the double cover of $\Omega\setminus\{a_1,\dots,a_n\}$ which branches around each of $a_1,\dots,a_n$.
\item We will often compare spinors defined on $\left[\Omega,a_1,\dots,a_n\right]$ with those defined on $[\C,a_1]$ (e.g., with 1/$\sqrt{z-a_1}$ or $\sqrt{z-a_1}$) near the common branching point $a_1$. Such equations are understood to be valid in a small neighborhood of $a_1$, together with the natural correspondence between the sheets of both covers.
\end{itemize}
We will also consider double covers of discrete domains. In this case the following (slightly modified) notation will be convenient:
\begin{itemize}
\item For a discrete domain $\Omega_{\delta}$ and a \emph{face} $a\in\mathrm{Int}\mathcal{V}_{\Omega_{\delta}}^{\circ}$, we set $\left[\Omega_{\delta},a^{\rightarrow}\right]:=\left[\Omega_{\delta},a\right]\setminus\left\{a+\frac{\delta}{2}\right\}$, excluding both points over the corner $a+\frac{\delta}{2}$ from the natural double cover branching at $a$. Similarly, we set $\left[\Omega_{\delta},a_1^{\rightarrow},a_2,\dots a_n\right]:=\left[\Omega_{\delta},a_1,\dots,a_n\right]\setminus\left\{a_1+\frac{\delta}{2}\right\}$, if several faces $a_1,\dots,a_n\in\mathrm{Int}\mathcal{V}_{\Omega_{\delta}}^{\circ}$ are marked.
\end{itemize}

\subsubsection{Contours}\label{sub:contours}
Recall that we consider the critical Ising model on the \emph{faces} of $\Omega_\delta$. In order to define the main tool of the paper (the discrete holomorphic spinors), we need some additional notation related to the contour representation of the model known as the \emph{low-temperature expansion}, see \cite[Chapter 1]{palmer}.
\begin{itemize}
\item We denote by $\mathcal{C}_{\Omega_{\delta}}$ the family of all collections of closed
contours on $\Omega_{\delta}$, i.e. the family of subsets of edges $\omega\subset\mathcal{V}_{\Omega_{\delta}}^\mathrm{m}$
such that every vertex $v\in\mathcal{V}_{\Omega_{\delta}}^{\bullet}$ belongs to an even number of edges in $\omega$.
\end{itemize}
The set $\mathcal{C}_{\Omega_{\delta}}$ is in a natural one-to-one correspondence with the spin configurations on $\Omega_\delta$
with $+$ boundary conditions: trace an edge between any two adjacent faces with opposite spins. Under this mapping, the probability of a collection of interfaces $\omega\subset\Omega_{\delta}$ becomes proportional to $\alpha_\mathrm{c}^{\#\mathrm{edges}\left(\omega\right)}$,
where $\alpha_\mathrm{c}=\exp\left(-2\beta_\mathrm{c}\right)=\sqrt{2}-1$.

\smallskip

Below we also introduce families of contour collections which, besides a number of closed loops, contain a single path running from one fixed corner $x$ to another corner or an edge midpoint $y$:
\begin{itemize}
\item For $x,y\in\mathcal{V}_{\Omega_{\delta}}^{\mathrm{cm}}$, let $\pi_{x,y}=x\sim v_{1}\sim\ldots \sim v_{n}\sim y$
be some simple lattice path with $v_{1},\ldots,v_{n}\in\mathcal{V}_{\Omega_{\delta}}^{\bullet}$.
We set $\mathcal{C}_{\Omega_{\delta}}\left(x,y\right):=\{\omega\oplus\pi_{x,y},\omega\in\mathcal{C}_{\Omega_{\delta}}\}$,
where $\oplus$ denotes the XOR, or symmetric difference. It is easy to see that the set $\mathcal{C}_{\Omega_{\delta}}\left(x,y\right)$ does not depend on the particular choice of $\pi_{x,y}$. Note that, for any  $\gamma\in\mathcal{C}_{\Omega_{\delta}}\left(x,y\right)$, there exists a (non-unique) decomposition of $\gamma$ into a collection of disjoint, simple loops and a path  $\mathrm{p}\left(\gamma\right)\subset\gamma$ running from $x$ to~$y$. By a decomposition we mean that each edge in $\gamma\in\mathcal{C}_{\Omega_{\delta}}\left(x,y\right)$ belongs to exactly one loop (or to $\mathrm{p}(\gamma)$) and is visited only once, and that there are no transversal intersections or self-intersections (see Figure~\ref{Fig:spinor}).
\end{itemize}

\subsection{\label{sub:discrete-spinors}Construction of the discrete spinor observables}

Now we are ready to introduce the discrete spinor observables. The following definition generalizes the construction given in \cite{chelkak-izyurov} to the case when a ``source point'' is inside $\Omega_\delta$.

\begin{defn}\label{def:discrete-spinor} Let $\Omega_{\delta}$ be a discrete domain and $a_1,\dots,a_n\in\mathrm{Int}\mathcal{V}_{\Omega_{\delta}}^{\circ}$
be inner faces. For a corner $z\in\mathcal{V}_{\left[\Omega_{\delta},a_1^{\rightarrow},\dots,a_n\right]}^{\mathrm{c}}$ (below we also extend this definition to edge midpoints, see Remark~\ref{rem:spinor-def}(iii)), we define
\begin{equation}
\label{spinor-corners-def} F_{\left[\Omega_{\delta},a_1,\dots,a_n\right]}\left(z\right) ~:=~ \frac{1}{\mathcal{Z}_{\Omega_{\delta}}^+\left[\sigma_{a_1}\dots\sigma_{ a_n}\right]}\sum_{\gamma\in\mathcal{C}_{\Omega_{\delta}}(a_1+\frac{\delta}{2},z)}\alpha_{\mathrm{c}}^{\#\mathrm{edges}\left(\gamma\right)}
\cdot\phi_{a_1,\dots,a_n}\left(\gamma,z\right),
\end{equation}
where
\begin{itemize}
\item $\#\mathrm{edges}\left(\gamma\right)$ is the number of full
edges contained in $\gamma$ and $\alpha_{\mathrm{c}}=\sqrt{2}-1$;
\item the complex phase $\phi_{a_1,\dots,a_n}(\gamma,z)$ is defined by (see also Figure~\ref{Fig:spinor})
\[
\phi_{a_1,\dots,a_n}\left(\gamma,z\right):= e^{-\frac{i}{2}\mathrm{wind}\left(\mathrm{p}\left(\gamma\right)\right)} \cdot \left(-1\right)^{\#\mathrm{loops}_{a_1,\dots,a_n}\left(\gamma\setminus\mathrm{p}\left(\gamma\right)\right)}\cdot \mathrm{sheet}_{a_1,\dots,a_n}\left(\mathrm{p}\left(\gamma\right),z\right),
\]
where, for a decomposition of $\gamma$ mentioned in Section~\ref{sub:contours},
\begin{itemize}
\item $\mathrm{wind}\left(\mathrm{p}\left(\gamma\right)\right)$ is
the total winding (increment of the argument of the tangent vector) of the path $\mathrm{p}\left(\gamma\right)$ when going from
$a_1+\frac{\delta}{2}$ to $z$,
\item $\#\mathrm{loops}_{a_1,\dots,a_n}\left(\gamma\setminus\mathrm{p}\left(\gamma\right)\right)$
is the number of loops in $\gamma\setminus\mathrm{p}\left(\gamma\right)$
that contain an odd number of marked points $a_1,\dots,a_n$ (equivalently, that do not lift to the double cover $[\Omega_\delta,a_1,\dots,a_n]$ as closed loops),
\item the last factor $\mathrm{sheet}_{a_1,\dots,a_n}\left(\mathrm{p}\left(\gamma\right),z\right)$
is equal to $+1$ if $z$ is on the same sheet of $\left[\Omega_{\delta},a_1,\dots,a_n\right]$ as the end of the lift of $\mathrm{p}\left(\gamma\right)$,
and to $-1$ otherwise (more precisely, we fix one of the two points lying over the ``source'' $a_1+\frac{\delta}{2}$ once forever and identify all other $z\in[\Omega_\delta,a_1,\dots,a_n]$ with paths running from \emph{this} $a_1+\frac{\delta}{2}$ to $z$ modulo homotopy and an appropriate index $2$ subgroup of the fundamental group);
\end{itemize}
\item the normalizing factor $\mathcal{Z}_{\Omega_{\delta}}^+\left[\sigma_{a_1}\dots\sigma_{a_n}\right]$
is defined by
\begin{eqnarray}
\label{eq: Z+=sum}
\mathcal{Z}_{\Omega_{\delta}}^+\left[\sigma_{a_1}\dots\sigma_{a_n}\right] & := & \sum_{\omega\in\mathcal{C}_{\Omega_{\delta}}}\alpha_{\mathrm{c}}^{\#\mathrm{edges}\left(\omega\right)} \left(-1\right)^{\#\mathrm{loops}_{a_1,\dots,a_n}\left(\omega\right)}.
\end{eqnarray}
\end{itemize}
\end{defn}

\begin{figure}

\centering{\includegraphics[height=0.33\textheight]{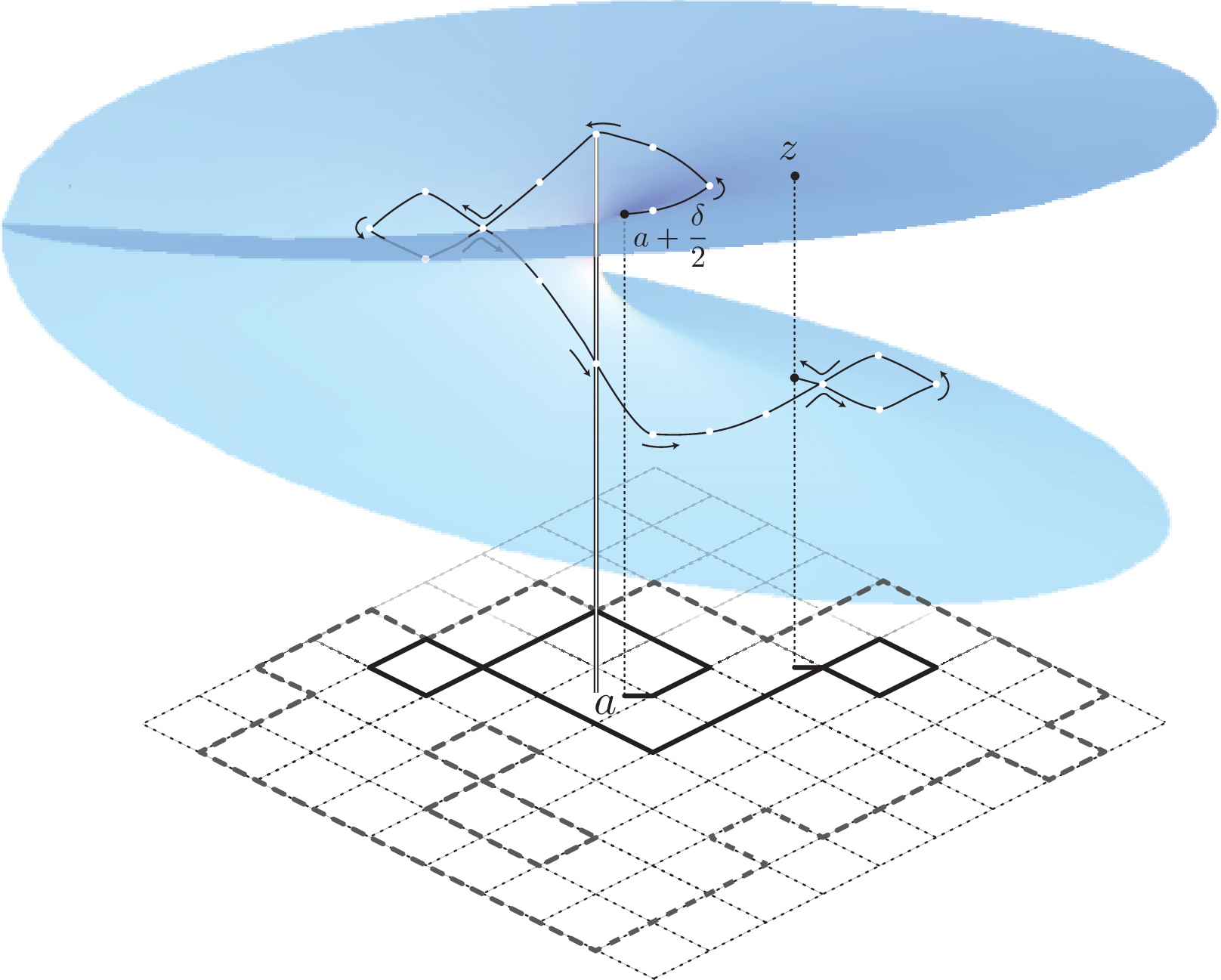}}

\caption{\label{Fig:spinor} A contour collection $\gamma\in\mathcal{C}_{\Omega_\delta}(a+\frac{\delta}{2},z)$ decomposed into non-intersecting loops (dashed) and a path $\mathrm{p}(\gamma)$. Running from $a+\frac{\delta}{2}$ to the projection of $z$, this path makes a $3\pi$ turn counterclockwise, thus $e^{-\frac{i}{2}\wind(\mathrm{p}(\gamma))}=i$. There is a single loop in $\gamma$ surrounding $a$, hence $(-1)^{\mathrm{loops}_a(\gamma\setminus \mathrm{p}(\gamma))}=-1$. Being lifted to the double cover $[\Omega_\delta,a]$, this path ends on the other sheet, thus $\mathrm{sheet}_a(\mathrm{p}(\gamma),z)=-1$, and $\phi_a(\gamma,z)=i$.}
\end{figure}

\begin{rem}\label{rem:spinor-def}
(i) For any $\omega\in\mathcal{C}_{\Omega_{\delta}}$, the sign $(-1)^{\#\mathrm{loops}_{a_1,\dots,a_n}\left(\omega\right)}$ coincides with the product of spins $\sigma_{a_1}\dots\sigma_{a_n}$ in the corresponding Ising model configuration with $+$ boundary values. Since $\alpha_{\mathrm{c}}^{\#\mathrm{edges}\left(\omega\right)}$ is just the Ising weight of $\omega$, one concludes that
\[
\mathcal{Z}_{\Omega_{\delta}}^+\left[\sigma_{a_1}\dots\sigma_{a_n}\right]= \mathbb{E}_{\Omega_{\delta}}^+\left[\sigma_{a_1}\dots\sigma_{a_n}\right]\cdot \mathcal{Z}_{\Omega_{\delta}}^+ >0,
\]
where $\mathcal{Z}_{\Omega_{\delta}}^+=\sum_{\omega\in\mathcal{C}_{\Omega_{\delta}}}\alpha_{\mathrm{c}}^{\#\mathrm{edges}}$ is the partition function of the model. 

\smallbreak

\noindent (ii) It is easy to check that the complex phase $\phi_{a_1,\dots,a_n}\left(\gamma,z\right)$ is independent of the choice of a decomposition of $\gamma$ into a path $\mathrm{p}(\gamma)$ and a collection of loops, e.g., see discussion in~\cite{chelkak-izyurov}. Note that there are four types of corners: lying to the right of a nearby vertex $v$, below $v$, to the left of $v$, and upper $v$. For each of these groups, the total turning of the path $\mathrm{p}(\gamma)$ is defined uniquely modulo $2\pi$. Therefore, the discrete spinors introduced above always have purely real values at the first group corners, are collinear to $\lambda:=e^{\frac{\pi}{4}i}$ for the second group, etc. This motivates the following notation:
\begin{itemize}
\item we partition the set $\mathcal{V}_{\mathbb{C}_{\delta}}^{\mathrm{c}}$ of all corners into four subsets $\mathcal{V}_{\Omega_{\delta}}^{1}$, $\mathcal{V}_{\Omega_{\delta}}^{\lambda}$, $\mathcal{V}_{\Omega_{\delta}}^{i}$ and $\mathcal{V}_{\Omega_{\delta}}^{\overline{\lambda}}$ depending on the position of a nearby vertex $v\in\mathcal{V}_{\mathbb{C}_\delta}^{\bullet}$ (to the left, upper, to the right, below) with respect to the corner.
\end{itemize}

\smallbreak
\noindent (iii) \emph{We extend Definition~\ref{def:discrete-spinor} to edge midpoints $z\in\mathcal{V}_{\left[\Omega_{\delta},a_1^{\rightarrow},\dots,a_n\right]}^{\mathrm{m}}$ by adding the factor $(\cos\frac{\pi}{8})^{-1}$ to the formula (\ref{spinor-corners-def})}, with $\#\mathrm{edges}\left(\gamma\right)$ being the number of full edges contained in $\gamma$ and the complex phase $\phi_{a_1,\dots,a_n}(\gamma,z)$ being defined as above. Note that each edge midpoint $z$ can be reached by a path $\mathrm{p}(\gamma)$ from two opposite sides, and both types of configurations are included in the sum. Thus, in this case the argument of the spinor value $F_{\left[\Omega_{\delta},a_1,\dots,a_n\right]}\left(z\right)$ is no longer fixed.

\smallbreak
\noindent (iv) The definition of $F_{\left[\Omega_{\delta},a_1,\dots,a_n\right]}$ is invariant under permutations of $a_2,\dots,a_n$. The reader should always keep in mind, however, that the point $a_1$ plays a special role. The same applies to $\mathcal{A}_\Omega(a_1,\dots,a_n)$ and other related notation below.
\end{rem}

\subsection{S-holomorphicity and boundary conditions}
 \label{sub:s-holomorphicity}
 A version of discrete holomorphicity, the notion of \emph{s-holomorphicity} was introduced in \cite{smirnov-i} together with the nonbranching version of discrete holomorphic observables as a tool to study the critical Ising model on the square lattice. The properties of such functions were further investigated in \cite{chelkak-smirnov-ii} for a more general class of graphs. On the square grid, s-holomorphic functions may be thought of as (more classical) discrete holomorphic functions whose real part is defined on $\mathcal{V}^{1}_{\Omega_\delta}$ and imaginary part on $\mathcal{V}^{i}_{\Omega_\delta}$, extended in a particular way to $\mathcal{V}^{\lambda}_{\Omega_\delta}$,~$\mathcal{V}^{\overline{\lambda}}_{\Omega_\delta}$, and further to $\mathcal{V}^{\mathrm{m}}_{\Omega_\delta}$ (see more details in Section~\ref{sub:s-holomorphicity-proof}). Our definitions resemble those in \cite{smirnov-i}.

\begin{defn}
\label{def: shol}
With each corner $x\in\mathcal{V}_{\mathbb{C}_{\delta}}^{\tau}$ (with $\tau\in\left\{ 1,i,\lambda,\overline{\lambda}\right\} $),
we associate the line $\ell\left(x\right):=\tau\mathbb{R}$ in the complex plane, and denote by $\mathsf{P}_{\ell\left(x\right)}$ the projection onto that line, defined by
\[
\mathsf{P}_{\ell\left(x\right)}\left[w\right]:=\tfrac{1}{2}\left(w+\tau^{2}\overline{w}\right),\quad w\in\C.
\]
We say that a function $F:\mathcal{V}_{\Omega_{\delta}}^{\mathrm{cm}}\to\mathbb{C}$
is \emph{s-holomorphic in $\Omega_\delta$} if for every $x\in\mathcal{V}_{\Omega_{\delta}}^{\mathrm{c}}$
and $z\in\mathcal{V}_{\Omega_{\delta}}^{\mathrm{m}}$ that are adjacent, one has
\[
F\left(x\right)=\mathsf{P}_{\ell\left(x\right)}\left[F\left(z\right)\right].
\]
For functions defined on double covers, we introduce the notion of s-holomorphicity exactly in the same manner.
\end{defn}

The following proposition contains the crucial properties of $F_{[\Omega_\delta,a_1,\dots,a_n]}$ that will allow us to analyze their scaling limits. For $z\in \partial\mathcal{V}_{\Omega_{\delta}}^{\mathrm{m}}$, let $\nu_{\mathrm{out}}(z)$ denote the discrete \emph{``outer normal to the boundary at $z$''}: the edge whose midpoint is $z$, oriented towards the exterior of the domain and viewed as a complex number.
\begin{prop}
\label{prop: s-hol+bc}
 The function $F_{[\Omega_\delta,a_1,\dots,a_n]}$ is s-holomorphic and has (multiplicative) monodromy $-1$ around each of the marked points $a_1,\dots,a_n$, thus being a discrete s-holomorphic spinor on $\mathcal{V}_{\left[\Omega_{\delta},a_1^{\rightarrow},\dots,a_n\right]}^{\mathrm{cm}}$. Also,
\begin{equation}
\label{eq: discrete_bc}
\Im\mathfrak{m}\,\left[F_{\left[\Omega_{\delta},a_1,\dots,a_n\,\right]}\left(z\right){\sqrt{\nu_{\mathrm{out}}(z)}}\right]=0~~\text{for~all~~} z\in\partial\mathcal{V}_{\left[\Omega_{\delta},a_1^{\rightarrow},\dots,a_n\right]}^{\mathrm{m}}.
\end{equation}
\end{prop}
\begin{proof}
We give a proof (based on the standard XOR bijection, cf.~\cite{chelkak-smirnov-ii}) in Section~\ref{sub:s-holomorphicity-proof}.
\end{proof}

\begin{rem}
\label{rem:intF^2}
The boundary conditions (\ref{eq: discrete_bc}) are a priori not robust enough to pass to the scaling limit: even if the limiting domain $\Omega$ has a smooth boundary, the discrete normal $\nu_{\mathrm{out}}\left(z\right)$ can possibly admit only the values $e^{\pm\frac{\pi i}{4}}$ and $e^{\pm\frac{3\pi i}{4}}$, and so does not (pointwise) converge to its continuous counterpart. These conditions become much more tractable, if one finds a way to ``integrate'' the square of $F_{\left[\Omega_{\delta},a_1,\dots,a_n\right]}$: the real part of this primitive (antiderivative) will satisfy Dirichlet boundary conditions on $\partial\Omega$ due to (\ref{eq: discrete_bc}). This approach is not as straightforward as in the continuum, since the square of a discrete holomorphic function is, in general, not discrete holomorphic, and so does not have a well-defined discrete primitive. However, the following remarkable fact has been observed in~\cite{smirnov-i}: one can naturally define the real part of the integral, using the \emph{s-holomorphicity} of discrete observables, which is a stronger
property than the usual discrete holomorphicity. Moreover, a technique developed in \cite{chelkak-smirnov-ii} allows one to treat this real part essentially as if it were a harmonic function, see further details in Section~\ref{sub: riemann-problem}.
\end{rem}

\subsection{\label{sub:discrete-spinors-to-ratios-corr}From discrete spinors
to ratios of correlations}
The following lemma expresses ratios of spin correlations 
in terms of the spinor observables introduced in Section~\ref{sub:discrete-spinors}, providing a crucial ingredient for the proof of Theorems~\ref{thm:log-derivatives} and~\ref{thm:ratio-free-plus}.

\begin{lem}
\label{lem:ratios-disc-magnet-spinor} For any $n=1,2,\dots$, we have
\begin{eqnarray}
\frac{\mathbb{E}_{\Omega_{\delta}}^{+}\left[\sigma_{a_1+2\delta}\sigma_{a_2}\dots\sigma_{a_n}\right]} {\mathbb{E}_{\Omega_{\delta}}^{+}\left[\sigma_{a_1}\dots\sigma_{a_n}\right]} & = & F_{\left[\Omega_{\delta},a_1,\dots,a_n\right]}\left(a_1+\tfrac{3\delta}{2}\right),\label{eq:log-diff-spinor-hor}
\end{eqnarray}
where we take the corner $a_1+\frac{3\delta}{2}$ on the same sheet as the ``source point'' $a_1+\frac{\delta}{2}$.
Moreover, in the case of just two marked points, we also have
\begin{eqnarray}
\frac{\mathbb{E}_{\Omega_{\delta}^\bullet}^{\mathrm{free}}\left[\sigma_{a+\delta}\sigma_{b+\delta}\right]} {\mathbb{E}_{\Omega_{\delta}}^{+}\left[\sigma_{a}\sigma_{b}\right]} & = & \pm i F_{\left[\Omega_{\delta},a,b\right]}\left(b+\tfrac{\delta}{2}\right),\label{eq:spinor-duality}
\end{eqnarray}
where $\mathbb{E}_{\Omega_{\delta}^\bullet}^{\mathrm{free}}$ denotes the expectation for the critical Ising model defined on the vertices of $\Omega_\delta$ (with free boundary conditions outside the set $\mathcal{V}_{\Omega_{\delta}}^{\bullet}$) and the sign $\pm$ depends on the sheet where the corner $b+\frac{\delta}{2}$ is taken.
\end{lem}
\begin{proof}
Recall that
\begin{eqnarray*}
F_{\left[\Omega_{\delta},a_1,\dots, a_n\right]}\left(a_1+\tfrac{3\delta}{2}\right) & = & \frac{\sum_{\gamma\in\mathcal{C}_{\Omega_{\delta}}\left(a_1+\frac{\delta}{2},a_1+\frac{3\delta}{2}\right)} \alpha_{\mathrm{c}}^{\#\mathrm{edges}\left(\gamma\right)}\phi_{a_1,\dots,a_n}\left(\gamma,a_1+\frac{3\delta}{2}\right)} {\mathcal{Z}_{\Omega_{\delta}}^+\left[\sigma_{a_1}\dots\sigma_{a_n}\right]},
\end{eqnarray*}
while, taking into account Remark~\ref{rem:spinor-def}(i),
\begin{eqnarray*}
\frac{\mathbb{E}_{\Omega_{\delta}}^{+}\left[\sigma_{a_1+2\delta}\sigma_{a_2}\dots\sigma_{a_n}\right]} {\mathbb{E}_{\Omega_{\delta}}^{+}\left[\sigma_{a_1}\dots\sigma_{a_n}\right]} & = & \frac{\sum_{\omega\in\mathcal{C}_{\Omega_{\delta}}}\alpha_{\mathrm{c}}^{\#\mathrm{edges} \left(\omega\right)}\left(-1\right)^{\#\mathrm{loops}_{a_1+2\delta,\dots,a_n}(\omega)}} {\mathcal{Z}_{\Omega_{\delta}}^+\left[\sigma_{a_1}\dots\sigma_{a_n}\right]}.
\end{eqnarray*}
There is a simple bijection between the sets $\mathcal{C}_{\Omega_{\delta}}\left(a_1+\frac{\delta}{2},a_1+\frac{3\delta}{2}\right)$ and $\mathcal{C}_{\Omega_{\delta}}$: removing the two corner-edges $(a_1+\frac{\delta}{2},a_1+\delta)$ and $(a_1+\delta,a_1+\frac{3\delta}{2})$ from a given $\gamma$, we obtain a collection of closed loops $\omega(\gamma)\in\mathcal{C}_{\Omega_{\delta}}$ and vice versa. So, it suffices to show that
\[
\phi_{a_1\dots,a_n}\left(\gamma,a_1+\tfrac{3\delta}{2}\right)=\left(-1\right)^{\#\mathrm{loops}_{a_1+2\delta,a_2,\dots,a_n}\left(\omega(\gamma)\right)}.
\]
Let us pick any loop in $\gamma$ and remove it. The left-hand side (respectively, the right-hand side) has changed the sign if and only if there was an odd number of points $a_1,\dots,a_n$ (respectively, $a_1+2\delta,a_2,\dots,a_n$) inside the loop. However, no loop in $\gamma$ separates $a_1$ from $a_1+2\delta$ (such a loop would intersect $\mathrm{p}(\gamma)$), so the two sides can only change sign simultaneously. Thus it is sufficient to consider the case when $\gamma$ is just a single non-self-intersecting path $\mathrm{p}(\gamma)$ running from $a_1+\frac{\delta}{2}$ to $a_1+\frac{3\delta}{2}$, which is treated by the following observations:  $\text{sheet}_{a_1\dots,a_n}(\mathrm{p}(\gamma),a_1+\frac{3\delta}{2})=-1$ if and only if there is an odd number of points $a_1,\dots,a_n$ inside the loop $\omega(\mathrm{p}(\gamma))$, and $\mathrm{wind}(\mathrm{p}(\gamma))=2\pi\mod 4\pi$ if an only if $\omega(\mathrm{p}(\gamma))$ separates $a_1$ from $a_1+2\delta$.

For (\ref{eq:spinor-duality}), the Kramers-Wannier duality (e.g., see \cite[Chapter 1]{palmer}) implies
\[
\frac{\mathbb{E}_{\Omega_{\delta}^\bullet}^{\mathrm{free}}\left[\sigma_{a+\delta}\sigma_{b+\delta}\right]} {\mathbb{E}_{\Omega_{\delta}}^{+}\left[\sigma_{a}\sigma_{b}\right]}= \frac{\sum_{\omega\in\mathcal{C}_{\Omega_{\delta}}(a+{\frac{\delta}{2}},b+\frac{\delta}{2})}\alpha_{\mathrm{c}}^{\#\mathrm{edges}\left(\omega\right)}} {\mathcal{Z}_{\Omega_\delta}^+[\sigma_a\sigma_b]},
\]
hence it is sufficient to prove that the (purely imaginary) number $\phi_{a,b}\left(\gamma,b+\frac{\delta}{2}\right)$ does not depend on $\gamma$. We have $\#\mathrm{loops}_{a,b}(\gamma\setminus\mathrm{p}(\gamma))=0$, since any loop in $\gamma$ either surrounds both $a,b$ or none of them (otherwise it would intersect the path $\mathrm{p(\gamma)}$ joining $a+\frac{\delta}{2}$ and $b+\frac{\delta}{2}$). Further, let $\pi_{ba}^\circ=b\sim v_1\sim \dots \sim v_m\sim a$ be some simple lattice path with $v_j\in \mathcal{V}^{\circ}_{\Omega_\delta}$. Then, $\mathrm{p}(\gamma)\cup(b+\frac{\delta}{2},b)\cup\pi_{ab}^\circ\cup(a,a+\frac{\delta}{2})$ is a loop, which we denote by $l(\gamma)$. Let $n(\gamma)$ be the number of self-intersections of $l(\gamma)$ (in order words, the number of intersections of $\mathrm{p}(\gamma)$ with $\pi_{ba}^\circ$). Since $\exp[-\frac{i}{2}\mathrm{wind}(l(\gamma))]=(-1)^{n(\gamma)+1}$ (this is true for any closed loop), we see that
\[
\exp[-\tfrac{i}{2}\mathrm{wind}(\mathrm{p(\gamma)})]=(-1)^{n(\gamma)+1} \cdot \exp[\tfrac{i}{2} \mathrm{wind}((b+\tfrac{\delta}{2},b)\cup\pi_{ba}^\circ\cup(a,a+\tfrac{\delta}{2}))],
\]
where the second factor does not depend on $\gamma$. But we may also view $\pi_{ba}^\circ$ as a cut defining a sheet of the double cover $[\Omega,a,b]$, meaning that $\text{sheet}_{a,b}(p(\gamma),b+\frac{\delta}{2})=(-1)^{n(\gamma)}$, which proves the desired result.
\end{proof}

\begin{rem}
\label{rem:ratios-disc-magnet-spinor} (i) Similarly to (\ref{eq:log-diff-spinor-hor}), one can check that
\begin{eqnarray}
\label{eq:nearby-ratios-disc-magnet-spinor}
\frac{\mathbb{E}_{\Omega_{\delta}}^{+}\left[\sigma_{a_1+\left(1\pm i\right)\delta}\sigma_{a_2}\dots\sigma_{a_n}\right]} {\mathbb{E}_{\Omega_{\delta}}^{+}\left[\sigma_{a_1}\dots\sigma_{a_n}\right]} & = & e^{\pm\frac{\pi i}{4}}\: F_{\left[\Omega_{\delta},a_1,\dots,a_n\right]}\left(a_1+\left(1\!\pm\!\tfrac{i}{2}\right)\delta\right). \label{eq:log-diff-spinor-diag}
\end{eqnarray}
The proof boils down to the identity
\[
e^{\pm\frac{\pi i}{4}}\phi_{a_1,\dots,a_n}\left(\gamma,a_1+\left(1\!\pm\!\tfrac{i}{2}\right)\delta\right) =
\left(-1\right)^{\#\mathrm{loops}_{a_1+\delta(1\pm i),a_2,\dots,a_n}\left(\omega(\gamma)\right)}
\]
for the natural bijection $\gamma\mapsto\omega(\gamma)$ removing the two corner-edges $(a_1+\frac{\delta}{2},a_1+\delta)$ and $(a_1+\delta,a_1+(1\pm\frac{i}{2})\delta)$ from a given $\gamma\in\mathcal{C}_{\Omega_{\delta}}\left(a_1+\frac{\delta}{2},a_1+(1\pm\frac{i}{2})\delta\right)$, which we leave to the reader.

\noindent (ii) The identity (\ref{eq:spinor-duality}) can be extended to the case of $2n$ marked points, see \cite[Proposition~5.6]{chelkak-izyurov}, thus allowing one to treat $2n$-point correlation functions with free boundary conditions.
\end{rem}

\subsection{\label{sub:cont-spinors}Continuous spinors} In this Section, we introduce the continuous counterparts of the discrete spinor observables defined in Section~\ref{sub:discrete-spinors}: the continuous holomorphic spinors $f_{[\Omega,a_1,\dots,a_n]}$. We define them as
solutions to the conformally covariant Riemann boundary value problem (\ref{def:spinor-bdry})\,--\,(\ref{def:spinor-1}), which is a continuous analogue of the corresponding discrete boundary value problem (see Remark~\ref{rem:cont-spinor-motivation} below).

\begin{defn}
Let $\Omega$ be a bounded simply connected domain with smooth boundary, and $a_1,\dots,a_n\in\Omega$.
We define $f_{\left[\Omega,a_1,\dots,a_n\right]}$ to be the (unique) holomorphic spinor on $[\Omega,a_1,\dots,a_n]$, branching around each of $a_1,\dots,a_n$ and satisfying the following conditions:
\begin{eqnarray}
\Im\mathfrak{m}\,\left[f_{[\Omega,a_1,\dots,a_n]}(z)\sqrt{\nu_\mathrm{out}(z)}\,\right]&=&0,~~z\in\partial \Omega;\label{def:spinor-bdry}\\
\lim\limits_{z\rightarrow a_1}\sqrt{z-a_1}\cdot f_{[\Omega,a_1,\dots,a_n]}(z)&=&1; \label{def:spinor-2}\\
 \lim\limits_{z\rightarrow a_k}\sqrt{z-a_k}\cdot f_{[\Omega,a_1,\dots,a_n]}(z)&\in&i\R, \quad k=2,\dots,n,\label{def:spinor-1}
\end{eqnarray}
where $\nu_\mathrm{out}(z)$ denotes the outer normal to the boundary of $\Omega$ at $z$.
\label{def:spinor-bounded}
\end{defn}

\begin{lem}
\label{lem: obs_uniqueness}
\label{lem: spinor_rough}
The boundary value problem (\ref{def:spinor-bdry})\,--\,(\ref{def:spinor-1}) has a unique solution. If $\varphi:\Omega\rightarrow \Omega'$ is a conformal mapping, then one has
\begin{equation}
\label{eq: obs_covariance}
f_{[\Omega,a_1\dots,a_n]}(z)=f_{[\Omega',\varphi(a_1)\dots,\varphi(a_n)]}(\varphi(z))\cdot \left(\varphi'(z)\right)^{1/2}.
\end{equation}
\end{lem}
\begin{rem}
\label{rem: sp_rough_domain}
We use this covariance property as a \emph{definition} of the continuous spinor in an arbitrary simply connected domain $\Omega$, taking $\Omega'$ to be any smooth bounded domain.
\end{rem}
\begin{proof}
If $f_1,f_2$ both satisfy (\ref{def:spinor-bdry})\,--\,(\ref{def:spinor-1}), then the spinor $f_1-f_2$ satisfies (\ref{def:spinor-bdry}) and (\ref{def:spinor-1}), while $\lim_{z\rightarrow a_1}\sqrt{z-a_1}\cdot (f_1(z)-f_2(z))=0$. Applying the Cauchy residue theorem to the \emph{single-valued} function $(f_1(z)-f_2(z))^2$, one arrives at
\begin{equation}
\label{(f1-f2)^2-int}
0\le i^{-1}\int_{\partial\Omega}(f_1(z)-f_2(z))^2dz=2\pi \sum_{k=2}^n \lim\limits_{z\rightarrow a_k}(z-a_k)(f_1(z)-f_2(z))^2\le 0,
\end{equation}
where the first inequality easily follows from (\ref{def:spinor-bdry}) and the second from (\ref{def:spinor-1}). Hence, $f_1\equiv f_2$. Moreover, the conformal covariance property (\ref{eq: obs_covariance}) now follows from an easy observation that its right-hand side satisfies all the conditions (\ref{def:spinor-bdry})\,--\,(\ref{def:spinor-1}).

To prove existence, consider the case $\Omega=\bH$ and observe that if $Q(z)$ is a polynomial of degree $n-1$ with real coefficients, then
\begin{equation}
\label{eq: f_Q}
f_Q(z):=e^{\frac{\pi i}{4}}\cdot \frac{Q(z)}{\sqrt{(z-a_1)(z-\overline{a}_1)\dots (z-a_n)(z-\overline{a}_n)}},
\end{equation}
is a holomorphic spinor on $[\bH,a_1,\dots,a_n]$ satisfying (\ref{def:spinor-bdry}) and the additional regularity condition $f_Q(z)=O(|z|^{-1})$ at infinity. Note that (\ref{(f1-f2)^2-int}) shows that, if $f_Q$  satisfies (\ref{def:spinor-1}), then $\Im\mathfrak{m} \lim\limits_{z\rightarrow a_1}\sqrt{z-a_1}\cdot f_{Q}(z)=0$. Therefore,  conditions (\ref{def:spinor-2}), (\ref{def:spinor-1}) give rise to $n$ (real) linear equations on $n$ unknown coefficiens of $Q$. By the above argument, the homogeneous counterpart of this linear system has no nontrivial solutions; thus the system is non-degenerate. The solution for a given bounded smooth domain $\Omega$ is then constructed by (\ref{eq: obs_covariance}) with $\Omega'=\bH$, the decay of $f_Q(z)$ at infinity ensures that the right-hand side is bounded near $\varphi^{-1}(\infty)$.
\end{proof}

\begin{rem}
 \label{rem:cont-spinor-motivation} The first condition (\ref{def:spinor-bdry}) in Definition~\ref{def:spinor-bounded} is a natural counterpart of (\ref{eq: discrete_bc}). The third condition~(\ref{def:spinor-1}) comes from the following observation: a discrete primitive $H_{[\Omega,a_1,\dots,a_n]}$ of the ``discrete differential form'' $\Re\mathfrak{e}\,[F_{[\Omega_\delta,a_1,\dots,a_n]}^2dz]$ (which may be defined due to the s-holomorphicity property of the discrete observable, see Remark~\ref{rem:intF^2}) remains bounded from below near the branching points $a_2,\dots,a_n$ as $\delta\to 0$. Thus, we impose the same condition for the scaling limits, which means that $\Re\mathfrak{e}\,[\int f_{\left[\Omega,a_1,\dots,a_n\right]}^2dz]$ should behave like $c_k\log|z-a_k|$ for some \emph{nonpositive} $c_k$ as $z\to a_k$, $k=2,\dots,n$, implying~(\ref{def:spinor-1}). The second condition (\ref{def:spinor-2}) which fixes the behavior of $f_{\left[\Omega,a_1,\dots,a_n\right]}$ near the ``source point'' $a_1$ is the
most delicate one and will be clarified later on (see Section~\ref{sub:full-plane}, particularly Lemma~\ref{lem: remove_sing}). Note that it is sufficient to assume that $f_{\left[\Omega,a_1,\dots,a_n\right]}$ does not blow up faster than $1/\sqrt{z-a_1}$ at $a_1$. Indeed, in this case the argument similar to (\ref{(f1-f2)^2-int}) shows that $\lim\limits_{z\rightarrow a_1}(z-a_1)(f_{[\Omega,a_1,\dots,a_n]}(z))^2>0$ and the rest is just a proper choice of the normalization.
\end{rem}

Let us now introduce the quantities that play a central role in our computations of scaling limits of spin correlations, as they will turn out to be the limits of discrete logarithmic derivatives in Theorem~\ref{thm:log-derivatives}.
\begin{defn}
We define the complex number $\mathcal{A}_\Omega(a_1,\dots,a_n)=\mathcal{A}_{[\Omega,a_1,\dots,a_n]}$ as the coefficient in the expansion
\begin{equation}
\label{eq: expansion-a}
f_{[\Omega,a_1,\dots,a_n]}=\frac{1}{\sqrt{z-a_1}}+2\mathcal{A}_{[\Omega,a_1,\dots,a_n]}\sqrt{z-a_1} + O(|z-a_1|^{3/2})
\end{equation}
of $f_{[\Omega,a_1,\dots,a_n]}$ near the point $a_1$. In the special case $n=2$, we also define the quantity $\mathcal{B}_\Omega(a,b)=\mathcal{B}_{[\Omega,a,b]}>0$ as the coefficient in the  expansion of $f_{[\Omega,a,b]}$ near $b$:
\begin{equation}
\label{eq: expansion-b}
f_{[\Omega,a,b]}=\pm \frac{i\mathcal{B}_{[\Omega,a,b]}}{\sqrt{z-b}} + O(|z-b|^{1/2}),
\end{equation}
where the sign $\pm$ depends on the sheet of $[\Omega,a,b]$.
\end{defn}
\begin{rem}
Note that the covariance rule (\ref{A-covariance}) for $\mathcal{A}_{[\Omega,a_1,\dots,a_n]}$ directly follows from the conformal covariance (\ref{eq: obs_covariance}) of spinor observables: if $\varphi:\Omega\to\Omega'$ is a conformal mapping, then one has
\begin{align*}
&f_{\left[\Omega,a_1,\dots,a_n\right]}\left(z\right) ~=~ (\varphi'\left(z\right))^{1/2}f_{\left[\Omega',\varphi(a_1),\dots,\varphi(a_n)\right]}(\varphi(z)) \\
&\phantom{f_{\left[\Omega,a_1,\dots,a_n\right]}\left(z\right)} ~=~ \left[\frac{\varphi'\left(z\right)}{\varphi\left(z\right)-\varphi\left(a_1\right)}\right]^\frac{1}{2} \cdot~\left[1+2\mathcal{A}_{\varphi}\cdot(\varphi\left(z\right)-\varphi\left(a_1\right))+\dots\right]\\
 & =~
 \biggl[\frac{1+\frac{\varphi''(a_1)}{\varphi'(a_1)}(z-a_1)+\dots}{(z-a_1)(1+\frac{\varphi''(a_1)}{2\varphi'(a_1)}(z-a_1)+\dots)}\biggr]^\frac{1}{2} \cdot~\left[1+2\mathcal{A}_{\varphi}\cdot\varphi'(a_1)(z-a_1)+\dots\right] \\
 & =~
\frac{1}{\sqrt{z-a_1}}\left[1+2\left(\mathcal{A}_\varphi\cdot\varphi'\left(a_1\right)+\frac{1}{8}\frac{\varphi''\left(a_1\right)}{\varphi'\left(a_1\right)}\right)(z-a_1)+\dots\right],
\end{align*}
where $\mathcal{A}_{\varphi}=\mathcal{A}_{\left[\Omega',\varphi(a_1),\dots,\varphi(a_n)\right]}$. Similar arguments show that the coefficient $\mathcal{B}_{[\Omega,a,b]}$ is conformally invariant, i.e.,
\begin{equation}
\label{B-invariance}
\mathcal{B}_\Omega(a,b)=\mathcal{B}_{\Omega'}(\varphi(a),\varphi(b)).
\end{equation}
We defer the further analysis of $\mathcal{A}_\Omega$ and $\mathcal{B}_\Omega$ until Lemma \ref{lem: A_explicit} and Appendix, where we compute those quantities explicitly for $n\le 2$ and $n>2$, respectively.
\end{rem}

\subsection{Convergence of the spinors}\label{sub:convergence-of-spinors}
The main purpose of this section is to derive Theorem~\ref{thm:log-derivatives} and Theorem~\ref{thm:ratio-free-plus} from the convergence results for the discrete spinor observables $F_{\left[\Omega_{\delta},a_1,\dots,a_n\right]}$ which are formulated in Theorems~\ref{thm:obs-away},~\ref{thm:localization-near-a} and~\ref{thm:localization-near-b} below. The proofs of those are given in Section~\ref{sec:analysis-of-spinors}, which is the most technical part of our paper. We use the following definitions concerning convergence of discrete s-holomorphic functions:
\begin{itemize}
\item we say that a family $\Omega_{\delta}$ of discrete domains approximates a continuous domain $\Omega\subset\mathbb{C}$ as $\delta\to 0$, if $\partial\Omega_{\delta}$ converges to $\partial\Omega$ in the Hausdorff sense (note that our proofs can be easily generalized for the Carath\'eodory convergence of planar domains which is weaker than the Hausdorff one used in this paper for simplicity);
\item we say that an s-holomorphic function (or a spinor) $F_\delta$ defined in $\Omega_\delta^\mathrm{cm}$ (or its double cover) tends to a holomorphic function (or a spinor) $f$ as $\delta\to 0$, if the ``mid-edge values'' $F_\delta\big|_{\Omega_\delta^\mathrm{m}}$ approximate the values of $f$, while the ``corner values'' $F\big|_{\Omega_\delta^\tau}$, $\tau\in\{1,\lambda,i,\overline{\lambda}\}$, tend to the projections of $f$ onto the corresponding lines $\tau\mathbb{R}$ (see Definition~\ref{def: shol});
\item we say that a convergence of discrete functions $F_\delta(z)=F_\delta(z;a_1,a_2,\dots)$ to $f(z)=f(z;a_1,a_2,\dots)$ is uniform on some compact set, iff the differences $|F_\delta(z;a_1,a_2,\dots)-f(z;a_1,a_2,\dots)|$ are uniformly small as $\delta\to 0$, when we interpret lattice vertices (or mid-edges, corners, etc) $z,a_1,a_2,\dots$  as the corresponding complex points when we plug them into $f$.
\end{itemize}

The crucial ingredient of our proofs is the interplay between (a) the values of discrete spinor observables near their branching points, which are related to the ratios of spin correlations by Lemma~\ref{lem:ratios-disc-magnet-spinor}, and (b) the mid-range behavior of these observables, which can be further related to the asymptotics expansions of their scaling limits (\ref{eq: expansion-a}),~(\ref{eq: expansion-b}).

As a main tool to relate (a) and (b), we use a \emph{full-plane version $F_{\left[\mathbb{C}_{\delta},a\right]}$ of the spinor observable} (since we are interested in local considerations, it is sufficient to stick to the case of one marked point). Though it could be constructed as an infinite-volume limit of the finite-domain observables, we prefer a more explicit strategy, which is outlined after the following lemma claiming the existence of $F_{\left[\mathbb{C}_{\delta},a\right]}$.

\begin{lem}
\label{lem: def_f_C}
For $a\in\mathcal{V}_{\mathbb{\C}_\delta}^\circ$, there exists a (unique) s-holomorphic spinor $F_{\left[\mathbb{C}_{\delta},a\right]}:\mathcal{V}^{\mathrm{cm}}_{[\C_\delta,a^\rightarrow]}\rightarrow \C$ such that $F_{\left[\mathbb{C}_{\delta},a\right]}(a+\frac{3\delta}{2})=1$ and $F_{\left[\mathbb{C}_{\delta},a\right]}(z)=o(1)$ as $z\to\infty$. Moreover,
\begin{eqnarray}
\label{F_C-convergence}
\frac{1}{\vartheta(\delta)}F_{\left[\mathbb{C}_{\delta},a\right]}(z) & \underset{\delta\to0}{\longrightarrow} &  \frac{1}{\sqrt{z-a}}~=:f_{[\mathbb{C},a]}(z),
\end{eqnarray}
uniformly on compact subsets of $\mathbb{C}\setminus \{a\}$, where 
$\vartheta(\delta)$ is defined as
\begin{equation}
\label{vartheta-def}
\vartheta(\delta):=F_{[\C_\delta,a]}\left(a+\tfrac{3\delta}{2}+2\delta\lfloor\tfrac{1}{2\delta}\rfloor\right).
\end{equation}
\end{lem}

\begin{proof} The detailed proof is given in Section~\ref{sub:full-plane}. First, we define the (real) values of $F_{[\C_\delta,a]}$ on ${\mathcal{V}_{\mathbb{\C}_\delta}^1}$ as the discrete harmonic measure of the tip point $a+\frac{3\delta}{2}$ in the slit discrete plane ${\mathcal{V}_{\mathbb{\C}_\delta}^1}\setminus\{x+a:x\le 0\}$. This definition is motivated by the following observation: in the continuous setup, the function $\Re\mathfrak{e}\,[{1}/{\sqrt{z-a}}\,]$ can be viewed as the properly normalized harmonic measure of the (small neighborhood of) tip $a$ in the slit plane $\mathbb{C}\setminus\{x+a:x\le 0\}$. Second, we extend $F_{[\C_\delta,a]}$ to $\mathcal{V}_{\mathbb{\C}_\delta}^i$ by harmonic conjugation, then by symmetry to another sheet of $[\C_\delta,a]$, and eventually as an s-holomorphic function to $\mathcal{V}^{\lambda}_{\mathbb{C}_\delta},\mathcal{V}^{\overline{\lambda}}_{\mathbb{C}_\delta}$ and $\mathcal{V}^\mathrm{m}_{\mathbb{C}_\delta}$. The convergence~(\ref{F_C-convergence})
follows from known results on convergence of harmonic measures (e.g., see~\cite{chelkak-smirnov-i}).
\end{proof}

\begin{rem}
The normalizing factor $\vartheta(\delta)$ is essentially the value of $F_{[\C_\delta,a]}$ at $a+1$. In Section~\ref{sub:full-plane} we show that
\begin{equation}
\label{eq: vartheta-estimate}
\mathrm{C}_{-}\sqrt{\delta}\leq\vartheta\left(\delta\right)\leq\mathrm{C}_{+}\sqrt{\delta},
\end{equation}
where $\mathrm{C}_\pm>0$ are some absolute constants. Note that one can compute the limit $\lim_{\delta\to 0}\vartheta(\delta)/\sqrt{\delta}$ using the recent work of Dub\'edat \cite{dubedat-i}, but we do not need this sharp result.
\end{rem}

We further use the normalizing factors $\vartheta(\delta)$ introduced above in order to formulate the following convergence theorem for discrete spinor observables away from $\partial\Omega$ and $a_1,\dots,a_n$:
\begin{thm}
\label{thm:obs-away}
Let discrete simply connected domains $\Omega_{\delta}$ approximate a bounded simply connected domain $\Omega$ as $\delta\to 0$. Then, for any $\epsilon>0$ and any $n=1,2,\dots$, we have
\begin{eqnarray*}
\frac{1}{\vartheta\left(\delta\right)}F_{\left[\Omega_{\delta},a_1,\dots,a_n\right]}(z) & \underset{\delta\to0}{\longrightarrow} & f_{\left[\Omega,a_1,\dots,a_n\right]}(z),
\end{eqnarray*}
uniformly over all $a_1,\dots,a_n\in\mathcal{V}_{\Omega_{\delta}}^{\circ}$ and $z\in\mathcal{V}_{\Omega_{\delta}}^{\mathrm{m}}$ which are at the distance at least $\epsilon$ from $\partial\Omega$ and from each other.
\end{thm}
\begin{proof}
See Section~\ref{sub:convergence-away-sing}.
\end{proof}

Since we are interested in the coefficient $\mathcal{A}_{[\Omega,a_1,\dots,a_n]}$ in front of the term $\sqrt{z-a_1}$ in (\ref{eq: expansion-a}), along with the discrete analogue $F_{[\mathbb{C}_\delta,a]}$ of the function $1/\sqrt{z-a}$ given by Lemma~\ref{lem: def_f_C}, we also need a discrete counterpart $G_{[\mathbb{C}_\delta,a]}$ of the function $\sqrt{z-a}$. We construct $G_{[\mathbb{C}_\delta,a]}$ by ``discrete integration'' of $F_{[\mathbb{C}_\delta,a]}$, just mimicking the continuous setup. It is sufficient to define $G_{[\mathbb{C}_\delta,a]}$ on $\mathcal{V}_{\mathbb{C}_\delta}^1$ only, as Lemma~\ref{lem:ratios-disc-magnet-spinor} deals with the (real) values of discrete observables at the point $a_1+\tfrac{3\delta}{2}\in \mathcal{V}_{\mathbb{C}_\delta}^1$.

\begin{lem}
\label{def_r_delta}
For $a\in\mathcal{V}_{\mathbb{\C}_\delta}^\circ$, there exists a (unique) discrete harmonic spinor $G_{[\mathbb{C}_\delta,a]}:\mathcal{V}^{1}_{[\C_\delta,a]}\to \R$ such that $G_{[\mathbb{C}_\delta,a]}$ vanishes on the half-line $\{x+a:x\le 0\}$, $G_{[\mathbb{C}_\delta,a]}(a+\tfrac{3\delta}{2})=\delta$, and $G_{[\mathbb{C}_\delta,a]}(z)=O(|z-a|^{1/2})$ as $z\to\infty$. Moreover,
\begin{eqnarray}
\label{G_C-convergence}
\frac{1}{\vartheta(\delta)}G_{[\mathbb{C}_\delta,a]}(z)& \underset{\delta\to0}{\longrightarrow} & \Re\mathfrak{e}\,\sqrt{z-a}~=:g_{[\mathbb{C},a]}(z),
\end{eqnarray}
uniformly on compact subsets of $\mathbb{C}\setminus \{a\}$.
\end{lem}
\begin{proof}
For $z\in\mathcal{V}_{[\mathbb{\C}_\delta,a]}^1$, we set
$G_{[\mathbb{C}_\delta,a]}(z):=\delta\cdot\sum_{j=0}^{\infty} F_{[\mathbb{C}_\delta,a]}(z-2j\delta)$. Certainly, one should check the convergence of this series and the harmonicity on the half-line $\{x+a:x\ge 0\}$, see further details in Section~\ref{sub:full-plane}.
\end{proof}

In the continuum limit, the leading term in the expansion of $f_{[\Omega,a_1,\dots,a_n]}-f_{[\C,a_1]}$ near $a_1$ is given by $2\mathcal{A}_{[\Omega,a_1,\dots,a_n]}\sqrt{z-a_1}$. It is hence plausible to believe that the same holds true for the discrete spinors, and one has
\[
(F_{[\Omega_\delta,a_1,\dots,a_n]}-F_{[\C_\delta,a_1]})\left(a_1+\tfrac{3\delta}{2}\right) ~\approx~ 2\Re\mathfrak{e}\,\mathcal{A}_{[\Omega_\delta,a_1,\dots,a_n]}\cdot G_{[\C_\delta,a_1]}\left(a_1+\tfrac{3\delta}{2}\right)
\]
up to higher-order terms (the real part appears due to discrete complex analysis subtleties, as real and imaginary parts of s-holomorphic functions are defined on different lattices, and $a_1+\frac{3\delta}{2}\in\mathcal{V}^{1}_{\Omega_\delta}$).
We justify this heuristics in
\begin{thm}
Under conditions of Theorem~\ref{thm:obs-away}, we have
\begin{eqnarray}
F_{\left[\Omega_{\delta},a_1,\dots,a_n\right]}\left(a_1+\tfrac{3\delta}{2}\right)-1- 2\Re\mathfrak{e}\,\mathcal{A}_{\left[\Omega,a_1,\dots,a_n\right]}\cdot\delta & = & o\left(\delta\right)
\label{eq:localization-near-a}
\end{eqnarray}
as $\delta\to 0$, uniformly over all $a_1,\dots,a_n\in\mathcal{V}_{\Omega_{\delta}}^{\circ}$ which are at the distance at least $\epsilon$ from $\partial\Omega$ and from each other.
\label{thm:localization-near-a}
\end{thm}
\begin{proof}
See Section~\ref{sub:convergence-near-sing}.
\end{proof}

\begin{rem}
\label{rem:ratios-spinspin-nearby}
In the proof of Theorem~\ref{thm:localization-near-a} we show that $F_{\left[\Omega_{\delta},a_1,\dots,a_n\right]}$ and $F_{\left[\C_\delta,a_1\right]}$ are $\delta$-close to each other at all points around $a_1$, in particular at $a_1+(1\pm\frac{i}{2})\delta$. Together with (\ref{eq:nearby-ratios-disc-magnet-spinor}), this implies
\[
\frac{\mathbb{E}_{\Omega_{\delta}}^{+}\left[\sigma_{a_1+(1\pm i)\delta}\sigma_{a_2}\dots\sigma_{a_n}\right]} {\mathbb{E}_{\Omega_{\delta}}^{+}\left[\sigma_{a_1}\dots\sigma_{a_n}\right]}=1+O(\delta) 
\]
as $\delta\to 0$, uniformly over all $a_1,\dots,a_n\in\mathcal{V}_{\Omega_{\delta}}^{\circ}$ which are at the distance at least $\epsilon$ from $\partial\Omega$ and from each other.
\end{rem}

The similar analysis for the quantity $\mathcal{B}_{[\Omega,a,b]}$ defined by the expansion (\ref{eq: expansion-b}) is simpler since we need to match the first-order coefficients instead of the second-order ones. The result is given by

\begin{thm}
For $n=2$, under conditions of Theorem~\ref{thm:obs-away}, we have
\begin{eqnarray}
F_{\left[\Omega_{\delta},a,b\right]}\left(b+\tfrac{\delta}{2}\right)\pm i\mathcal{B}_{[\Omega,a,b]} & = & o\left(1\right)
\label{eq:localization-near-b}
\end{eqnarray}
as $\delta\to 0$, uniformly over all $a,b\in\mathcal{V}_{\Omega_{\delta}}^{\circ}$ which are at the distance at least $\epsilon$ from $\partial\Omega$ and from each other, where the sign $\pm$ depends on the sheet of $[\Omega,a,b]$.
\label{thm:localization-near-b}
\end{thm}
\begin{proof}
See Section~\ref{sub:convergence-near-sing}.
\end{proof}


\begin{proof}[\textbf{Proof of Theorems~\ref{thm:log-derivatives} and~\ref{thm:ratio-free-plus}}]
Due to Lemma~\ref{lem:ratios-disc-magnet-spinor}, asymptotics (\ref{eq:log-der-x}) is a reformulation of (\ref{eq:localization-near-a}), while Theorem~\ref{thm:ratio-free-plus} is equivalent to Theorem~\ref{thm:localization-near-b} (the sign $\pm$ is fixed due to the positivity of spin-spin correlations). To check (\ref{eq:log-der-y}), we rotate our domain $\Omega$  around $a$ by 90\textdegree{} clockwise. According to conformal covariance rule (\ref{A-covariance}), the coefficient $\mathcal{A}$ multiplies by $i$, so the desired result follows from $\Re\mathfrak{e}\,[i\mathcal{A}]=-\Im\mathfrak{m}\,\mathcal{A}$.
\end{proof}

\subsection{Integrating logarithmic derivatives and explicit formulae}\label{sub:expl_formulae}
We start this Section by stating explicit formulae for $\mathcal{A}_{[\Omega,\Ak]}$ and $\mathcal{B}_{[\Omega,a,b]}$.
\begin{lem}
\label{lem: A_explicit}
One has
 \begin{eqnarray}
 \label{eq: A_expl}
  \mathcal{A}_{[\Omega,a_1,\dots,a_n]}&=&2\partial_{a_1}\log \langle\sigma_{a_1}\dots\sigma_{a_n}\rangle^{+}_{\Omega};\\
  \label{eq: B_expl}
  \mathcal{B}_{[\Omega,a,b]}&=&\frac{\langle\sigma_a\sigma_b\rangle_\Omega^\mathrm{free}}{\langle\sigma_a\sigma_b\rangle_\Omega^+},
  \end{eqnarray}
 where $\partial_{a_1}=\tfrac{1}{2}(\partial_{x_1}-i\partial_{y_1})$ if $a_1=x_1+iy_1$, and
the quantities $\langle\sigma_{a_1}\dots\sigma_{a_n}\rangle_{\Omega}$ are defined by the conformal covariance rule (\ref{conf-cov}) and explicit formulae (\ref{Expl_free}), (\ref{12ptFcts}).
\end{lem}
\begin{proof}
Note that (\ref{conf-cov}) readily implies that $2\partial_{a_1}\log \langle\sigma_{a_1}\dots\sigma_{a_n}\rangle^{+}_{\Omega}$ satisfies the same covariance rule (\ref{A-covariance}) as $\mathcal{A}_{[\Omega,a_1,\dots,a_n]}$, and both sides in (\ref{eq: B_expl}) are conformally invariant. Therefore, it suffices to check (\ref{eq: A_expl}), (\ref{eq: B_expl}) for $\Omega=\bH$. This amounts to computing the continuous spinors $f_{[\bH,a_1,\dots,a_n]}$, which can be done by solving a linear system for the coefficients of polynomials $Q$ introduced in the proof of Lemma \ref{lem: spinor_rough}.

Let us first illustrate this procedure for $n=1,2$. It is straightforward to check that the spinors
\begin{equation*}
f_{[\mathbb{H},a]}(z)=\frac{(2i\Im\mathfrak{m}\, a)^\frac{1}{2}}{\sqrt{(z-a)(z-\overline{a})}},
\end{equation*}
\begin{equation*}
f_{[\mathbb{H},a,b]}(z)=\frac{(2i\Im\mathfrak{m}\, a)^\frac{1}{2}}{|b-\overline{a}|+|b-a|} \cdot \frac{[(\overline{b}-\overline{a})(\overline{b}-a)]^\frac{1}{2}(z-b)+[(b-a)(b-\overline{a})]^\frac{1}{2}(z-\overline{b})} {[(z-a)(z-\overline{a})(z-b)(z-\overline{b})]^{1/2}},
\end{equation*}
satisfy (\ref{def:spinor-bdry})\,--\,(\ref{def:spinor-1}). Working out the expansions (\ref{eq: expansion-a}) and (\ref{eq: expansion-b}), one arrives at
\[
\mathcal{A}_\mathbb{H}(a)=-\frac{1}{8i\Im\mathfrak{m}\, a},
\]
\[
\mathcal{A}_\mathbb{H}(a,b)=-\frac{1}{8i\Im\mathfrak{m}\, a} +\frac{|b-\overline{a}|-|b-a|}{4(|b-\overline{a}|+|b-a|)}\left(\frac{1}{b-a}-\frac{1}{\overline{b}-a}\right)
\]
and
\[
\mathcal{B}_\mathbb{H}(a,b)=\frac{(4\Im\mathfrak{m}\, a\Im\mathfrak{m}\, b)^\frac{1}{2}}{|b-\overline{a}|+|b-a|}
= \frac{\langle\sigma_a\sigma_b\rangle_\mathbb{H}^\mathrm{free}}{\langle\sigma_a\sigma_b\rangle_\mathbb{H}^+},
\]
proving (\ref{eq: B_expl}). 
Taking the logarithmic derivatives with respect to $a$ of the explicit expressions for $\langle\sigma_a\rangle_\mathbb{H}^+$ and $\langle\sigma_a\sigma_b\rangle_\mathbb{H}^+$ given by (\ref{Expl_free}), one arrives at (\ref{eq: A_expl}) for $n=1,2$. The proof for higher $n$ is somewhat technical, and we defer it to Appendix.
\end{proof}

The following Proposition, which is a rephrasing of Corollary~\ref{cor:conv-up-to-normalization}, now follows easily from Theorem~\ref{thm:log-derivatives} via integration with respect to positions of points.

\begin{prop}
As $\Omega_\delta$ approximates $\Omega$, one has
\begin{equation}
\label{eq: conv_ratios}
 \frac{\E^{+}_{\Omega_{\delta}}[\sigma_{b_1}\dots\sigma_{b_n}]}{\E^{+}_{\Omega_{\delta}}[\sigma_{a_1}\dots\sigma_{a_n}]}
~\underset{\delta\to0}{\longrightarrow}~
\frac{\langle\sigma_{b_1}\dots\sigma_{b_n}\rangle^{+}_{\Omega}}{\langle\sigma_{a_1}\dots\sigma_{a_n}\rangle^{+}_{\Omega}}
\end{equation}
uniformly with respect to all $a_k$ at distance at least $\epsilon$ from $\partial\Omega$ and each other, and $b_k$ obeying the same condition.
\label{prop: conv_ratios}
\end{prop}

\begin{proof}
Color the faces of $\Omega_{\delta}$ black and white, in a chessboard
fashion. By Remark \ref{rem:ratios-spinspin-nearby}, the ratio of spin-spin correlations
at two adjacent spins tends to $1$, uniformly away from the boundary,
so we can assume that all $a_k,b_k$ are colored white. Let $a'_1\in\Omega$ be such that $[a_1,a'_1]$ is a horizontal segment contained in $\Omega$ and disjoint with $a_2,\dots,a_n$. Denote by \mbox{$a_1=v_{1}\sim\dots\sim v_{m_{\delta}}=a_1'$} a straight horizontal lattice path approximating $[a_1,a'_1]$. Then, by Theorem~\ref{thm:log-derivatives}, one has
\begin{align*}
\log\frac{\mathbb{E}_{\Omega_{\delta}}^{+}\left[\sigma_{v_{j+1}}\sigma_{a_2}\dots\sigma_{a_n}\right]} {\mathbb{E}_{\Omega_{\delta}}^{+}\left[\sigma_{v_{j}}\sigma_{a_2}\dots\sigma_{a_n}\right]} & = \left(\frac{\mathbb{E}_{\Omega_{\delta}}^{+}\left[\sigma_{v_{j+1}}\sigma_{a_2}\dots\sigma_{a_n}\right]} {\mathbb{E}_{\Omega_{\delta}}^{+}\left[\sigma_{v_{j}}\sigma_{a_2}\dots\sigma_{a_n}\right]}-1\right)(1 +o(1))\\ & = 2\delta\cdot[\Re\mathfrak{e}\,\mathcal{A}_\Omega(v_j;a_2,\dots,a_n)+o(1)]
\end{align*}
as $\delta\to 0$, where the $o(1)$ terms are unform in $j$. Consequently,
\begin{align}\label{LogE/Econv}
\log\frac{\mathbb{E}_{\Omega_{\delta}}^{+}\left[\sigma_{a_1'}\sigma_{a_2}\dots\sigma_{a_n}\right]} {\mathbb{E}_{\Omega_{\delta}}^{+}\left[\sigma_{a_1}\sigma_{a_2}\dots\sigma_{a_n}\right]} & = 2\delta\cdot\sum\limits_{j=1}^{m_\delta}\Re\mathfrak{e}\,\mathcal{A}_\Omega(v_j,a_2,\dots,a_n)+o(1) \cr
& \underset{\delta\to0}{\longrightarrow} \int_{[a_1,a_1']}\Re\mathfrak{e}\,\mathcal{A}_\Omega(x_1+iy_1,a_2,\dots,a_n)dx_1.
\end{align}
A similar formula with $-\Im\mathfrak{m}\,\mathcal{A}_\Omega(x_1+iy_1,a_2,\dots,a_n)dy_1$ in the right-hand side applies to the case when $[a_1,a'_1]$ is a vertical segment. Since by (\ref{eq: A_expl}) one has
\begin{eqnarray*}
\Re\mathfrak{e}\,\mathcal{A}_\Omega(x_1+iy_1,a_2,\dots,a_n) & = &\partial_{x_1}\log \langle\sigma_{x_1+iy_1}\dots\sigma_{a_n}\rangle^{+}_{\Omega},\\
-\Im\mathfrak{m}\,\mathcal{A}_\Omega(x_1+iy_1,a_2,\dots,a_n)& = &\partial_{y_1}\log \langle\sigma_{x_1+iy_1}\dots\sigma_{a_n}\rangle^{+}_{\Omega},
\end{eqnarray*}
the proposition readily follows by applying the above computation, if necessary, to several horizontal and vertical segments and to other marked points $a_2,\dots,a_n$.
\end{proof}

\begin{rem}
\label{rem:implicit_integration_L}
Note that whereas Lemma \ref{lem: A_explicit} allows us to write the limits of correlations in Theorem \ref{thm:1-3pts} explicitly, it is otherwise not essential for the proof. In fact, Theorem \ref{thm:log-derivatives} implies that discrete exact differential forms $d\log\mathbb{E}_{\Omega_{\delta}}^{+}\left[\sigma_{a_1}\dots\sigma_{a_n}\right]$ have continuous limits $\mathcal{L}_{\Omega,n}$ given in terms of $\mathcal{A}_\Omega$, see (\ref{LogE/Econv}). Thus, one could simply \mbox{\emph{define}} the functions $\langle\sigma_{a_1}\dots\sigma_{a_n}\rangle^{+}_{\Omega}$ to be $\exp(\int\mathcal{L}_{\Omega,n})$ with properly chosen constants of integration. The only property of $\langle\sigma_{a_1}\dots\sigma_{a_n}\rangle^{+}_{\Omega}$ that we use below (and which also defines their multiplicative normalization for $n>2$ uniquely) is the following:
\begin{equation}
\label{eq: bdry_decor}
\frac{\langle\sigma_{a_1}\dots\sigma_{a_n}\rangle_{\Omega}^{+}} {\langle\sigma_{a_1}\rangle_{\Omega}^{+}\langle\sigma_{a_2}\dots\sigma_{a_n}\rangle_{\Omega}^{+}}\to 1 ~~\text{as}~~a_1\to\partial\Omega.
\end{equation}
In principle, existence of constants of integration such that (\ref{eq: bdry_decor}) holds can be derived from asymptotics
\begin{eqnarray*}
\mathcal{A}_\Omega(a_1,\dots,a_n)-\mathcal{A}_\Omega(a_1)=o(1),& a_1\to\partial\Omega&\\
\mathcal{A}_\Omega(a_1,\dots,a_n)-\mathcal{A}_\Omega(a_1,\dots\hat{a}_k\dots,a_n)=o(1), &a_k\to \partial\Omega, &k=2,\dots,n,
\end{eqnarray*}
which could be obtained by comparing solutions to the corresponding boundary value problems (\ref{def:spinor-bdry})\,--\,(\ref{def:spinor-1}) as $a_1\to \partial\Omega$.
\end{rem}

\subsection{From ratios of correlations to Theorem~\ref{thm:2pts}}\label{sub:ratios-to-thm-2pts}
This section is devoted to the proof of Theorem~\ref{thm:2pts}. Our goal is to relate the normalizing factors $\rho_2(\delta,\Omega_\delta)$ from Corollary~\ref{cor:conv-up-to-normalization} (which, in principle, might depend on $\Omega$), with the full-plane normalization $\rho(\delta)$. The proof is based on Theorem~\ref{thm:ratio-free-plus}, which claims the convergence of the ratios of free and $+$ spin-spin correlations to an explicit limit $\mathcal{B}_\Omega(a,b)={\left\langle \sigma_{a}\sigma_{b}\right\rangle _{\Omega}^{\mathrm{free}}}/{\left\langle \sigma_{a}\sigma_{b}\right\rangle _{\Omega}^{+}}$. We also use classical FKG (e.g., see \cite[Chapter~2]{grimmett}) and GHS (see~\cite{griffith-hurst-sherman}) inequalities for the Ising model. The small additional ingredient is given by
\begin{rem}
\label{rem: ratio_free_plus}
The following property holds:
\begin{equation}
\label{eq: ratio_free_plus}
\mathcal{B}_\Omega(a,b)\to 1 ~~\text{as}~~\mathcal{D}_\Omega(a,b)\to 0, \quad \text{and} \quad
\mathcal{B}_\Omega(a,b)\to 0 ~~\text{as}~~ a\to\partial \Omega,
\end{equation}
where
\begin{equation}
\label{eq: D-def}
\mathcal{D}_\Omega(a,b):=\frac{|a-b|}{\text{dist}(\{a,b\},\partial\Omega)}
\end{equation}
is a quantity that measures how deeply in the bulk of $\Omega$ the points $a,b$ are. This follows readily from  conformal invariance of $\mathcal{B}_\Omega(a,b)$ and explicit formulae (\ref{12ptFcts}) for the half-plane. Also, if one denotes
\[
\langle \sigma_a\sigma_b\rangle_{\C}^{+}:=|a-b|^{-\frac{1}{4}},
\]
then
\begin{equation}
\label{eq: bulk_limit}
\frac{\left\langle \sigma_{a}\sigma_{b}\right\rangle_{\Omega}^{+}}{\left\langle \sigma_{a}\sigma_{b}\right\rangle _{\C}^{+}}\to 1 ~~\text{as}~~ \mathcal{D}_\Omega(a,b)\to 0.
\end{equation}
Indeed, let $\varphi$ be a conformal map from $\Omega$ to $\bH$ such that $\varphi(a)=i$. Due to standard estimates, one has $\varphi(b)\to i$ and $|\varphi'(b)|\to |\varphi'(a)|$ as $\mathcal{D}_\Omega(a,b)\to 0$. Therefore, 
\[
\frac{\left\langle \sigma_{a}\sigma_{b}\right\rangle _{\Omega}^{+}}{\left\langle \sigma_{a}\sigma_{b}\right\rangle _{\C}^{+}}
=
\frac{\langle \sigma_{\varphi(a)}\sigma_{\varphi(b)}\rangle _{\bH}^{+}|\varphi'(a)|^{\frac{1}{8}}|\varphi'(b)|^{\frac{1}{8}}}{\left\langle \sigma_{a}\sigma_{b}\right\rangle _{\C}^{+}}
=
\frac{|\varphi(a)\!-\!\varphi(b)|^{-\frac14}|\varphi'(a)|^{\frac{1}{4}}}{|a-b|^{-\frac14}}\left(1+o(1)\right)\to 1.
\]
\end{rem}

\begin{lem}
For any $\eta>0$ there exists an $\epsilon>0$ such that the following holds: if $\mathcal{D}_\Omega(a,b)<\epsilon$ and $\Omega_\delta$ approximates $\Omega$, then
\[
1-\eta \leq\frac{\E_{\C_\delta}[\sigma_{a}\sigma_{b}]}{\E_{\Omega_\delta}^{+}[\sigma_{a}\sigma_{b}]}\leq 1
\]
provided that $\delta$ is small enough.
\label{lem: decor_disc}
\end{lem}
\begin{proof}
By FKG inequality,
$
\mathbb{E}_{\Lambda_{\delta}}^{\mathrm{free}}\left[\sigma_{a}\sigma_{b}\right]\leq\mathbb{E}_{\mathbb{C}_{\delta}}\left[\sigma_{a}\sigma_{b}\right]\leq\mathbb{E}_{\Omega_{\delta}}^{+}\left[\sigma_{a}\sigma_{b}\right]
$ for any domain $\Lambda_\delta$ containing $a,b$, hence the right-hand side readily follows. For the left-hand side, choose $\Lambda_\delta=\Omega^{\bullet}_\delta-\delta$. We have, by Theorem~\ref{thm:ratio-free-plus},
\[
\frac{\mathbb{E}_{\Lambda_{\delta}}^{\mathrm{free}}\left[\sigma_{a}\sigma_{b}\right]}{\mathbb{E}_{\Omega_{\delta}}^{+}\left[\sigma_{a}\sigma_{b}\right]}
=
\frac{\mathbb{E}_{\Omega^{\bullet}_\delta}^{\mathrm{free}}\left[\sigma_{a+\delta}\sigma_{b+\delta}\right]} {\mathbb{E}_{\Omega_{\delta}}^{+}\left[\sigma_{a}\sigma_{b}\right]} ~\underset{\delta\to0}{\longrightarrow}~\mathcal{B}_\Omega(a,b). 
\]
Due to (\ref{eq: ratio_free_plus}), one can choose $\epsilon$ so that $\mathcal{B}_\Omega(a,b)\ge 1-\frac{\eta}{2}$, which gives the result.
\end{proof}

\begin{proof}[\textbf{Proof of Theorem~\ref{thm:2pts}}]
Fix $\eta>0$. For shortness, below we will write $a+\epsilon$ for its lattice approximations $a+2\delta\lfloor\frac{\epsilon}{2\delta}\rfloor$, and $a+1$ for its lattice approximation $a+2\delta\lfloor\frac{1}{2\delta}\rfloor$. Denote
\begin{equation}
\label{eq: R+def}
R^{\Omega_\delta}_{a,b}:=
(\varrho(\delta))^{-1}{\mathbb{E}_{\Omega_{\delta}}^{+}\left[\sigma_{a}\sigma_{b}\right]}.
\end{equation}
Recall that $\varrho(\delta)={\mathbb{E}_{\mathbb{C}_{\delta}}\left[\sigma_{a}\sigma_{a+1}\right]}$ by definition, so one can write
\[
R^{\Omega_\delta}_{a,b}=\frac{\mathbb{E}_{\Omega_{\delta}}^{+}\left[\sigma_{a}\sigma_{b}\right]}{\mathbb{E}^{+}_{\Omega_{\delta}}\left[\sigma_{a}\sigma_{a+\epsilon}\right]}\cdot
\frac{\mathbb{E}_{\Omega_{\delta}}^{+}\left[\sigma_{a}\sigma_{a+\epsilon}\right]}{\mathbb{E}_{\mathbb{C}_{\delta}}\left[\sigma_{a}\sigma_{a+\epsilon}\right]}\cdot
\frac{\mathbb{E}_{\C_{\delta}}\left[\sigma_{a}\sigma_{a+\epsilon}\right]}{\mathbb{E}_{\mathbb{C}_{\delta}}\left[\sigma_{a}\sigma_{a+1}\right]}.
\]
By Lemma \ref{lem: decor_disc}, we can find a small $\epsilon>0$ and a large domain $\Lambda_\delta$ containing $a$, $b$, $a+\epsilon$ and $a+1$, such that
\[
1-\eta ~\leq~ \frac{\mathbb{E}_{\Omega_{\delta}}^{+}\left[\sigma_{a}\sigma_{a+\epsilon}\right]}{\mathbb{E}_{\mathbb{C}_{\delta}}\left[\sigma_{a}\sigma_{a+\epsilon}\right]} ~\leq~ 1\quad \text{and}\quad 1-\eta ~\le~ \frac{\mathbb{E}_{\C_{\delta}}\left[\sigma_{a}\sigma_{a+\epsilon}\right]}{\mathbb{E}_{\mathbb{C}_{\delta}}\left[\sigma_{a}\sigma_{a+1}\right]}\cdot
\frac{\mathbb{E}_{\Lambda_{\delta}}^{+}\left[\sigma_{a}\sigma_{a+1}\right]}{\mathbb{E}_{\Lambda_{\delta}}^+\left[\sigma_{a}\sigma_{a+\epsilon}\right]} ~\le~ 1+\eta,
\]
provided that $\delta$ is small enough. Consequently,
\[
(1-\eta)^2
\frac{\mathbb{E}_{\Omega_{\delta}}^{+}\left[\sigma_{a}\sigma_{b}\right]}{\mathbb{E}^+_{\Omega_{\delta}}\left[\sigma_{a}\sigma_{a+\epsilon}\right]}\cdot
\frac{\mathbb{E}_{\Lambda_{\delta}}^{+}\left[\sigma_{a}\sigma_{a+\epsilon}\right]}{\mathbb{E}^+_{\Lambda_{\delta}}\left[\sigma_{a}\sigma_{a+1}\right]}
\,\leq\, R^{\Omega_\delta}_{a,b}
\,\leq\, (1+\eta)
\frac{\mathbb{E}_{\Omega_{\delta}}^{+}\left[\sigma_{a}\sigma_{b}\right]}{\mathbb{E}^+_{\Omega_{\delta}}\left[\sigma_{a}\sigma_{a+\epsilon}\right]}\cdot
\frac{\mathbb{E}_{\Lambda_{\delta}}^{+}\left[\sigma_{a}\sigma_{a+\epsilon}\right]}{\mathbb{E}^+_{\Lambda_{\delta}}\left[\sigma_{a}\sigma_{a+1}\right]},
\]
and, by convergence of the \emph{ratios} of spin correlations proven in Proposition~\ref{prop: conv_ratios},
\[(1-\eta)^3
\frac{\langle\sigma_{a}\sigma_{b}\rangle_{\Omega}^{+}}{\langle\sigma_{a}\sigma_{a+\epsilon}\rangle^+_{\Omega}}\cdot
\frac{\langle\sigma_{a}\sigma_{a+\epsilon}\rangle_{\Lambda}^{+}}{\langle\sigma_{a}\sigma_{a+1}\rangle^+_{\Lambda}}
\,\leq\,
R^{\Omega_\delta}_{a,b}
\,\leq\,
(1+\eta)^2
\frac{\langle\sigma_{a}\sigma_{b}\rangle_{\Omega}^{+}}{\langle\sigma_{a}\sigma_{a+\epsilon}\rangle^+_{\Omega}}\cdot
\frac{\langle\sigma_{a}\sigma_{a+\epsilon}\rangle_{\Lambda}^{+}}{\langle\sigma_{a}\sigma_{a+1}\rangle^+_{\Lambda}}
\]
for $\delta$ small enough. Since $\eta$ can be chosen arbitrary small, and the bounds do not depend on $\delta$, it only remains to show that we can make the factor
\[
\frac{\langle\sigma_{a}\sigma_{a+\epsilon}\rangle_{\Lambda}^{+}}{ \langle\sigma_{a}\sigma_{a+\epsilon}\rangle^+_{\Omega}\langle\sigma_{a}\sigma_{a+1}\rangle^+_{\Lambda}}
\]
as close to $1$ as we wish by choosing $\epsilon$ small enough and $\Lambda$ large enough. However, this follows readily from (\ref{eq: bulk_limit}) if we multiply this factor by $\langle\sigma_{a}\sigma_{a+1}\rangle_{\C}=1$. Thus, ${\mathbb{E}_{\Omega_{\delta}}^{+}\left[\sigma_{a}\sigma_{b}\right]}\sim {\varrho(\delta)} \langle\sigma_{a}\sigma_{b}\rangle_{\Omega}^{+}$ as $\delta\to 0$.
To derive the asymptotics of two-point correlations for free boundary conditions as $\delta\to 0$, note that Theorem~\ref{thm:ratio-free-plus} implies
\[
{\mathbb{E}_{\Omega_{\delta}^\bullet-\delta}^\mathrm{free}\left[\sigma_a\sigma_b\right]}\sim \mathcal{B}_\Omega(a,b)\cdot {\mathbb{E}_{\Omega_{\delta}}^{+}\left[\sigma_{a}\sigma_{b}\right]} \sim
{\varrho(\delta)}\cdot\mathcal{B}_\Omega(a,b)\langle\sigma_{a}\sigma_{b}\rangle_{\Omega}^{+}= {\varrho(\delta)}\langle\sigma_{a}\sigma_{b}\rangle_{\Omega}^\mathrm{free}.
\]
The fact that we have $\Omega_{\delta}^\bullet-\delta$ instead of $\Omega_\delta$ plays no role, since they both approximate the same continuous domain  $\Omega$ and the convergence of $(\varrho(\delta))^{-1}\mathbb{E}_{\Omega_{\delta}}^{+}\left[\sigma_{a}\sigma_{b}\right]$ is independent of the particular choice of lattice approximations.
\end{proof}
\begin{rem}
As a simple byproduct of our analysis, we obtain the rotational invariance of the full-plane correlations recently proven by Pinson \cite{pinson}:
by FKG inequality, for any (large) domain $\Omega_\delta$, one has $\mathbb{E}_{\Omega_{\delta}}^\mathrm{free}\left[\sigma_{a}\sigma_{b}\right]\le \mathbb{E}_{\C_{\delta}}\left[\sigma_{a}\sigma_{b}\right] \le \mathbb{E}_{\Omega_{\delta}}^{+}\left[\sigma_{a}\sigma_{b}\right]$ and, due to Theorem~\ref{thm:2pts} and (\ref{eq: ratio_free_plus}), both sides have the same asymptotics when $\Omega_\delta$ exhausts $\C_\delta$. Then, (\ref{eq: bulk_limit}) gives the desired result:
\[
(\rho(\delta))^{-1}\E_{\C_\delta}[\sigma_a\sigma_b]\rightarrow \langle\sigma_{a}\sigma_{b}\rangle_{\C}=|a-b|^{-\frac{1}{4}}~~\text{as}~~\delta\to 0.
\]
\end{rem}

\subsection{Decorrelation near the boundary and the proof of Theorem~\ref{thm:1-3pts}}\label{sub:proof-thm-k-pts}
This section is devoted to the proof of Theorem~\ref{thm:1-3pts}. Note that it was already proven above in the special case $n=2$, as a part of Theorem~\ref{thm:2pts}. Our goal is to relate the normalizing factors in the Corollary \ref{cor:conv-up-to-normalization} with $\varrho(\delta)$. Below we rely upon the asymptotics (\ref{eq: bdry_decor}) which readily follow from explicit formulae (\ref{12ptFcts}) and the conformal covariance rule (\ref{conf-cov}).
\begin{lem}
 \label{lem: dec_bdry_discrete}
Given a domain $\Omega$ with marked points $a_2,\dots,a_n$, $n\ge 2$, and a number $\eta>0$, there exists $\epsilon>0$ such that the following holds: if $a_1\in\Omega$ is $\epsilon$-close to the boundary, $\Omega_\delta$ approximates $\Omega$ and $\delta$ is small enough, then
\begin{equation*}
 1-\eta\leq\frac{\E_{\Omega_\delta}^{+}[\sigma_{a_1}]\E_{\Omega_\delta}^{+}[\sigma_{a_2}\dots\sigma_{a_n}]}{\E_{\Omega_\delta}^{+}[\sigma_{a_1}\dots\sigma_{a_n}]}\leq 1.
\end{equation*}
\end{lem}
\begin{proof}
The upper bound follows readily from FKG inequality. For the lower one, consider first the case $n=2$. A celebrated application (e.g., see~\cite{dembo-montanari}) of the GHS inequality \cite{griffith-hurst-sherman} reads
$
\E_{\Omega_\delta}^{+}[\sigma_{a}\sigma_{b}]-\E_{\Omega_\delta}^{+}[\sigma_{a}]\E_{\Omega_\delta}^{+}[\sigma_{b}]\leq \E_{\Omega_\delta}^{\text{free}}[\sigma_{a}\sigma_{b}],
$
or equivalently,
\[
1-\frac{\E_{\Omega_\delta}^{\text{free}}[\sigma_{a}\sigma_{b}]}{\E_{\Omega_\delta}^{+}[\sigma_{a}\sigma_{b}]}
\leq\frac{\E_{\Omega_\delta}^{+}[\sigma_{a}]\E_{\Omega_\delta}^{+}[\sigma_{b}]}{\E_{\Omega_\delta}^{+}[\sigma_{a}\sigma_{b}]}.
\]
As $\delta\to 0$, the left-hand side converges to $1-\mathcal{B}_\Omega(a;b)$, 
so (\ref{eq: ratio_free_plus}) implies the claim.

To prove the result for $n\geq 3$, assume that we have already proved Theorem~\ref{thm:1-3pts} for all $n'<n$ (the precise description of our induction scheme is given in the proof of Theorem~\ref{thm:1-3pts} below, see (\ref{InductionScheme})). Let $\gamma$ be a cross-cut (simple path) in $\Omega$ separating $a_1$ from $a_2,\dots,a_n$, and let $\Omega'$ and $\Omega''$ be the corresponding connected components. Note that FKG inequality implies
\[
\frac{\E_{\Omega_\delta}^{+}[\sigma_{a_1}]\E_{\Omega_\delta}^{+}[\sigma_{a_2}\dots\sigma_{a_n}]}{\E_{\Omega_\delta}^{+}[\sigma_{a_1}\dots\sigma_{a_n}]} \cdot \frac{\E_{\Omega'_\delta}^{+}[\sigma_{a_1}]\E_{\Omega''_\delta}^{+}[\sigma_{a_2}\dots\sigma_{a_n}]} {\E_{\Omega_\delta}^{+}[\sigma_{a_1}]\E_{\Omega_\delta}^{+}[\sigma_{a_2}\dots\sigma_{a_n}]} \,\ge\, 1
\]
as $\E_{\Omega'_\delta}^{+}[\sigma_{a_1}]\E_{\Omega''_\delta}^{+}[\sigma_{a_2}\dots\sigma_{a_n}]$ is equal to the correlation of $\sigma_{a_1},\dots,\sigma_{a_n}$ in $\Omega_\delta$ conditioned on the event that all spins neighboring $\gamma$ are $+$.
By the induction assumption, the second factor converges to
\[
\frac{\langle\sigma_{a_1}\rangle_{\Omega'}^{+}}{\langle\sigma_{a_1}\rangle_{\Omega}^{+}}\cdot
\frac{\langle\sigma_{a_2}\dots\sigma_{a_n}\rangle_{\Omega''}^{+}}{\langle\sigma_{a_2}\dots\sigma_{a_n}\rangle_{\Omega}^{+}}
\]
as $\delta\to 0$, hence it is sufficient to show that we can make this quantity arbitrarily close to $1$ by choosing $a_1$ and $\gamma$ appropriately. We first choose a cross-cut $\gamma$ in such a way that $\Omega''$ would be Carath\'eodory-close to $\Omega$ as seen from $a_2,\dots,a_n$ and then put $a_1$  deeply inside $\Omega'$, so that $\Omega'$ would be Carath\'eodory close to $\Omega$ as seen from $a_1$. If two domains are Carath\'eodory-close, then the conformal maps from these domains to $\bH$ mapping the marked point to $i$ (say, with positive derivative there) are uniformly close on compact sets together with their derivatives. Thus, the lemma follows from continuity of the half-plane functions $\langle\sigma_{a_1}\dots\sigma_{a_n}\rangle_{\bH}^{+}$ with respect to positions of $a_1,\dots,a_n$.
\end{proof}

\begin{proof}[\textbf{Proof of Theorem~\ref{thm:1-3pts}}] Recall that the special case $n=2$ is already done as a part of Theorem~\ref{thm:2pts}. We proceed by induction which (together with the proof of Lemma~\ref{lem: dec_bdry_discrete} given above) starts as follows:
\begin{equation}
\label{InductionScheme}
T_2\& L_2\Rightarrow T_1\Rightarrow L_3\Rightarrow T_3\Rightarrow L_4\Rightarrow T_4\Rightarrow\dots
\end{equation}
(where $T_j$ and $L_j$ mean the particular cases of Theorem~\ref{thm:1-3pts} and Lemma~\ref{lem: dec_bdry_discrete}).

Let $n=1$ and $\eta>0$ be fixed. Similarly to (\ref{eq: R+def}), denote
\[
R^{\Omega_\delta}_{a}:=(\varrho(\delta))^{-\frac12}\mathbb{E}_{\Omega_{\delta}}^{+}\left[\sigma_{a}\right].
\]
For any $b\in\Omega$, one can write
\[
(R^{\Omega_\delta}_{a})^2 
=
\frac{\mathbb{E}_{\Omega_{\delta}}^{+}\left[\sigma_{a}\right]}{\mathbb{E}_{\Omega_{\delta}}^{+}\left[\sigma_{b}\right]}\cdot
\frac{\mathbb{E}_{\Omega_{\delta}}^{+}\left[\sigma_{a}\sigma_{b}\right]}{\varrho\left(\delta\right)}\cdot \frac{\mathbb{E}_{\Omega_{\delta}}^{+}\left[\sigma_{a}\right]\mathbb{E}_{\Omega_{\delta}}^{+}\left[\sigma_{b}\right]} {\mathbb{E}_{\Omega_{\delta}}^{+}\left[\sigma_{a}\sigma_{b}\right]}.
\]
By Lemma \ref{lem: dec_bdry_discrete}, if we choose $b$ close enough to the boundary, then
\[
1-\eta\leq\frac{\mathbb{E}_{\Omega_{\delta}}^{+}\left[\sigma_{a}\right]\mathbb{E}_{\Omega_{\delta}}^{+}\left[\sigma_{b}\right]} {\mathbb{E}_{\Omega_{\delta}}^{+}\left[\sigma_{a}\sigma_{b}\right]}\leq 1,
\]
provided that $\delta$ is small enough. Due to Proposition~\ref{prop: conv_ratios} and Theorem~\ref{thm:2pts}, one has
\[
\frac{\mathbb{E}_{\Omega_{\delta}}^{+}\left[\sigma_{a}\right]}{\mathbb{E}_{\Omega_{\delta}}^{+}\left[\sigma_{b}\right]} ~\underset{\delta\to0}{\longrightarrow}~ \frac{\langle \sigma_a \rangle_{\Omega}^{+}}{\langle \sigma_b \rangle_{\Omega}^{+}}\quad \text{and} \quad \frac{\mathbb{E}_{\Omega_{\delta}}^{+}\left[\sigma_{a}\sigma_{b}\right]}{\varrho\left(\delta\right)} ~\underset{\delta\to0}{\longrightarrow}~ \langle\sigma_{a}\sigma_{b}\rangle_\Omega^+.
\]
Consequently,
\[(1-\eta)^3\left(\langle \sigma_a \rangle_{\Omega}^{+}\right)^2\cdot\frac{\langle\sigma_{a}\sigma_{b}\rangle_{\Omega}^{+}}{\langle \sigma_a \rangle_{\Omega}^{+}\langle \sigma_b \rangle_{\Omega}^{+}}
\leq
(R^{\Omega_\delta}_{a})^2 
\leq
(1+\eta)^2\left(\langle \sigma_a \rangle_{\Omega}^{+}\right)^2\cdot\frac{\langle\sigma_{a}\sigma_{b}\rangle_{\Omega}^{+}}{\langle \sigma_a \rangle_{\Omega}^{+}\langle \sigma_b \rangle_{\Omega}^{+}}
\]
provided that $\delta$ is small enough. Since $\eta$ can be chosen arbitrary small, and by (\ref{eq: bdry_decor}) we can make the term ${\langle\sigma_{a}\sigma_{b}\rangle_{\Omega}^{+}}/{\langle \sigma_a \rangle_{\Omega}^{+}\langle \sigma_b \rangle_{\Omega}^{+}}$ arbitrary close to $1$ choosing $b$ sufficiently close to $\partial\Omega$, we complete the proof for $n=1$ by remark that positivity of magnetization fixes the sign of $R^{\Omega_\delta}_{a}$.

To get the convergence of $R^{\Omega_\delta}_{a_1,\dots,a_n}:=\varrho(\delta)^{-\frac{n}{2}}\E^{+}_{\Omega_\delta}[\sigma_{a_1},\dots\sigma_{a_n}]$ for $n\ge 3$, proceed by induction and write it as
\[
R^{\Omega_\delta}_{a_1,\dots,a_n} 
=
\frac{\E^{+}_{\Omega}[\sigma_{a_1}\sigma_{a_2}\dots\sigma_{a_n}]}{\E^{+}_{\Omega}[\sigma_{b}\sigma_{a_2}\dots \sigma_{a_n}]}
\cdot \frac{\E^{+}_{\Omega}[\sigma_{b}\sigma_{a_2}\dots\sigma_{a_n}]}{\E^{+}_{\Omega}[\sigma_{b}]\E^{+}_{\Omega}[\sigma_{a_2}\dots\sigma_{a_n}]}
\cdot R^{\Omega_\delta}_{b}R^{\Omega_\delta}_{a_2,\dots,a_n}. 
\]
The proof is finished similarly to the case $n=1$: take $b$ close to the boundary, estimate the second ratio in the right-hand side by Lemma~\ref{lem: dec_bdry_discrete}, then use convergence of the first ratio and $R^{\Omega_\delta}_{b}R^{\Omega_\delta}_{a_2,\dots,a_n}$ as $\delta\to 0$ (this follows from Proposition~\ref{prop: conv_ratios} and the induction hypothesis, respectively) and asymptotics (\ref{eq: bdry_decor}) for the continuous correlation functions: ${\langle\sigma_b\rangle_{\Omega}^{+}\langle\sigma_{a_2}\dots\sigma_{a_n}\rangle_{\Omega}^{+}}/{\langle\sigma_b\sigma_{a_2}\dots\sigma_{a_n}\rangle_{\Omega}^{+}}\to 1$ as $b\to\partial\Omega$.
\end{proof}

\section{\label{sec:analysis-of-spinors}Proofs of the main convergence Theorems \ref{thm:obs-away}, \ref{thm:localization-near-a} and \ref{thm:localization-near-b} for~discrete~spinors}

\subsection{\label{sub:s-holomorphicity-proof}S-holomorphicity of discrete observables}
The notion of s-holomorphicity was essentially introduced by Smirnov in \cite{smirnov-i} and
used for the study of the planar critical Ising model in \cite{smirnov-i, smirnov-ii, chelkak-smirnov-ii, hongler-i, hongler-smirnov-ii, chelkak-izyurov}. Our definitions follow \cite{hongler-smirnov-ii} and are equivalent to those of \cite{smirnov-ii, chelkak-smirnov-ii,chelkak-izyurov} after the multiplication by $\sqrt{i}$, see also Section~\ref{sub: riemann-problem} below.

Recall that, for $\tau\in\left\{ 1,i,\lambda,\overline{\lambda}\right\}$, we associate the line $\ell\left(x\right)=\tau\mathbb{R}$ in the
complex plane with each corner $x\in\mathcal{V}_{\mathbb{C}_{\delta}}^{\tau}$, and denote by $\mathsf{P}_{\ell\left(x\right)}$
the orthogonal projection onto that line:
\[
\mathsf{P}_{\ell\left(x\right)}\left[w\right]=\tfrac{1}{2}\left(w+\tau^{2}\overline{w}\right),\quad w\in\mathbb{C}.
\]

Definition~\ref{def: shol} says that a function $F:\mathcal{V}_{\Omega_{\delta}}^{\mathrm{cm}}\to\mathbb{C}$
is \emph{s-holomorphic in $\Omega_\delta$} (we use the same definitions working on double covers) if for every $x\in\mathcal{V}_{\Omega_{\delta}}^{\mathrm{c}}$
and $z\in\mathcal{V}_{\Omega_{\delta}}^{\mathrm{m}}$ that are adjacent, one has
\[
F\left(x\right)=\mathsf{P}_{\ell\left(x\right)}\left[F\left(z\right)\right].
\]

\begin{rem}
\label{rem: s-hol-extended}
The set $\mathcal{V}^{1,i}_{\C_\delta}=\mathcal{V}^{1}_{\C_\delta}\cup\mathcal{V}^{i}_{\C_\delta}$ may be viewed as a square lattice, divided into $\mathcal{V}^{1}_{\C_\delta}$ and $\mathcal{V}^{i}_{\C_\delta}$ in a chessboard fashion. By definition, the restriction of an \mbox{s-holo}morphic function $F$ to $\mathcal{V}^{1,i}_{\C_\delta}$ is real on $\mathcal{V}^{1}_{\C_\delta}$ and purely imaginary on $\mathcal{V}^{i}_{\C_\delta}$. It is not difficult to check (\cite{smirnov-ii}) that this restriction is in fact \emph{discrete holomorphic} in the most usual sense, that is, for any $x\in\mathcal{V}^{1,i}_{\C_\delta}$ one has
\begin{equation}
F(x+i\delta)-F(x+\delta)=i\cdot\left[F(x+(1\!+\!i)\delta)-F(x)\right].
\label{eq: disc_holom}
\end{equation}
The converse is also true: given discrete holomorphic function $F:\mathcal{V}^{1,i}_{\C_\delta}\to \R\cup i\R$, one can first extend it to $\mathcal{V}^{\mathrm{m}}_{\C_\delta}$ by the formula $F(z):=F(z+i\frac{\delta}{2})+F(z-i\frac{\delta}{2})$ and then to $\mathcal{V}^{\lambda,\overline{\lambda}}_{\C_\delta}$ by $F\left(x\right):=\mathsf{P}_{\ell\left(x\right)}\left[F\left(x\pm\frac{\delta}{2}\right)\right]$ (due to (\ref{eq: disc_holom}), these projections coincide).
\end{rem}

\def\Ato{{a_1^\rightarrow}}

We now check the s-holomorphicity of discrete spinor observables, essentially mimicking \cite{chelkak-smirnov-ii,hongler-smirnov-ii,chelkak-izyurov}. For shortness, below we use the notation
\[
\textstyle \Ato:=a_1+\frac{\delta}{2}\,\in\, \mathcal{V}^i_{\C_\delta}.
\]
\begin{proof}[\textbf{Proof of Proposition \ref{prop: s-hol+bc}.}]
Let $z\in\mathcal{V}_{\left[\Omega_{\delta},\Ak\right]}^{\mathrm{m}}$
be a medial vertex and $x$ 
be one of four nearby corners so that $|x-z|=\frac{\delta}{2}$. We should check that
\begin{eqnarray}
\label{eq: s-hol-spinor}
\mathsf{P}_{\ell\left(x\right)}\left[F_{\left[\Omega_\delta,\Ak\right]}\left(z\right)\right] & = & F_{\left[\Omega_\delta,\Ak\right]}\left(x\right),
\end{eqnarray}
where the values of $F_{\left[\Omega_\delta,\Ak\right]}$ are defined as sums over the sets $\mathcal{C}_{\Omega_{\delta}}\left(\Ato,z\right)$
and $\mathcal{C}_{\Omega_{\delta}}\left(\Ato,x\right)$, respectively. There is a simple bijection $\omega_{zx}:\gamma_z\mapsto\gamma_x$ between these two sets provided by taking XOR (symmetric difference) of a configuration with two half edges $(zv)$ and $(vx)$, where $v$ 
denotes the vertex that is adjacent to both $x$ and $z$. Hence, it is sufficient to check that for any $\gamma_z\in \mathcal{C}_{\Omega_{\delta}}\left(\Ato,z\right)$ one has
\begin{equation}
(\cos\tfrac{\pi}{8})^{-1}\cdot\mathsf{P}_{\ell\left(x\right)}\left[\alpha_c^{\#\textrm{edges}(\gamma_z)}\cdot\phi_{\Ak}(\gamma_z,z)\right]  = \alpha_c^{\#\textrm{edges}(\gamma_x)}\cdot\phi_{\Ak}(\gamma_x,x)
\label{eq: s-hol-reduced}
\end{equation}
(the additional factor $(\cos\frac{\pi}{8})^{-1}$ comes from our definition of the discrete observable on medial vertices, see Remark~\ref{rem:spinor-def}).
There are two cases: either $(zv)$ is contained in $\gamma$, which leads us to
\begin{equation*}
\#\textrm{edges}(\gamma_x)=\#\textrm{edges}(\gamma_z),\quad \exp[-\tfrac{i}{2}\wind(\mathrm{p}(\gamma_x))]=e^{\pm\frac{i\pi}{8}}\exp[-\tfrac{i}{2}\wind(\mathrm{p}(\gamma_z))],
\end{equation*}
or not, which leads to
\begin{equation*}
\#\textrm{edges}(\gamma_x)=\#\textrm{edges}(\gamma_z)+1,\quad \exp[-\tfrac{i}{2}\wind(\mathrm{p}(\gamma_x))]=e^{\pm\frac{3i\pi}{8}}\exp[-\tfrac{i}{2}\wind(\mathrm{p}(\gamma_z))].
\end{equation*}
Let us also note that in both cases
\[
\left(-1\right)^{\#\mathrm{loops}_{\Ak}(\gamma_x)}\mathrm{sheet}_{\Ak}(\gamma_x,x)= \left(-1\right)^{\#\mathrm{loops}_{\Ak}(\gamma_z)}\mathrm{sheet}_{\Ak}(\gamma_z,z)
\]
since if $\omega_{zx}$ destroys a loop in $\gamma_z$ that changed the sheet of $[\Omega_\delta,a_1,\dots,a_n]$ (leading to the change of the first factor), then this loop becomes a part of $\mathrm{p}(\gamma_x)$, so the second factor changes simultaneously. Thus, one can factor out  $\alpha_c^{\#\textrm{edges}(\gamma_z)}\cdot\phi_{\Ak}(\gamma_z,z)$ from both sides of (\ref{eq: s-hol-reduced}). In the first case (\ref{eq: s-hol-reduced}) readily follows, while in the second it becomes equivalent to
\[
(\cos\tfrac{\pi}{8})^{-1}\cos\tfrac{3\pi}{8}=\sqrt{2}-1=\alpha_\mathrm{c}\,.
\]

Thus, $F_{[\Omega,\Ak]}$ is s-holomorphic. It has multiplicative monodromy $-1$ around each of the marked points $a_1,\dots,a_n$ due to the factor $\mathrm{sheet}_{\Ak}(\gamma,z)$ in (\ref{spinor-corners-def}). In order to prove that $F_{[\Omega,\Ak]}$ obeys boundary conditions (\ref{eq: discrete_bc}), it is sufficient to note that $\wind(\mathrm{p}(\gamma))=\nu_\mathrm{out}(z) \mod 2\pi$, if $z$ is on the boundary.
\end{proof}

The spinor $F_{[\Omega,\Ak]}$ is not defined at the corner $\Ato$ and is not s-holomorphic there. 
The next lemma shows, in particular, that its values at the nearby medial vertices $a_1+\frac{1\pm i}{2}\delta$ have different projections onto the imaginary line $i\R=\ell(\Ato)$, so one cannot extend $F_{[\Omega,\Ak]}$ to $\Ato$ in an s-holomorphic way.
\begin{lem}
\label{lem: sing_f_Omega}
For medial vertices $a_1+\frac{1\pm i}{2}\delta$ taken on the same sheet of the double cover $\left[\Omega_\delta,\Ak\right]$ as the ``source'' $\Ato$, one has
\[
 \mathsf{P}_{i\mathbb{R}}\left[F_{\left[\Omega_\delta,\Ak\right]}\left(a_1+\tfrac{1\pm i}{2}\delta\right)\right]  =  \mp i.
\]
\end{lem}
\begin{proof}
We consider the medial vertex $z:=a_1+\frac{1+i}{2}\delta$, the opposite case is similar. Given a configuration $\gamma\in \mathcal{C}_\delta(\Ato,z)$ and applying, as above, the XOR bijection with two half-edges $(\Ato,a_1+\delta)$ and $(a_1+\delta,z)$, we obtain a configuration $\omega(\gamma)\in\mathcal{C}_\delta$. Since the normalizing factor $\mathcal{Z}_{\Omega_\delta}^+[\Ak]$ is a sum over $\mathcal{C}_\delta$ (see (\ref{eq: Z+=sum})), it is sufficient to show that, for any $\gamma$,
\[
(\cos\tfrac{\pi}{8})^{-1}\cdot\mathsf{P}_{i\mathbb{R}}\left[\alpha_\mathrm{c}^{\#\textrm{edges}(\gamma)}\phi_{\Ak}(\gamma,z)\right]= -i\cdot\alpha_\mathrm{c}^{\#\text{edges}(\omega(\gamma))}(-1)^{\#\text{loops}_{\Ak}(\omega(\gamma))}.
\]
Consider two cases, as in the proof of Proposition~\ref{prop: s-hol+bc} above. If $(a_1+\delta,z)\in\gamma$ (respectively, $(a_1+\delta,z)\notin\gamma$), then
\[
\#\textrm{edges}(\omega(\gamma))=\#\textrm{edges}(\gamma)\quad\left(\text{respectively,}\quad \#\textrm{edges}(\omega(\gamma))=\#\textrm{edges}(\gamma)+1\right).
\]
We may disregard all loops in $\omega(\gamma)$ that do not contain the edge $(a_1+\delta,a_1+i\delta)$, as they contribute the same sign to both sides. Further, if $(a_1+\delta,z)\in\gamma$, then $\wind(\mathrm{p}(\gamma))=\frac{3\pi}{4}\mod 4\pi$, otherwise $\wind(\mathrm{p}(\gamma))=\frac{7\pi}{4}\mod 2\pi$. Hence the lemma boils down to the following elementary identities:
\[
(\cos\tfrac{\pi}{8})^{-1}\cdot \mathsf{P}_{i\mathbb{R}}[e^{-\frac{3\pi i}{8}}]=-i\quad\text{and}\quad
(\cos\tfrac{\pi}{8})^{-1}\cdot \mathsf{P}_{i\mathbb{R}}[e^{-\frac{7\pi i}{8}}]=-i\alpha_c. \qedhere
\]
\end{proof}

\subsection{The full-plane discrete spinor $F_{[\C_\delta,a]}$ and its discrete primitive $G_{[\C_\delta,a]}$} \label{sub:full-plane}
This subsection is mainly devoted to the construction of the full-plane analogue of discrete spinor observables, as announced in Lemma~\ref{lem: def_f_C}. After this, we also prove Lemma~\ref{def_r_delta} and the double-sided estimate (\ref{eq: vartheta-estimate}) of the normalizing factor $\vartheta(\delta)$.

\def\X{\mathbb{X}_\delta}
\def\Y{\mathbb{Y}_\delta}
\def\ato{{a^\rightarrow}}

\subsubsection{Discrete harmonic measure in the slit plane}
We start with an important technical ingredient -- the discrete Beurling estimate with optimal exponent $\frac{1}{2}$.
On the square lattice, it was obtained by Kesten \cite{kesten}, and then generalized by Lawler and Limic \cite{lawler-limic}.
Given a face $a$, let $\X\subset\mathcal{V}^{1}_{\C_\delta}$ denote the slit discrete plane:
\[
\mathrm{Int}\,\X:=\mathcal{V}^{1}_{\C_\delta}\backslash\, L_a,\qquad L_a:=\{x+\ato\!+\delta:x<0\}
\]
(recall that $\ato\!+\delta=a+\tfrac{3\delta}{2}\in\mathcal{V}^{1}_{\C_\delta}$).

\begin{lem}
\label{lem:discrete-beurling}
For all $z\in\X$, $A\subset\X$, and some absolute constant $C>0$, the following estimates are fulfilled:
\begin{equation}
\label{eq: hm-beurling}
\mathrm{hm}_{\{\ato+\delta\} }^{\X}\left(z\right)\leq C\delta^\frac{1}{2}|z-a|^{-\frac{1}{2}},
\end{equation}
\begin{equation}
\label{eq:beurling-Lambda}
\mathrm{hm}_{A}^{\X}\left(\ato\!+\delta\right)\leq C{\delta^\frac{1}{2}}(\mathrm{dist}(a;A))^{-\frac{1}{2}},
\end{equation}
where $\mathrm{hm}^{\X}_A(z)$ denotes the discrete harmonic measure of a set $A$ in $\X$ viewed from $z$, i.e., the probability for the simple random walk on $\mathcal{V}^{1}_{\C_\delta}$ (considered as a shifted square grid $(2\delta\mathbb{Z})^2$) started at $z$ to reach $A$ before it hits the boundary of $\X$.
\end{lem}
\begin{proof}
This easily follows from \cite{lawler-limic} and simple reversibility arguments for random walks.
\end{proof}
Below we also need some additional estimates for the discrete harmonic measure $\mathrm{hm}_{\left\{\ato+\delta\right\} }^{\X}$ and its discrete derivatives.
\begin{lem}
\label{lem: hm-estimate-La+grad}
(i) For $z\in\X$ such that $|\arg(z\!-\!a)-\pi|\le \frac{\pi}{4}$, one has
\begin{equation}
\label{eq: hm-near-La}
\mathrm{hm}_{\left\{ \ato+\delta\right\} }^{\X}\left(z\right)\le
C\delta^\frac{1}{2}|\Im\mathfrak{m}\,(z-a)|\,|z-a|^{-\frac{3}{2}}\,.
\end{equation}
(ii) For all neighboring $z,z'\in \X$, one has
\begin{equation}
\label{eq: hm-gradients}
\delta^{-1}\left|\,\mathrm{hm}_{\{\ato+\delta\}}^{\X}\left(z'\right) - \mathrm{hm}_{\{\ato+\delta\}}^{\X}\left(z\right)\right|\le C\delta^\frac{1}{2}|z-a|^{-\frac{3}{2}}.
\end{equation}
(iii) Being normalized by the value $\vartheta(\delta)$ at the proper lattice approximation of the point $a+1$, the functions $(\vartheta(\delta))^{-1}\mathrm{hm}_{\{\ato+\delta\}}^{\X}$ converge to $\Re\mathfrak{e}\,[1/\sqrt{z-a}\,]$ as $\delta\to 0$, uniformly on compact subsets of $\C\setminus L_a$.  Moreover, this convergence holds true in $C^1$-sense meaning that the discrete derivatives (\ref{eq: hm-gradients}) converge to the corresponding partial derivatives as well.
\end{lem}
\begin{proof}
(i) In order to prove (\ref{eq: hm-near-La}), note that the probability for the random walk started at $z$ to leave the ball of radius $\frac{1}{2}|z-a|$ around $z$ before hitting $\partial\X$ is $O(|\Im\mathfrak{m}\,(z-a)|\cdot|z-a|^{-1})$, and once this has happened the probability to hit $\ato+\delta$ is uniformly bounded by $O(\delta^\frac{1}{2}|z-a|^{-\frac{1}{2}})$ due to the Beurling estimate (\ref{eq: hm-beurling}).

\smallskip

\noindent (ii) The estimate (\ref{eq: hm-gradients}) for the discrete derivatives follows from (\ref{eq: hm-beurling}), (\ref{eq: hm-near-La}) and the (discrete) Harnack estimate (e.g., see \cite[Proposition~2.7]{chelkak-smirnov-i}), applied to the ball of radius $\frac{1}{2}|z-a|$ (or $\frac{1}{2}|\Im\mathfrak{m}\,(z-a)|$, if $z$ is close to $L_a$) around~$z$.

\smallskip

\noindent (iii) This is essentially a special case of \cite[Theorem~3.13]{chelkak-smirnov-i} which claims the $C^1$-convergence of discrete Poisson kernels normalized at some inner point to their continuous counterparts. The fact that our domain is unbounded plays no role here as, for any $r>0$,
the positive discrete harmonic functions $f_a^\delta:=(\vartheta(\delta))^{-1}\mathrm{hm}_{\{\ato+\delta\}}^{\X}$ are uniformly bounded in the annulus $\{z:|z-a|\ge r\}$. Indeed, if $f_a^\delta$ is big at some point $v$ in this annulus, then, by the maximum principle, $f_a^\delta$ is also big along some nearest-neighbor path running from $v$ to $\ato\!+\delta$. Since the discrete harmonic measure of such a path as seen from $a+1$ is bounded from below (by a constant depending on $r$ but not on $\delta$), this leads to a contradiction with $f_a^\delta(a+1)=1$.
\end{proof}

\subsubsection{Construction of the full-plane spinor $F_{[\C_\delta,a]}$}
Now we are ready to construct $F_{[\C_\delta,a]}$, as announced in Lemma~\ref{lem: def_f_C}.

\begin{proof}[\textbf{Proof of Lemma \ref{lem: def_f_C}}]
Let $\X^{\pm}$ denote two copies of the slit plane $\X\subset\mathcal{V}^{1}_{\C_\delta}$. We first define the real component $F^1_{[\C_\delta,a]}$ on $\X^{\pm}$ as
\begin{equation}
\label{F-1-as-hm}
F^1_{[\C_\delta,a]}(z)= \pm\mathrm{hm}^{\X^\pm}_{\{\ato+\delta\} }\left(z\right),\quad z\in\X^\pm.
\end{equation}
Since $F^1_{[\C_\delta,a]}$ vanishes identically on $L_a$, by identifying the upper side of this cut in $\X^+$ with the lower side in $\X^-$ and vice versa, we obtain a function $F^1_{[\C_\delta,a]}$ which is defined and harmonic everywhere on $\mathcal{V}^{1}_{[\C_\delta,a]}$ except at the (two) points over $\ato\!+\delta$. Recall that Lemma~\ref{lem: hm-estimate-La+grad}(iii) ensures the convergence
\begin{equation}
\label{F1-hm-convergence}
(\vartheta(\delta))^{-1}F^1_{[\C_\delta,a]}(z) ~\underset{\delta\to0}{\longrightarrow}~ \Re\mathfrak{e}\,[{1}/{\sqrt{z-a}}\,]
\end{equation}
as well as the convergence of discrete derivatives of $F^1_{[\C_\delta,a]}(z)$ to partial derivatives of $\Re\mathfrak{e}\,[{1}/{\sqrt{z-a}}\,]$, uniformly on compact subsets of $\C\setminus L_a$.

\smallskip

We then define the imaginary component $F^i_{[\C_\delta,a]}:\mathcal{V}^{i}_{[\C_\delta,a^\rightarrow]}=\mathcal{V}^{i}_{[\C_\delta,a]}\backslash \{\ato\}\to i\R$ as discrete harmonic conjugate of $F^1_{[\C_\delta,a]}$, that is, by integrating the identity (\ref{eq: disc_holom}) along paths on $\mathcal{V}^{i}_{[\C_\delta,a^\rightarrow]}$ starting from, say, one of the two fibers of the point $\ato\!+2\delta$. Since $F^1_{[\C_\delta,a]}$ is harmonic, its discrete harmonic conjugate is well-defined at least on the universal cover of $[\C_\delta,a^\rightarrow]$. Further, let $\varpi$ denote some simple loop in $\mathcal{V}^{i}_{\C_\delta}$, starting and ending at $\ato\!+2\delta$, symmetric with respect to the horizontal line $\{x:\Im\mathfrak{m}\,(x-a)=0\}$, and surrounding the singularity (so, it lifts to $[\C_\delta,a^\rightarrow]$ as a path connecting two different fibers of $\ato\!+2\delta$). It follows from the antisymmetry of $F^1_{[\C_\delta,a]}$ with respect to $L_a$ (also, note that $\varpi$ changes
the sheet once passing across the cut), that
the total increment of $F^i_{[\C_\delta,a]}$ along $\varpi$ is zero. Thus, $F^i_{[\C_\delta,a]}$ vanishes at both fibers of the point $\ato\!+2\delta$, and hence inherits the spinor property of its discrete harmonic conjugate $F^1_{[\C_\delta,a]}$. Moreover, for $b:=\ato\!+2j\delta$, $j\ge 1$, the discrete holomorphicity equation (\ref{eq: s-hol-reduced}) and
the symmetry of $F^1_{[\C_\delta,a]}$ with respect to the half-line $R_a:=\{x+\ato:x>0\}$ imply
\[
F^i_{[\C_\delta,a]}(b+2\delta)-F^i_{[\C_\delta,a]}(b)
= \mp i[F^1_{[\C_\delta,a]}(b\pm i\delta)-2F^1_{[\C_\delta,a]}(b+\delta)+F^1_{[\C_\delta,a]}(b+(2\pm i)\delta)]
=0.
\]
Therefore, $F^i_{[\C_\delta,a]}$ vanishes
everywhere on $R_a$.
Further, the estimate (\ref{eq: hm-gradients}) guarantees that $F^i_{[\C_\delta,a]}$ is uniformly bounded: just take a path of discrete integration in the definition of $F^i_{[\C_\delta,a]}$ running from $R_a$ to $z$ along the circular arc centered at $0$.

\smallskip

It is worth to note that $F^i_{[\C_\delta,a]}$ also admits a discrete harmonic measure representation similar to (\ref{F-1-as-hm}). Namely, let $\Y^{\pm}$ denote the two sheets of $\mathcal{V}^i_{[\C_\delta,a^\rightarrow]}\setminus R_a$, where signs in the notation are chosen so that $F^1_{[\C_\delta,a]}> 0$ in the upper half-plane on $\Y^+$ and in the lower half-plane on $\Y^-$. Then,
\begin{equation}
\label{F-i-as-hm'}
F^i_{[\C_\delta,a]}(z)=\mp i\cdot\mathrm{hm}^{\Y^\pm}_{\{\ato\}}(z),\quad z\in\Y^\pm.
\end{equation}
Indeed, on each of the sheets $\Y^\pm$ one can further extend $F^i_{[\C_\delta,a]}$ (as a harmonic conjugate of $F^1_{[\C_\delta,a]}$) to the point $\ato$: the only obstruction to do so on  $[\C_\delta,a]$ was that the increment of $F^i_{[\C_\delta,a]}$ along the smallest loop surrounding $\ato\!+\delta$ would be non-zero as $F^1_{[\C_\delta,a]}$ is not harmonic at $\ato\!+\delta$, but now this loop intersects the cut~$R_a$. Since the bounded function $F^i_{[\C_\delta,a]}$ is harmonic on $\Y^\pm$ and vanishes on~$R_a$, it is proportional to $\mathrm{hm}^{\Y^\pm}_{\{\ato\}}$. Finally, it follows from symmetry arguments that
\begin{eqnarray*}
\pm\mathrm{hm}^{\Y^\pm}_{\{\ato\}}\left(\ato\!+(1\!+\!i)\delta\right)~-~(\pm 1)
& = & \pm\mathrm{hm}^{\X^\pm}_{\{\ato+\delta\}}\left(\ato\!+i\delta\right)-
(\pm 1) \\ & = & F_{[\C_\delta,a]}^1\left(\ato+i\delta\right)-
F_{[\C_\delta,a]}^1\left(\ato\!+\delta\right) \\
& = & i \cdot [F_{[\C_\delta,a]}^i\left(\ato\!+\left(1\!+\!i\right)\delta\right)-
F_{[\C_\delta,a]}^i\left(\ato\right)]
\end{eqnarray*}
which fixes the multiplicative constant $\mp i$ in (\ref{F-i-as-hm'}).

\smallskip

Similarly to (\ref{F1-hm-convergence}), we have
\begin{equation}
\label{Fi-hm-convergence}
(\vartheta(\delta))^{-1}F^i_{[\C_\delta,a]}(z) ~\underset{\delta\to0}{\longrightarrow}~ \Im\mathfrak{m}\,[{1}/{\sqrt{z-a}}\,]
\end{equation}
together with the convergence of discrete derivatives, uniformly on compact subsets of $\C\setminus R_a$. Now, following Remark \ref{rem: s-hol-extended}, we extend $F_{[\C_\delta,a]}$ to $\mathcal{V}^{\mathrm{m}}_{[\C,a]}$ and $\mathcal{V}^{\lambda,\overline{\lambda}}_{[\C,a]}$ as an s-holomorphic function. The convergence (\ref{F_C-convergence}) readily follows from (\ref{F1-hm-convergence}) and (\ref{Fi-hm-convergence}) (for points lying near $L_a$, one can approximate $F^1_{[\C_\delta,a]}$ by summing discrete derivatives of $F^i_{[\C_\delta,a]}$ and vice versa near $R_a$).
\end{proof}

\def\x{{\dag}}

The following lemma explains the future role of $F_{[\C_\delta,a]}$: it has the same ``discrete singularity'' (whatever it means) at $\ato$ as  $F_{[\Omega_\delta,\Ak]}$ has at $\Ato$. This allows one to handle this singularity by taking the difference $F_{[\Omega_\delta,\Ak]}-F_{[\C_\delta,a_1]}$
\begin{lem}
\label{lem: remove_sing}
The spinor
$
F_{[\Omega_\delta,\Ak]}^\x:=F_{[\Omega_\delta,\Ak]}-F_{[\C_\delta,a_1]},
$
extended to be zero at $\Ato$, is s-holomorphic in any simply connected neighborhood of $a_1$ which does not contain other branching points $a_2,\dots, a_n$.
\end{lem}
\begin{proof}
By Lemma \ref{lem: sing_f_Omega}, it suffices to check that
$\mathsf{P}_{i\mathbb{R}}\left[F_{[\C_{\delta},a_1]}\left(a_1+\frac{1\pm i}{2}\delta\right)\right]  =  \mp i$ on the sheet $\X^+$. The considerations given in the previous proof show that, being considered on $\Y^\pm$ instead of $\X^+$, the spinor $F_{[\C_\delta,a_1]}$ can be extended at $\Ato$ in an s-holomorphic way and
\[
 \mathsf{P}_{i\mathbb{R}}\left[F_{[\C_{\delta},a_1]}\left(a_1+\tfrac{1\pm i}{2}\delta\right)\right]  = F_{[\C_{\delta},a_1]}^i(\Ato) = \mp i~~\text{on}~~\Y^\pm.
\]
However, by definition, $\X^+$ coincides with $\Y^+$ in the upper half-plane, and with $\Y^-$ in the lower half-plane.
\end{proof}

\subsubsection{Construction of the harmonic spinor $G_{[\C_\delta,a]}$}
\label{sub:G-construction} Now we ``integrate'' the real component of $F_{[\C_\delta,a]}$ in order to construct a discrete counterpart of the harmonic spinor $\Re\mathfrak{e}\,\sqrt{z-a}$, as announced in Lemma~\ref{def_r_delta}. Along the way, we also prove the double-sided estimate (\ref{eq: vartheta-estimate}) of the normalizing factor $\vartheta(\delta)$. Note that the upper bound $\vartheta(\delta)\le \mathrm{C}_+ \sqrt{\delta}$ directly follows from the discrete Beurling estimate (\ref{eq: hm-beurling}), so we need to prove the lower bound only.

\begin{proof}[\textbf{Proof of Lemma \ref{def_r_delta} and the estimate (2.16).}]
We use the notation for the sheets and the cuts of $[\C_\delta,a]$ introduced above. For $z\in\X^+\subset\mathcal{V}^{1}_{[\C_\delta,a]}$, define
\begin{equation}
G_{[\C_{\delta},a]}(z):=
\delta\cdot{\sum\limits_{j=0}^{\infty}} F^1_{[\C_\delta,a]}(z-2j\delta).
\label{eq: def_r}
\end{equation}
Due to the estimate (\ref{eq: hm-near-La}) for $F^1_{[\C_\delta,a]}(\,\cdot\,)=\mathrm{hm}^{\X}_{\{\ato+\delta\}}(\,\cdot\,)$, this sum always converges. Note that $F^1_{[\C_\delta,a]}=0$ on the cut $L_a$, so $G_{[\C_{\delta},a]}$ vanishes on $L_a$ too.

We are going to prove that $G_{[\C_{\delta},a]}$ is harmonic \emph{everywhere} inside $\X^+$, including the point $\ato\!+\delta$ right near the cut $L_a$.
For $z$ outside $R_a$, this harmonicity readily follows from the harmonicity of $F^1_{[\C_{\delta},a]}$, so let $z\in R_a$.
Denote by $\mathcal{S}_N(z)$,
\[
\mathrm{Int}\,\mathcal{S}_N(z):=\mathcal{V}^{1}_{\C_\delta}\cap\{z:|\Re\mathfrak{e}\, (w-z)|,\, |\Im\mathfrak{m}\, (w-z)| \leq 2N\delta\}
\]
a sufficiently large square centered at $z$ and write the discrete Green formula:
\begin{align*}
\sum\limits_{j=\lfloor {z}/{2\delta} \rfloor}^{N} \Delta F^1_{[\C_\delta,a]}(z-2j\delta)\  = &\sum\limits_{w\in \mathcal{S}_N(z)} \Delta F^1_{[\C_\delta,a]}(w)
\\ = & \sum\limits_{(ww')\in \partial\mathcal{S}_N(z)} (F^1_{[\C_\delta,a]}(w')-F^1_{[\C_\delta,a]}(w)).
\end{align*}
Let $N$ be large enough so that $|w-a|\ge N\delta$ for all boundary edges of the square $\mathcal{S}_N(z)$.
Then, it immediately follows from the estimate (\ref{eq: hm-gradients}) that the last sum is $O(N\cdot \delta^{\frac{3}{2}}(N\delta)^{-\frac{3}{2}})=O(N^{-\frac{1}{2}})$. Passing to a limit as $N\to\infty$, we conclude that
\[
\Delta G_{[\C_\delta,a]}(z)=\delta\cdot\!\!\sum\limits_{j=\lfloor {z}/{2\delta} \rfloor}^{\infty} \Delta F^1_{[\C_\delta,a]}(z-2j\delta)=0.
\]
Thus, $G_{[\C_\delta,a]}$ is indeed discrete harmonic in $\X^+$. Moreover, since it vanishes on $L_a$, one can harmonically extend $G_{[\C_\delta,a]}$ to the double cover $\mathcal{V}^1_{[\C_\delta,a]}$ by symmetry.

Let $\nu(\delta)$ denote the value of $G_{[\C_\delta,a]}$ in a lattice approximation of the point $a+1$. Arguing as in the proof of Lemma~\ref{lem: hm-estimate-La+grad}(iii), we see that, uniformly on compact subsets of $\C\setminus L_a$,
\begin{equation}
\label{G-convergence}
(\nu(\delta))^{-1}G_{[\C_\delta,a]}(z) ~\underset{\delta\to0}{\longrightarrow}~ \Re\mathfrak{e}\,{\sqrt{z-a}}
\end{equation}
(since $G_{[\C_\delta,a]}$ vanishes everywhere on $L_a$ and remains bounded near $0$ by the maximum principle, in this case the limiting positive harmonic function should be proportional to $\Re\mathfrak{e}\,{\sqrt{z-a}}$, and the multiplicative normalization is fixed at $a+1$). Moreover, the similar convergence holds true for discrete derivatives, yielding
\[
(\nu(\delta))^{-1}\cdot{\textstyle\frac{1}{2}}F^1_{[\C_\delta,a]}(z) ~\underset{\delta\to0}{\longrightarrow}~ \partial_x\Re\mathfrak{e}\,{\sqrt{z-a}}={\textstyle\frac{1}{2}}\,\Re\mathfrak{e}\,\left[{1}/{\sqrt{z-a}}\,\right].
\]
In particular, $\nu(\delta)\sim\vartheta(\delta)$ as $\delta\to 0$ which allows us to give a simple proof of the lower bound in (\ref{eq: vartheta-estimate}): as discrete harmonic functions $(\nu\left(\delta\right))^{-1}G_{[\C_\delta,a]}$ are uniformly bounded near the unit circle around $a$ and vanish identically on $L_a$, discrete Beurling estimate(\ref{eq:beurling-Lambda}) implies
\[
(\nu\left(\delta\right))^{-1}\cdot{\delta}=(\nu\left(\delta\right))^{-1}{G_{[\C_\delta,a]}\left(\ato\!+\delta\right)}\le C\sqrt\delta.
\]
Finally, the convergence (\ref{G-convergence}) near $L_a$ follows from the convergence of $F^1_{[\C_\delta,a]}(z)$, since the tails in (\ref{eq: def_r}) are uniformly small due to the estimate (\ref{eq: hm-near-La}).
\end{proof}

\subsection{\label{sub: riemann-problem}The boundary value problem for spinors}
In this section we reformulate the Riemann-type boundary value problem for holomorphic spinors (both discrete and continuous) using primitives (antiderivatives) of their squares. Note that this approach is not completely straightforward, since the square of a discrete holomorphic function, in general, does not have discrete primitive. However, it was noted in~\cite{smirnov-i} that one can naturally define the real part of this primitive, using the \emph{s-holomorphicity} of observables. Moreover, a technique developed in \cite{chelkak-smirnov-ii} (see, in particular, sections 3.4 and 3.5 therein) allows one to treat this real part essentially if it were a harmonic function. Below, we summarize the tools we will use. We warn the reader that all our definitions are equivalent to those of \cite{chelkak-smirnov-ii} \emph{after the multiplication of the spinor by $\sqrt{i}$}, which means that imaginary part, sub-/super-harmonicity and positivity of functions and their (inner) normal derivatives used in \cite{chelkak-smirnov-ii}
should be replaced by real part, super-/sub-harmonicity and negativity, respectively.

\subsubsection{\label{sub:discrete-integration}Discrete integration of the squared spinor observables} Let $\Delta_{\delta}^{\circ}$ be the standard (unnormalized) discrete Laplacian
acting on functions $H_{\delta}^\circ:\mathcal{V}_{\Omega_{\delta}}^{\circ}\to\mathbb{R}$ (which are defined on faces of $\Omega_\delta$):
\[
\Delta_{\delta}^{\circ}H_{\delta}^\circ\left(z\right)~:=~\sum_{w\sim z}\left(H_{\delta}^\circ\left(w\right)-H_{\delta}^\circ\left(z\right)\right),\quad z\in\mathrm{Int}\mathcal{V}_{\Omega_{\delta}}^{\circ},
\]
where the sum is over the four neighbors $w\in{\mathcal{V}}_{\Omega_{\delta}}^{\circ}$
of $z$. Similarly, for functions $H_{\delta}^\bullet:{\mathcal{V}}_{\Omega_{\delta}}^{\bullet}\to\mathbb{R}$ defined on vertices of $\Omega_\delta$, let $\Delta_{\delta}^{\bullet}$ denote the slightly modified discrete Laplacian:
\[
\Delta_{\delta}^{\bullet}H_{\delta}^\bullet\left(z\right) ~:=~ \sum_{w\sim z}c_{zw}\cdot\left(H_{\delta}^\bullet\left(w\right)-H_{\delta}^\bullet\left(z\right)\right),\quad z\in\mathrm{Int}\mathcal{V}_{\Omega_{\delta}}^{\bullet},
\]
where the conductance $c_{zw}$ is equal to $1$ for inner edges (i.e., for $w\in\mathrm{Int}\mathcal{V}_{\Omega_{\delta}}^{\bullet}$) and $c_{zw}={2}({\sqrt{2}-1})$ for the boundary edges (see Section~3.6 in \cite{chelkak-smirnov-ii} or \cite{duminil-copin-hongler-nolin} for the reason of this ``boundary modification'' of $\Delta_{\delta}^{\bullet}$).
\begin{prop}
\label{prop:int-square-spinor} For an s-holomorphic spinor observable $F_\delta=F_{\left[\Omega_\delta,\Ak\right]}:[\Omega_\delta,a_1^\rightarrow,\dots,a_n]\to\C$ satisfying boundary conditions (\ref{eq: discrete_bc}), one can define a real-valued function $H_\delta=H_{\left[\Omega_{\delta},\Ak\right]}:{\mathcal{V}}_{\Omega_{\delta}}^{\bullet\circ}\to\mathbb{R}$ which is a discrete analogue of the primitive $\Re\mathfrak{e}\int(F_\delta(z))^{2}dz$, so that the following properties are fulfilled:
\begin{itemize}
\item for any adjacent $w\in{\mathcal{V}}_{\Omega_{\delta}}^{\circ}$ and $v\in\mathcal{V}_{\Omega_{\delta}}^{\bullet}$, one has
\begin{equation}
\label{eq:H-def}
H_\delta^\circ\left(w\right)-H_\delta^\bullet\left(v\right)  =  2\delta\left|F_\delta\left(\tfrac{1}{2}(w+v)\right)\right|^{2},
\end{equation}
where $\tfrac{1}{2}(w+v)\in\mathcal{V}_{\Omega_{\delta}}^{\mathrm{c}}$ is the corner between $v$ and $w$ (in accordance with Lemma~\ref{lem: sing_f_Omega}, we set $\left|F_\delta\left(\Ato\right)\right|:=1$ in the case $w=a_1$ and $v=a_1+\delta$);
\item $H_\delta$ satisfies Dirichlet boundary conditions:
$H_\delta^\circ\left(v\right)=0$ for any $w\in\partial\mathcal{V}_{\Omega_{\delta}}^{\circ}$, and
$H_\delta^\bullet\left(v\right)=0$ for any $v\in\partial\mathcal{V}_{\Omega_{\delta}}^{\bullet}$;
\item $H_\delta^\bullet$ has a ``nonpositive inner normal derivative'', i.e. $H_\delta\left(v\right)\leq 0$ for any vertex $v\in\mathcal{V}_{\Omega_{\delta}}^{\bullet}$ adjacent to a boundary vertex;
\item $H_\delta^{\circ}$ is $\Delta_{\delta}^{\circ}$-subharmonic
on $\mathcal{V}_{\Omega_{\delta}}^{\circ}\setminus\left\{ \Ak\right\}$, while
$H_\delta^\bullet$ is $\Delta_{\delta}^{\bullet}$-superharmonic on $\mathcal{V}_{\Omega_{\delta}}^{\bullet}\setminus\left\{a_1+\delta\right\}$.
\end{itemize}
\end{prop}
\begin{proof}
All the claims follow directly from the results of Section~3.3 in \cite{chelkak-smirnov-ii}. Since all listed properties are local, the spinor nature of $F_\delta$ plays no role here (note that the right-hand side of (\ref{eq:H-def}) does not depend on the sheet).
\end{proof}

\begin{rem}
\label{rem:maximum}
By construction, $H^{\circ}_\delta(w)\geq H^{\bullet}_\delta(v)$ for adjacent $w$ and $v$. Combined with sub-/super-harmonicity, this implies a maximum principle for $H_\delta$: if $\Omega'_\delta\subset \Omega_\delta$ does not contain $a_1+\delta$ (respectively, any of $\Ak$), then \[
\min\limits_{\Omega'_\delta}H_\delta=\min\limits_{\partial\Omega'_\delta}H_\delta^{\bullet}\quad \text{(respectively,}\; \max\limits_{\Omega'_\delta}H_\delta=\max\limits_{\partial\Omega'_\delta}H_\delta^\circ).
\]
Moreover, if $\mathrm{hm}_A(z)$ denotes the discrete harmonic measure of a set $A$ in $\Omega'_\delta$ viewed from $z$, then
\begin{eqnarray*}
H_\delta(z)\geq (1-\mathrm{hm}^\bullet_A(z))\min\limits_{\partial\Omega'_\delta}H_\delta^{\bullet}+ \mathrm{hm}^\bullet_A(z)\min\limits_{A}H_\delta^{\bullet};\\
H_\delta(z)\leq (1-\mathrm{hm}^\circ_A(z))\max\limits_{\partial\Omega'_\delta}H_\delta^{\circ}+\mathrm{hm}^\circ_A(z)\max\limits_{A}H_\delta^{\circ}.
\end{eqnarray*}
\end{rem}
\begin{rem}
\label{rem:fix_singularity}
The subharmonicity of $H^{\circ}_\delta$ fails at $a_1,\dots,a_n$ because $F_{[\Omega_\delta,a_1\dots a_n]}$ branches at those points, while the superharmonicity of $H^{\bullet}_\delta$ fails at $a_1+\delta$ because of the discrete singularity of $F_{[\Omega_\delta,a_1\dots a_n]}$ which is not defined at $\Ato$. Due to Lemma~\ref{lem: remove_sing}, one can remove this singularity subtracting the full-plane observable $F_{[\C_\delta,a_1]}$. Then, the function
\[
\textstyle H_{\delta}^\x=H_{[\Omega_\delta,a_1\dots,a_n]}^\x:= \Re\mathfrak{e}\int(F_\delta^\x(z))^2dz,\quad F_\delta^\x:=F_{[\Omega_\delta,a_1\dots,a_n]}-F_{[\C_\delta,a_1]}
\]
(defined in the same way as $H_{\delta}$ accordingly to Proposition~\ref{prop:int-square-spinor}) is subharmonic on $\mathcal{V}^\circ_{\Omega_\delta}\setminus\{a_1\}$ and superharmonic on $\mathcal{V}^\bullet_{\Omega_\delta}$ everywhere in a neighborhood of $a_1$. Moreover, since $F_\delta^\x(\Ato)=0$ on both sheets, the values of $F_\delta^\x$ at the nearby corners $a_1\pm\frac{\delta}{2},a_1\pm \frac{i\delta}{2}$ and midedges $a_1\pm\frac{1\pm i}{2}\delta$ satisfy the \mbox{s-holomorphicity} conditions as if $F_\delta^\x$ were nonbranching at $a_1$. Thus, $H_{\delta}^\x$ is subharmonic at the point $a_1\in\mathcal{V}^\circ_{\Omega_\delta}$ too.
\end{rem}

\subsubsection{Integration of squared spinors in the continuous setup.}
We now give a characterization of the continuous spinors solving boundary value problem (\ref{def:spinor-bdry})\,--\,(\ref{def:spinor-1}) in terms of the primitives of their squares. In the next section, this characterization will be used in the proofs of main convergence results.

\begin{prop}
\label{prop: obs_uniqueness}
 Let $\Omega$ be a simply connected domain, and suppose a holomorphic spinor $f=f_{[\Omega,a_1,\dots,a_n]}$ solves the boundary value problem (\ref{def:spinor-bdry})\,--\,(\ref{def:spinor-1}) (or is defined according to Remark \ref{rem: sp_rough_domain}(ii), if $\Omega$ is not smooth). Define two harmonic functions
 \[
 \textstyle h:=\Re\mathfrak{e}\int(f(z))^2dz\quad\text{and}\quad h^\x:=\Re\mathfrak{e}\int(f(z)-f_{[\C,a_1]}(z))^2dz,
 \]
where $f_{[\C,a_1]}(z):=1/\sqrt{z-a_1}$. Then, the following holds true:
\begin{enumerate}
 \item $h$ is a single-valued function in $\Omega\backslash\{\Ak\}$, continuous up to $\partial\Omega$, which satisfies Dirichlet boundary conditions $h\equiv \text{const}$ on $\partial \Omega$ (since $h$ is defined up to an additive constant, below we assume that $h\equiv 0$ on $\partial \Omega$);
 \item there is no point $z_0$ on $\partial\Omega$ such that $h(z)\geq 0$ in a neighborhood of $z_0$;
 \item $h$ is bounded from below in a neighborhood of each $a_2,\dots,a_n$;
 \item $h^\x$ is single-valued and bounded in a neighborhood of $a_1$.
\end{enumerate}
Moreover, if $h$ and $h^\x$ satisfy (1)\,--\,(4), then $f$ solves the 
problem (\ref{def:spinor-bdry})\,--\,(\ref{def:spinor-1}).
\end{prop}
\begin{proof}
Note that, being integrated, the covariance property (\ref{eq: obs_covariance}) claims the conformal invariance of both $h$ and $h^\x$. As the  properties (1)\,--\,(4) are preserved under conformal mappings too, we will further assume that $\partial \Omega$ is smooth. Note that the property (\ref{def:spinor-bdry}) is equivalent to (1) and (2): it states that $f^2(z)\nu_\mathrm{out}(z)$ is nonnegative on the boundary, which yields $\partial_{\tau(z)}h \equiv 0$ (where $\tau(z)$ denotes a tangent vector) and $\partial_{\nu_\mathrm{out}(z)}h \geq 0$ everywhere on $\partial\Omega$. A straightforward integration of the asymptotics of $f$ near $a_k$ given by the conditions (\ref{def:spinor-2}) and (\ref{def:spinor-1}) yields
\begin{equation}
\begin{array}{rcll}
h^\x(z)& = & O(1), & z\rightarrow a_1, \\
h(z)& = &- c_k\log|z-a_k|+O(1), & z\rightarrow a_k, ~ 2\le k \le n.
\end{array}
\end{equation}
for some $c_k\ge 0$, giving (3) and (4). Vice versa, (3),(4) imply (\ref{def:spinor-1}), (\ref{def:spinor-2}) by differentiating and taking the square root.
\end{proof}

\subsection{Convergence of discrete observables away from singularities} \label{sub:convergence-away-sing} In this section we prove the convergence of (properly normalized) discrete spinor observables to their continuous counterparts on compact subsets of $\Omega\setminus\{\Ak\}$.

\begin{proof}[\textbf{Proof of Theorem~\ref{thm:obs-away}}]
Let the discrete integrals
\[
H_{\delta}:=\Re\mathfrak{e}\int\left((\vartheta\left(\delta\right))^{-1}F_\delta(z)\right)^{2}dz,\quad F_\delta=F_{\left[\Omega_\delta,\Ak\right]},
\]
be defined on $\mathcal{V}_{\Omega_{\delta}}^{\bullet\circ}$ accordingly to Proposition \ref{prop:int-square-spinor}. Given $\epsilon>0$, denote
\[
\Omega_\delta(\epsilon):=\Omega_\delta\cap \{z:\mathrm{dist}(z;\{\Ak\})>\epsilon\}.
\]

 Assume for a moment that,
\begin{equation}
\label{assump-bdd-away}
\begin{array}{l}
\text{for~any}~\epsilon>0,~\text{the~functions}~H_{\delta}~\text{remain~uniformly} \\
\text{bounded~on}~\Omega_\delta(\epsilon)~\text{by~some~constant}~C(\epsilon)~\text{as}~\delta\to 0.
 \end{array}
\end{equation}
 Then by \cite[Theorem~3.12]{chelkak-smirnov-ii}, the functions $(\vartheta(\delta))^{-1}F_\delta$ are equicontinuous on $\Omega_\delta(\epsilon)$ for any $\epsilon>0$. Therefore, by passing to a subsequence and applying the diagonal process, we can assume that $(\vartheta\left(\delta\right))^{-1}F_\delta$ tends to a limit $\tilde{f}$ and $H_{\delta}\rightarrow \tilde{h}:=\Re\mathfrak{e}\int \tilde{f}^2$ uniformly on compact subsets of $\Omega\backslash\{\Ak\}$. Our goal is to check that $\tilde{f}$ satisfies the properties (1)\,--\,(4) given in Proposition \ref{prop: obs_uniqueness}. Then, the  uniqueness of a solution to the boundary value problem (\ref{def:spinor-bdry})\,--\,(\ref{def:spinor-1}) proven in Lemma~\ref{lem: spinor_rough} implies $\tilde{f}=f_{[\Omega,\Ak]}$.

Clearly, $\tilde{f}$ is a holomorphic spinor on $[\Omega,\Ak]$. By superharmonicity of $H^\bullet_\delta$, for $2\le k\le n$, we have
\[
\min\limits_{|z-a_k|\le \epsilon}H_\delta(z)\ \ge \min\limits_{\epsilon<|z-a_k|\le 2\epsilon}H^\bullet_\delta(z),
\]
 thus $\tilde{h}$ is bounded from below in the neighborhoods of $a_2,\dots,a_n$, so the property (3) holds true. By the maximum principle for $H_{\delta}$ (see Remark \ref{rem:maximum}), taking into account that $H_\delta\equiv0$ on $\partial\Omega_\delta$, we have that
\begin{equation}
\label{eq: F_boundary}
\begin{array}{rcl}
 H^{\bullet}_{\delta}(z) & \geq & -C(\epsilon)(1-\mathrm{hm}^{\Omega_\delta(\varepsilon)}_{\partial \Omega_\delta}(z))\\[3pt]
 H^{\circ}_{\delta}(z) & \leq & C(\epsilon)(1-\mathrm{hm}^{\Omega_\delta(\varepsilon)}_{\partial \Omega_\delta}(z))
 \end{array}
\end{equation}
Since $\mathrm{hm}^{\Omega_\delta(\varepsilon)}_{\partial \Omega_\delta}(z)\to 1$ uniformly in $\delta$ as $z$ approaches the boundary of $\Omega_\delta$, this implies $\tilde{h}\equiv 0$ on $\partial\Omega$, giving (1). Moreover, thanks to Remark~6.3 in~\cite{chelkak-smirnov-ii}, we also have (2): there is no point on $\partial \Omega$  such that $\tilde{h}\geq 0$ in its neighborhood.

Consider now the function discussed (up to normalization) in Remark~\ref{rem:fix_singularity}:
\[
H_\delta^\x:=\Re\mathfrak{e}\int ((\vartheta\left(\delta\right))^{-1}(F_\delta(z)-F_{[\C_\delta,a_1]}(z)))^2dz,
\]
which is well-defined in the disc $\{z:|z-a_1|<r\}$ provided that $r$ is small enough. Since $(\vartheta(\delta))^{-1}F_\delta$ and $(\vartheta(\delta))^{-1}F_{[\C_\delta,a_1]}$ converge to $\tilde{f}$ and $f_{[\C,a_1]}$, respectively (see Lemma~\ref{lem: def_f_C}), uniformly on compact subsets of $A_r:=\{z:0<|z-a_1|<r\}$, we conclude that $H_\delta^\x$ converges to $\tilde{h}^\x:=\Re\mathfrak{e}\int(\tilde{f}(z)-f_{[\C,a_1]}(z))^2dz$ everywhere in $A_r$. By Remark~\ref{rem:fix_singularity} (sub-/super-harmonicity of $H_\delta^\x$ on $\mathcal{V}^\circ/\mathcal{V}^\bullet$ near the point $a_1$), the functions $H_\delta^\x$ are uniformly bounded in $A_r$, so $\tilde{h}^\x$ is bounded in $A_r$ too, which concludes the proof of the property (4). Therefore, $\tilde{f}=f_{\left[\Omega,\Ak\right]}$.

\smallskip

It remains to justify (\ref{assump-bdd-away}). On the contrary, suppose that
\[
M_\delta(\epsilon):=\max\limits_{\Omega_\delta(\epsilon)}|H_{\delta}|\underset{\delta\rightarrow 0}{\longrightarrow}\infty
\]
for some $\epsilon>0$ and along some subsequence of $\delta$'s. Then the re-normalized functions $(M_\delta(\epsilon))^{-1}H_{\delta}$ are uniformly bounded in $\Omega_\delta(\varepsilon)$, and thus $(M_\delta(\epsilon))^{-1/2}\cdot(\vartheta(\delta))^{-1}F_\delta$ and $(M_\delta(\epsilon))^{-2}H_{\delta}$ have (taking subsequences) the limits $\tilde{f}$ and $\tilde{h}=\Re\mathfrak{e}\int\tilde{f}^2$ which are holomorphic and harmonic in $\Omega_\delta(\epsilon)$, respectively. An important observation, proven in Lemma~\ref{non-vanish} below, is that $\tilde{h}$ cannot be identically zero. In particular, for any $0<\epsilon'<\epsilon$, we have $M_\delta(\epsilon')\le CM_\delta(\epsilon)$ with some $C=C(\epsilon',\epsilon)$ independent of $\delta$.

Applying the diagonal procedure, we may assume that $(M_\delta(\epsilon))^{-1/2}\cdot(\vartheta(\delta))^{-1}F_\delta$ tends to a limit $\tilde{f}$ (and $(M_\delta(\epsilon))^{-1}H_{\delta}$ tends to $\tilde{h}=\Re\mathfrak{e}\int \tilde{f}^2$) uniformly on \emph{each} of $\Omega_\delta(\epsilon')$.
Arguing as above, we see that $\tilde{h}$ is harmonic in $\Omega\backslash \{\Ak\}$, satisfies Dirichlet boundary conditions, has nonnegative outer normal derivative, and is bounded from below near $\Ak$. Moreover, repeating the last step of the proof given above we see that the function
\begin{align*}
  \tilde{h}^\x=&\lim\limits_{\delta\rightarrow 0}\, (M_\delta(\epsilon))^{-1}\Re\mathfrak{e}\int ((\vartheta(\delta))^{-1}(F_\delta(z)-F_{[\C_{\delta},a_1]}(z)))^2dz \cr
    = &\lim\limits_{\delta\rightarrow 0}\, (M_\delta(\epsilon))^{-1}\Re\mathfrak{e}\int ((\vartheta(\delta))^{-1}F_\delta(z))^2dz=\tilde{h}
\end{align*}
is also bounded in a neighborhood of $a_1$ (one can neglect $F_{[\C_{\delta},a_1]}(z)$ in the last expression since $(\vartheta(\delta))^{-1}F_{[\C_{\delta},a_1]}(z)$ tends to $f_{[\C,a_1]}$ and $M_\delta(\epsilon)\to\infty$). Thus, $\tilde{h}$ is bounded from below near \emph{all} $\Ak$ and has nonnegative outer normal derivative which contradicts to the maximum principle, if it does not vanish identically.
\end{proof}

\begin{lem}
 \label{non-vanish}
In the notation of the proof above, none of the subsequential limits of $(M_\delta(\epsilon))^{-1}H_{\delta}$ is identically zero in $\Omega(\epsilon)$.
\end{lem}
\begin{proof}
Suppose by contradiction that $(M_\delta(\epsilon))^{-1}H_\delta\rightarrow 0$ uniformly on compact subsets of $\Omega_\delta(\epsilon)$. Let $z^\mathrm{max}_\delta$ be the point of $\Omega_\delta(\epsilon)$ where the maximum of $|H_\delta|$ is attained. Since $H_\delta$ vanishes on $\partial\Omega_\delta$, the sub-/super-harmonicity of $H_\delta$ implies that $z_\delta^\mathrm{max}$ belongs to one of the ``discrete circles''
 \[
 \varpi_k(\epsilon):=\{z:\epsilon\leq|z-a_k|\leq\epsilon+5\delta\}
 \]
 of radius $\epsilon$ around $a_k$. Passing to a subsequence, one can assume that $z_\delta^\mathrm{max}\rightarrow z^\mathrm{max}$.
 Recall that $H^\circ\ge H^\bullet$ at adjacent points (see (\ref{eq:H-def})). Hence, either $z_\delta^\mathrm{max}\in \mathcal{V}^\circ_{\Omega_\delta}$ and $M_\delta(\epsilon)=H_\delta^\circ(z_\delta^\mathrm{max})$, or $z_\delta^\mathrm{max}\in \mathcal{V}^\bullet_{\Omega_\delta}$ and $M_\delta(\epsilon)=-H_\delta^\bullet(z_\delta^\mathrm{max})$.


\smallskip

Suppose that $z^\mathrm{max}_\delta\in \varpi_k(\epsilon)$ for some $2\leq k\leq n$. Denote
\[
m_\delta(2\epsilon):=\min_{z:|z-a_k|\le 2\epsilon} {H}^\bullet_\delta(z).
\]
As ${H}^\bullet_\delta$ is superharmonic inside $\varpi_k(2\epsilon)$ and $(M_\delta(\epsilon))^{-1}H^\bullet_\delta$ tends to zero uniformly on $\varpi_k(2\epsilon)$ by our assumption, we have
\begin{equation}
\label{m/M->0}
(M_\delta(\epsilon))^{-1}\cdot m_\delta(2\epsilon) \underset{\delta\rightarrow 0}{\longrightarrow}0.
\end{equation}

Therefore, $z^\mathrm{max}_\delta\in \mathcal{V}^{\circ}_{\Omega_\delta}$. 
By subharmonicity of $H^\circ_\delta$, we can find a discrete nearest-neighbor path $\gamma^\circ:= \{z^\mathrm{max}_\delta=z_1\sim z_2\sim \dots\}\subset\mathcal{V}_{\Omega_\delta}^\circ$ with $M_\delta(\epsilon)\le\dots\le H^\circ_\delta(z_{j})\le H^\circ_\delta(z_{j+1})\le\dots$, which may only end up at $a_k$, where the subharmonicity fails. By \cite[Remark~3.10]{chelkak-smirnov-ii}, the functions $H^\bullet_\delta-m_\delta(2\epsilon)$ and $H^\circ_\delta-m_\delta(2\epsilon)$ are uniformly comparable at adjacent points inside $\varpi_{k}(2\epsilon)$. Taking into account (\ref{m/M->0}), this implies $H^\bullet_{\delta}(z)\ge cM_\delta(\epsilon)$ for some absolute constant $c>0$ and all $z\in\gamma^\bullet$, where $\gamma^\bullet$ is the set of vertices adjacent to $\gamma^\circ$. Further, the maximum principle gives
\[
H^\bullet_\delta(z)\geq cM_\delta(\epsilon)\mathrm{hm}_\gamma^\bullet(z)+(1-\mathrm{hm}_\gamma^\bullet(z))m_\delta(2\epsilon)~~\text{for}~z:|z-a_k|\leq 2\epsilon,
\]
where $\mathrm{hm}_\gamma^\bullet$ denotes the discrete harmonic measure of $\gamma$ in $\{z:|z-a_k|\leq 2\epsilon\}$. It follows from the discrete Beurling estimate that $\mathrm{hm}^\bullet_\gamma(z)\ge \frac{1}{2}$, if $z$ is chosen close enough (but at fixed distance that is independent of $\delta$) to $z^\mathrm{max}$ where the path $\gamma$ starts. For such $z$, one has $(M_\delta(\epsilon))^{-1}H^\bullet_\delta(z)\ge \frac{1}{2}c + o(1)$ as $\delta\to 0$, and hence the limit of $(M_\delta(\epsilon))^{-1}H_\delta(z)$ cannot be identically zero in $\Omega(\epsilon)$.

\smallskip

It remains to treat the case $z^\mathrm{max}_\delta\in \varpi_{1}(\epsilon)$. Consider the function
\[
\textstyle (M_\delta(\epsilon))^{-1}H_\delta^\x= (M_\delta(\epsilon))^{-1}\Re\mathfrak{e}\int((\vartheta(\delta))^{-1}(F_\delta(z)-F_{[\C_{\delta},a_1]}(z)))^2dz.
\]
Note that it tends to zero on $\varpi_{1}(2\epsilon)$, since one can neglect the term $F_{[\C_{\delta},a_1]}(z)$ (recall that
$(\vartheta(\delta))^{-1}F_{[\C_{\delta},a_1]}(z)\to f_{[\C,a_1]}$ and $M_\delta(\epsilon)\rightarrow \infty$). Therefore, by Remark~\ref{rem:fix_singularity} and the maximum principle, it also tends to zero in a neighborhood of $\varpi_{1}(\epsilon)$, so we consequently derive that each of the functions
\[
\frac{F_\delta- F_{[\C_{\delta},a_1]}}{(M_\delta(\epsilon))^{1/2}\vartheta(\delta)}\,, \quad  \frac{F_\delta}{(M_\delta(\epsilon))^{1/2}\vartheta(\delta)}\quad\text{and}\quad \frac{H_\delta}{M_\delta(\epsilon)}
\]
tends to zero uniformly on $\varpi_{1}(\epsilon)$. In particular, $1=(M_\delta(\epsilon))^{-1}|H_\delta(z^\mathrm{max}_\delta)|\rightarrow 0$, which is a contradiction.
\end{proof}

\subsection{Analysis near the singularities}\label{sub:convergence-near-sing}
We now pass to the most delicate part of our analysis: matching the second-order terms in the values $F_\delta(a_1+\tfrac{3\delta}{2})$ with the second coefficient in the expansion of the continuous spinor near $a_1$. For shortness, below we use the notation $a:=a_1$, $F_\delta=F_{[\Omega_\delta,a,a_2,\dots,a_n]}$ and $\mathcal{A}=\mathcal{A}_{\left[\Omega,a,a_2,\dots,a_n\right]}$.

\begin{proof}[\textbf{Proof of Theorem~\ref{thm:localization-near-a}}]
Let $\mathcal{R}$ denote the reflection with respect to the horizontal line $\left\{x:\Im\mathfrak{m}\,(x-a)=0\right\}$ and $\Lambda_\delta$ be a small neighborhood of $a$ in $\Omega_{\delta}\cap\mathcal{R}\left(\Omega_{\delta}\right)$. Recall the notation $L_{a}=\{x+a+\tfrac{3\delta}{2}:x<0\}$, and denote by $\Lambda_\delta^+\subset \mathcal{V}^1_{\Lambda_\delta}$ one of two sheets of $[\Lambda_\delta,a]\setminus L_a$ such that $F_\delta(a+\tfrac{3\delta}{2})>0$ on $\Lambda_\delta^+$. We define a \emph{real-valued} function $S_\delta:\Lambda_\delta^+\to\R$ by
\[
S_{\delta}:=(\vartheta\left(\delta\right))^{-1}({\textstyle\frac{1}{2}}(F^{\vphantom{(\mathcal{R})}}_\delta+F^{(\mathcal{R})}_\delta)- F_{[\C_\delta,a]}-2\Re\mathfrak{e}\,\mathcal{A}\cdot G_{[\C_\delta,a]}),
\]
where $F_\delta^{(\mathcal{R})}=F_{[\mathcal{R}(\Omega_\delta),a,\mathcal{R}(a_2),\dots, \mathcal{R}(a_n)]}$ and the functions $F_{[\C_\delta,a]}$, $G_{[\C_\delta,a]}$ were constructed in Section~\ref{sub:full-plane}. By symmetry, one has $F_\delta^{(\mathcal{R})}(a+\tfrac{3\delta}{2})=F_\delta^{\vphantom{\mathcal{R}}}(a+\tfrac{3\delta}{2})$. Thus,
\begin{equation}
\label{x-near-sing-1}
S_{\delta}(a+\tfrac{3\delta}{2})=(\vartheta\left(\delta\right))^{-1}\left(F_\delta(a+\tfrac{3\delta}{2})-1-2\Re\mathfrak{e}\,\mathcal{A}\cdot \delta\right)
\end{equation}
and our goal is to estimate this value. Note that $S_\delta$ vanishes on the cut $L_{a}$: both $F_{[\C_\delta,{a}]}$ and $G_{[\C_\delta,a]}$ vanish by construction, and $F^{(\mathcal{R})}_\delta=-F^{\vphantom{\mathcal{R}}}_\delta$ on $L_{a}$ due to the spinor property ($F_\delta$ changes the sign between opposite sides of $L_{a}$, since they belong to different sheets). It is clear that $S_\delta$ is discrete harmonic everywhere in ${\Lambda_\delta^+}$ except at the point $a+\frac{3\delta}{2}$ since all terms are discrete harmonic there. Moreover, due to Lemma~\ref{lem: remove_sing}, it is discrete harmonic at $a+\frac{3\delta}{2}$ also. Therefore, for any (small, but fixed) $r>0$, discrete Beurling estimate (\ref{eq: hm-beurling}) implies
\begin{equation}
\label{x-near-sing-2}
|S_{\delta}(a+\tfrac{3\delta}{2})|\le 
C{\delta^\frac{1}{2}}\cdot r^{-\frac{1}{2}}\max_{\varpi(r)}|S_\delta|,
\end{equation}
where $\varpi(r):=\{z:r\le|z-a|\le r+5\delta\}$ denotes the ``discrete circle'' of radius $r$ around $a$. Further, it follows from Theorem~\ref{thm:obs-away} and Lemmas~\ref{lem: def_f_C},~\ref{def_r_delta} that
\[
S_\delta  ~\underset{\delta\to0}{\longrightarrow}~
s:=\Re\mathfrak{e}\,[\tfrac{1}{2}(f+f^{(\mathcal{R})})- f_{[\C,a]}]-2\Re\mathfrak{e}\,\mathcal{A}\cdot g_{[\C,a]},
\]
uniformly on $\varpi(r)$, where $f=f_{[\Omega,\Ak]}$, $f^{(\mathcal{R})}=f_{[\mathcal{R}(\Omega),a,\mathcal{R}(a_2),\dots,\mathcal{R}(a_n)]}$, and $f_{[\C,a]}(z)= \Re\mathfrak{e}\,[1/\sqrt{z-a}\,]$ and $g_{[\C,a]}(z)= \Re\mathfrak{e}\,\sqrt{z-a}$. Recall that, by definition of the coefficient $\mathcal{A}$, one has
\[
f-f_{[\C,a]}=2\mathcal{A}\sqrt{z-a}+O(|z-a|^{3/2}),\quad z\to a.
\]
It is easy to check that $f^{(\mathcal{R})}(z)\equiv \overline{f(\mathcal{R}(z))}$ (since this spinor solves the corresponding boundary value problem), hence
\[
f^{(\mathcal{R})}-f_{[\C,a]}=2\overline{\mathcal{A}}\sqrt{z-a}+O(|z-a|^{3/2}),\quad z\to a.
\]
Thus, we arrive at $s(z)=O(|z-a|^{3/2})$ as $z\to a$, which means
\begin{equation}
\label{x-near-sing-3}
r^{-\frac{1}{2}}\max_{\varpi(r)}|S_\delta| ~\underset{\delta\to0}{\longrightarrow}~ r^{-\frac{1}{2}}\cdot O(r^{3/2})=O(r).
\end{equation}
Combining (\ref{x-near-sing-1})\,--\,(\ref{x-near-sing-3}), one concludes that, for any given $r>0$,
\[
|F_\delta(a+\tfrac{3\delta}{2})-1-2\Re\mathfrak{e}\,\mathcal{A}\cdot \delta |\le C\vartheta(\delta)\delta^\frac{1}{2}r,
\]
if $\delta$ is small enough. Since $\vartheta(\delta)=O(\delta^{\frac{1}{2}})$ by (\ref{eq: vartheta-estimate}), and $r$ can be chosen arbitrary small, this yields (\ref{eq:localization-near-a}). All estimates are uniform with respect to $\Ak$ at definite distance from the boundary and each other.
\end{proof}

To prove Theorem~\ref{thm:localization-near-b}, we need to analyze the discrete spinor $F_{\left[\Omega_\delta,a,b\right]}$ near the point $b$.
In contrast to Theorem~\ref{thm:localization-near-a}, where second-order information near $a$ was extracted, here we only need to match
the first-order coefficients. Note that the situation is slightly different from the first-order analysis near $a$, as $F_{\left[\Omega_\delta,a,b\right]}$
does not have an explicit discrete singularity, remaining s-holomorphic near the branching point $b$, while its limit blows up at $b$. Still, the strategy for the identification resembles the one used above and appeals to the symmetrization with respect to the horizontal line passing through $b$. Let $L_b':=\{x+b+\frac{\delta}{2}:x<0\}$.

\begin{proof}[\textbf{Proof of Theorem~\ref{thm:localization-near-b}}]
Let $\mathcal{R}$ denote the reflection with respect to the line $\left\{x:\Im\mathfrak{m}(x-b)=0\right\}$ and $\Lambda_\delta\subset\Omega_{\delta}\cap\mathcal{R}\left(\Omega_{\delta}\right)$ be a small neighborhood of $b$. Fix the sheet $\Lambda_\delta^+\subset\mathcal{V}^i_{\Lambda_\delta}$ of $[\Omega_\delta,a,b]\setminus L'_b$ so that $F_\delta(b+\frac{\delta}{2})=i\mathcal{B}_\delta$ with $\mathcal{B}_\delta>0$ (recall that the values of discrete spinors on $\mathcal{V}^i_{\C_\delta}$ are purely imaginary), and do the same for the continuous spinor: fix a sheet so that $\Im\mathfrak{m}\,f_{[\Omega,a,b]}(b+x)>0$ as $x\downarrow 0$ (see (\ref{eq: expansion-b})).

Let $W_\delta:=\mathrm{hm}_{\{b+\frac{\delta}{2}\}}(\,\cdot\,)$ denote the harmonic measure of the point $b+\frac{\delta}{2}$ in the slit discrete plane $\C_\delta^{i}\setminus L_b'$ and
\[
T_{\delta}:=(\vartheta\left(\delta\right))^{-1}(\tfrac{1}{2}(F^{\vphantom{(\mathcal{R})}}_\delta+F^{(\mathcal{R})}_\delta)- i\mathcal{B}_\delta\cdot W_\delta)~:\Lambda_\delta^+\to i\R,
\]
where $F_\delta=F_{[\Omega_\delta,a,b]}$ and $F^{(\mathcal{R})}_\delta=F_{[\mathcal{R}(\Omega),\mathcal{R}(a),b]}$. By symmetry, $F^{(\mathcal{R})}_\delta(b+\frac{\delta}{2})=i\mathcal{B}_\delta$ (if one fixes the sheet for $F_\delta^{(\mathcal{R})}$ by the same condition $\Im\mathfrak{m}\,F^{(\mathcal{R})}_\delta(b+\frac{\delta}{2})>0$).
We have to prove that $\mathcal{B}_\delta\to\mathcal{B}=\mathcal{B}_{[\Omega,a,b]}$ as $\delta\to 0$. Note that, by passing to a subsequence, one may assume that $\mathcal{B}_\delta\to\widetilde{\mathcal{B}}$ for some $\widetilde{\mathcal{B}}\in[0,+\infty]$.

 Suppose that $\widetilde{\mathcal{B}}$ is finite. Then, for any fixed (small) $r>0$, Theorem~\ref{thm:obs-away} and Lemma~\ref{lem: hm-estimate-La+grad}(iii) imply the uniform convergence
\begin{equation}
\label{x-near-sing-b1}
T_\delta(z) ~\underset{\delta\to0}{\longrightarrow}~ t(z):= i\cdot \Im\mathfrak{m}\,[\tfrac{1}{2}(f(z)+f^{(\mathcal{R})}(z))] - i\widetilde{B}\cdot\Re\mathfrak{e}\,[{1}/{\sqrt{z-b}}\,]
\end{equation}
on compact subsets of $\{z:0<|z-b|\le r\}$, where $f=f_{[\Omega,a,b]}$ and $f^{(\mathcal{R})}(z)=f_{[\mathcal{R}(\Omega),\mathcal{R}(a),b]}(z)= -\overline{f(\mathcal{R}(z))}$ due to our conventions about sheets. Note that
\begin{equation}
\label{x-near-sing-b2}
t(z)=i(\mathcal{B}-\widetilde{\mathcal{B}})\cdot\Re\mathfrak{e}\,[{1}/{\sqrt{z-b}}\,]+O(|z-b|^{1/2}),\quad z\to b.
\end{equation}

\smallskip

Clearly, the function $T_\delta$ is harmonic everywhere in $\Lambda_\delta^+$ except at the point $b+\frac{\delta}{2}$, and $T_\delta=0$ on $L'_b$: indeed, $F^{(\mathcal{R})}_\delta=-F^{\vphantom{(\mathcal{R})}}_\delta$ on $L_b'$ due to the spinor property of $F_\delta$ and symmetry reasons. Moreover, $T_\delta(b+\frac{\delta}{2})=0$ as $W_\delta(b+\frac{\delta}{2})=1$. Therefore, $\widetilde{\mathcal{B}}\ne\mathcal{B}$ would contradict the maximum principle for the discrete harmonic function $T_\delta$: in this case (\ref{x-near-sing-b1}) and (\ref{x-near-sing-b2}) imply that, for sufficiently small $\delta$, the values of $T_\delta$ near $b$ are bigger than those near the circle $\{z:|z-b|=r\}$.

The similar argument works, if $\widetilde{\mathcal{B}}=+\infty$: in this case one would have
\[
\mathcal{B}_\delta^{-1}T_\delta(z)~\underset{\delta\to0}{\longrightarrow}~ -i\cdot\Re\mathfrak{e}\,[{1}/{\sqrt{z-b}}\,]
\]
which contradicts to the maximum principle for $T_\delta$, just as before.
\end{proof}

\renewcommand{\thesection}{A}
\section{Appendix} \label{sec:appendix-continuous-computations}
The goal of this Appendix is to derive the explicit formulae (\ref{12ptFcts}) for the scaling limits of $n$-point spin correlations in the upper half-plane $\bH$. Recall that we have reduced this problem to the computation of the coefficients $2\mathcal{A}_{[\bH,\Ak]}$ in front of $\sqrt{z-a_1}$ in the asymptotic expansion (\ref{eq: expansion-a}) of the spinor $f_{[\bH,a_1,\dots,a_n]}$. The latter is defined to be the unique solution to the boundary value problem \mbox{(\ref{def:spinor-bdry})\,--\,(\ref{def:spinor-1})} with additional regularity condition $f=O(|z|^{-1})$ at infinity.
What remains to be proven is the equality (\ref{eq: A_expl}) which claims that those coefficients $\mathcal{A}_{[\bH,\Ak]}$ are logarithmic derivatives of the explicit quantities $\langle \sigma_{a_1}\dots\sigma_{a_n}\rangle_\Omega^{+}$ given by (\ref{12ptFcts}).

\begin{proof}[Proof of (\ref{eq: A_expl}) for $n\geq 3$.] Recall from the proof of Lemma \ref{lem: spinor_rough} that $f_{[\bH,a_1,\dots,a_n]}$ can be written in the form
\begin{equation}
\label{xAp:fQ=}
f_Q(z)=e^{\frac{\pi i}{4}}\cdot \frac{Q(z)}{\sqrt{(z-a_1)(z-\overline{a}_1)\dots (z-a_n)(z-\overline{a}_n)}},
\end{equation}
where $Q$ is a polynomial of degree $n-1$ with real coefficients. This polynomial is to be determined from the conditions
\begin{align*}
\Re\mathfrak{e}\,(\lim\limits_{z\rightarrow a_1}\sqrt{z-a_1}\cdot f_{[\Omega,a_1,\dots,a_n]}(z))&=1, 
\\
\Re\mathfrak{e}\,(\lim\limits_{z\rightarrow a_k}\sqrt{z-a_k}\cdot f_{[\Omega,a_1,\dots,a_n]}(z))&=0, \quad k=2,\dots,n, 
\end{align*}
which are equivalent to
\begin{equation} 
\label{eq: problem}
\frac{Q(a_k)}{\prod\limits_{j\neq k}\sqrt{(a_k-a_j)(a_k-\overline{a}_j)}}+\frac{Q(\overline{a}_k)}{\prod\limits_{j\neq k}\sqrt{(\overline{a}_k-a_j)(\overline{a}_k-\overline{a}_j)}}=2\Im\mathfrak{m}\,a_k\cdot\delta_{1k}, 
\end{equation}
where $k=1,\dots,n$ and $\delta_{1k}$ is the Kronecker delta. We can write $Q(z)$ in the form
\begin{equation}
\label{xAp:Q=}
Q(z)\,=\,\sum\limits_{j=1}^{n}\frac{q_j}{z-a_j}\,\cdot\,\prod\limits_{j=1}^n(z-a_j),
\end{equation}
with some unknown coefficients $q_j$'s. Then, (\ref{eq: problem}) is equivalent to
\[
{q_k}\cdot\frac{\prod_{j\ne k}\sqrt{a_k-a_j}}{\prod_{j\ne k}\sqrt{a_k-\overline{a}_j}}+ \sum_{s=1}^n q_s\cdot\frac{(\overline{a}_k-a_k)\prod_{j\ne k}\sqrt{\overline{a}_k-a_j}}{(\overline{a}_k-a_s)\prod_{j\ne k}\sqrt{\overline{a}_k-\overline{a}_j}}=2\Im\mathfrak{m}\,a_k\cdot\delta_{1k},
\]
which can be further rewritten as
\begin{equation}
\label{eq: linear}
\sum_{s=1}^n (D_{ks}+C_{ks})q_s= c(a_1,\dots,a_k)\cdot\delta_{1k},\quad k=1,\dots,n,
\end{equation}
for a constant $c(a_1,\dots,a_k)\neq 0$ which only affects the overall normalization of the spinor $f_{[\bH,\Ak]}$, where $D$ is the diagonal matrix with entries
\[
D_{kk}=\frac{1}{\overline{a}_k-a_k}\prod\limits_{j\neq k}\chi^{\frac12}_{kj}\,,\qquad \chi_{kj}=\chi_{jk}=\frac{(a_k-a_j)(\overline{a}_k-\overline{a}_j)}{(a_k-\overline{a}_j)(\overline{a}_k-a_j)}=\left|\frac{a_k-a_j}{a_k-\overline{a}_j}\right|^2,
\]
and $C$ is the Cauchy matrix
\[
C_{ks}=\frac{1}{\overline{a}_k-a_s}\,,\quad 1\le k,s\le n.
\]
Solving (\ref{eq: linear}) by Cramer's rule, one gets $q_j/q_1=(-1)^{j+1}\det A_{[1j]}/\det A_{[11]}$, where $\det A_{[1j]}$ stands for minors of the matrix $A=D+C$. Plugging (\ref{xAp:fQ=}), (\ref{xAp:Q=}) into the expansion (\ref{eq: expansion-a}) of $f_{[\bH,\Ak]}(z)=f_Q(z)$ as $z\to a_1$, we arrive at the formula
\begin{align*}
 2\mathcal{A}_{[\bH,a_1,\dots,a_n]} &= \partial_z(\log [\,f_Q(z)\sqrt{z-a_1}\,])\big|_{z=a_1}
\cr &=\sum\limits_{j=2}^{n}\frac{\,\,(-1)^{j+1}\det A_{[1j]}}{(a_1-a_j)\det A_{[11]}}-\frac{1}{2(a_1-\overline{a}_1)}+\frac{1}{2}\sum_{j=2}^n\left(\frac{1}{a_1-a_j}-\frac{1}{a_1-\overline{a}_j}\right).
\end{align*}
Recall that principal minors of the Cauchy matrix $C$ are given by
$$
\det(C_S)\ =\ \prod\limits_{k\in S}\frac{1}{\overline{a}_k-a_k}\prod\limits_{\substack{k<m\\ k,m\in S}}\chi_{km}\ =\ \prod\limits_{k\in S}\frac{1}{\overline{a}_k-a_k}\prod\limits_{\substack{k\neq m\\ k,s\in S}}\chi^{\frac12}_{km}.
$$
Hence, expanding  $\det A_{[11]}$ by linearity and encoding a subset $S\subset\{2,\dots,n\}$ by the collection of signs $\mu\in \{\pm 1\}^n$, where $\mu_1=-1$ and $\mu_s=2\mathbbm{1}_{s\in S}-1$ for $s=2,\dots,n$, we find
\begin{align*}
\det A_{[11]}\ & = \sum\limits_{S\subset\{2,\dots,n\}} \left(\det C_{S}\cdot \det D_{\overline{S}}\right)  \cr & =\
\prod_{k=2}^n\frac{1}{\overline{a}_k-a_k}\ \cdot\!\!\sum_{\substack{S\subset\{2,\dots,n\}}}
\biggl(\prod_{\substack{k,m\in {S}\\ k\ne m}} \chi_{km}^{\frac12} \cdot \prod_{k \in \overline{S}}\prod_{\substack{j=1 \\ j\neq k}}^n(\chi_{kj}\chi_{jk})^{\frac14}\biggr)\cr
& =\ \prod_{k=2}^n\frac{1}{\overline{a}_k-a_k}\ \cdot \sum_{\substack{\mu\in\{\pm 1\}^n\\\mu_1=-1}}
\biggl(
\prod_{\substack{k,m=2\\k \neq m}}^n\chi^{\frac{3+\mu_k\mu_m}8}_{km} \prod_{\substack{k=2}}^n(\chi_{k1}\chi_{1k})^{\frac{1-\mu_k}{8}}
\biggr)\cr
&=\ \prod_{k=2}^n\frac{1}{\overline{a}_k-a_k}
\prod_{\substack{k,m=2\\k \neq m}}^n\chi^{\frac38}_{km}
\prod_{\substack{k=2}}^n\chi_{1k}^{\frac{1}{4}}\ \cdot
\sum_{\substack{\mu\in\{\pm 1\}^n\\\mu_1=-1}}
\prod_{\substack{k,m=1\\k\neq m}}^n\chi_{km}^{\frac{\mu_k\mu_m}{8}}.
\end{align*}
Note that we can also write
$$
\sum_{j=2}^{n}\frac{(-1)^{j+1}\det A_{[1j]}}{a_1-a_j} = \det[\widetilde{D}+\widetilde{C}],
$$
where $\widetilde{D}$ is a diagonal matrix such that $\widetilde{D}_{11}=0$ and $\widetilde{D}_{kk}=D_{kk}$ for $2\le k\le n$; $\widetilde{C}_{11}=0$, $\widetilde{C}_{1s}=1/(a_1-a_s)$ for $2\le s\le n$ and $\widetilde{C}_{ks}=C_{ks}=1/(\overline{a}_k-a_s)$, if $2\le k\le n$.

We now compute the principal minors of $\widetilde{C}$. If $1\notin S$, then $\det \widetilde{C}_S= \det C_S$.  Otherwise, denote $S:=\{1\}\cup S_1$, where $S_1\subset\{2,\dots,n\}$. Treating $\overline{a}_1$ and $a_1$ as independent variables (note that $\widetilde{C}$ does not contain $\overline{a}_1$), we obtain
\begin{multline*}
\det \widetilde{C}_S=
\lim_{\overline{a}_1\to a_1} \left(\det C_S - \frac{1}{\overline{a}_1-a_1} \det C_{S_1}\right)
\\
=\prod_{k\in S_1}\frac{1}{\overline{a}_k-a_k}\!\prod_{\substack{k<m\\k,m\in S_1}}\!\chi_{km}\cdot \lim_{\overline{a}_1\to a_1} \frac{\prod_{s\in S_1}\chi_{1s}\,-\,1}{\overline{a}_1-a_1}
\\
=\prod_{k\in S_1}\frac{1}{\overline{a}_k-a_k}\!\prod_{\substack{k\ne m \\ k,m\in S_1}}\!\chi^{\frac12}_{km}\,\cdot \sum_{s\in S_1}\biggl(\frac{1}{a_1-\overline{a}_s}-\frac{1}{a_1-a_s}\biggr).
\end{multline*}
Note that ${1}/{(a_1-\overline{a}_s)}-{1}/{(a_1-a_s)}=-\partial_{a_1}\log\chi_{1s}$.
Similarly to the computation of $\det A_{[11]}$ given above, we have
\begin{align*}
& \det[\widetilde{D}+\widetilde{C}]\ = \sum_{S\subset\{1,\dots,n\}}\left(\det \widetilde{C}_{S}\cdot \det \widetilde{D}_{\overline{S}}\right)\ = \sum_{\substack{S=\{1\}\cup S_1,\\ S_1\subset\{2,\dots,n\}}}\left(\det \widetilde{C}_{S}\cdot \det \widetilde{D}_{\overline{S}}\right)  \\
& = \prod_{k=2}^n\frac{1}{\overline{a}_k\!-\!a_k}
\prod_{\substack{k,m=2\\k \neq m}}^n\chi^{\frac38}_{km}
\prod_{\substack{k=2}}^n\chi_{1k}^{\frac{1}{4}}\,\,\cdot\!\!\!\!
\sum_{\substack{\mu\in\{\pm 1\}^n\\\mu_1=-1}}\biggl(
\prod_{\substack{k,m=1\\k\neq m}}^n\chi_{km}^{\frac{\mu_k\mu_m}{8}}\!\cdot \sum_{\substack{s=2}}^n\frac{1\!+\!\mu_s}{2}\,(-\partial_{a_1}\log\chi_{1s})\biggr),
\end{align*}
where we set $\mu_s=2\mathbbm{1}_{s\in S_1}-1$ for $s=2,\dots,n$. Taking into account the symmetry of $\mu_k\mu_m$ and $\mu_1\mu_s$ under the simultaneous sign flip of \emph{all} $\mu_j$, $j=1,\dots,n$, we obtain
\[
\frac{\det[\widetilde{D}+\widetilde{C}]}{\det A_{[11]}}=\frac{\sum_{\MuPm}\biggl(
\prod_{1\le k<m\le n}\chi_{km}^{\frac{\mu_k\mu_m}{4}}\cdot \sum_{\substack{s=2}}^n\frac{1-\mu_1\mu_s}{2}\,(-\partial_{a_1}\log\chi_{1s})\biggr)}
{\sum_{\MuPm}\prod_{1\le k<m\le n}\chi_{km}^{\frac{\mu_k\mu_m}{4}}}.
\]
Therefore,
\[
\frac{\det[\widetilde{D}+\widetilde{C}]}{\det A_{[11]}}= -\frac{1}{2}\sum_{\substack{m=2}}^n\partial_{a_1}\log\chi_{1m}+
\frac{2\partial_{a_1}\biggl(\sum_{\MuPm}
\prod_{1\le k<m\le n}\chi_{km}^{\frac{\mu_k\mu_m}{4}}\biggr)}
{\sum_{\MuPm}\prod_{1\le k<m\le n}\chi_{km}^{\frac{\mu_k\mu_m}{4}}},
\]
and we arrive at
$$
\mathcal{A}_{[\bH,\Ak]}=-\frac{1}{4(a_1-\overline{a}_1)}+\partial_{a_1}\log\biggl(\ \sum_{\MuPm}\prod_{1\le k<m\le n}\chi_{km}^{\frac{\mu_k\mu_m}{4}}\biggr),
$$
thus concluding the proof of the identity $\mathcal{A}_{[\bH,\Ak]}=2\partial_{a_1}\log \langle\sigma_{a_1}\dots\sigma_{a_n}\rangle^+_{\bH}$.
\end{proof}
\begin{rem} When preparing the first version of our paper, we were only able to check the explicit formula (\ref{12ptFcts}) in some special cases; thus, the exposition went along the lines of Remark \ref{rem:implicit_integration_L}. It later turned out that ``interpolation problems'' similar to (\ref{eq: problem}) also appear in the analysis of the critical Ising model with mixed free/fixed boundary conditions, in which case the collection $a_1,\overline{a}_1,\dots,a_n,\overline{a}_n$ should be replaced by $2n$ boundary points where these boundary conditions change. In particular, a computation similar to one given above, combined with the Edwards-Sokal coupling, yields general $2n$-point crossing formulae in the random cluster (Fortuin-Kasteleyn) representation of the critical Ising model, see \cite{izyurov-freebc}.
\end{rem}

\end{document}